\documentclass{article}
\usepackage{verbatim}
\usepackage{hyperref}
\usepackage{times}
\usepackage[pdftex]{graphicx}
\usepackage{amssymb,amsfonts,amsmath,amsthm}
\usepackage{float}
\usepackage{xcolor}
\usepackage{upgreek}
\usepackage{enumitem}
\usepackage{bm}
\usepackage{bbm}
\usepackage{tikz}
\usetikzlibrary{arrows.meta}
\newcommand{\oW}{\overline W}
\newcommand{\uW}{\underline W}
\newcommand{\oV}{\overline{\mathcal V}}
\newcommand{\uV}{\underline{\mathcal V}}
\newcommand{\oZ}{\overline Z}

\usepackage{caption}
\usepackage{subcaption}

\usepackage{tikz}
\usetikzlibrary{arrows.meta}      
\usetikzlibrary{calc}             
\usetikzlibrary{patterns}   

\DeclareMathOperator*{\argmax}{arg\,max}
\DeclareMathOperator*{\argmin}{arg\,min}

\DeclareMathOperator{\Tr}{Tr}

\newtheorem{theorem}{Theorem}
\newtheorem{proposition}{Proposition}
\newtheorem{lemma}{Lemma}
\newtheorem{ass}{Assumption}

\newtheorem{definition}{Definition}
\newtheorem{remark}{Remark}

\newtheorem{example}{Example}

\newcommand\HH {\mathbb H}
\newcommand\EE {\mathbb E}
\newcommand\FF {\mathbb F}

\newcommand\RR {\mathbb R}
\newcommand\PP {\mathbb P}

\newcommand\bS {\mathbb S}

\newcommand\cA {\mathcal A}
\newcommand\cH {\mathcal H}

\newcommand\cY {\mathcal Y}
\newcommand\cZ {\mathcal Z}
\newcommand\cL {\mathcal L}
\newcommand\cC {\mathcal C}
\newcommand\cD {\mathcal D}
\newcommand\cE {\mathcal E}
\newcommand\cF {\mathcal F}

\newcommand\cV {\mathcal V}

\newcommand\cS {\mathcal S}

\def\pa{\partial}

\def\a{\alpha}

\newcommand\1 {\mathbf 1}
\def\qed{\hskip6pt\vrule height6pt width5pt depth1pt}

\title{Principal-agent problems with adverse selection:\\ A stochastic target problem formulation\thanks{The authors thank Nizar Touzi for fruitful conversations and Jakša Cvitanić for identifying an error in an earlier version of this manuscript, which led to the improved statement of Theorem~\ref{thm:HJB_principal_screens}.
}}
\author{
  Guillermo Alonso Alvarez\thanks{Department of Mathematics, University of Michigan. \texttt{guialv@umich.edu}.}
  \and
  Ibrahim Ekren\thanks{Department of Mathematics, University of Michigan. \texttt{iekren@umich.edu}. I. Ekren is partially supported by the NSF grant DMS-2406240.}
  \and
  Liwei Huang\thanks{Department of Mathematics, University of Michigan. \texttt{huanglw@umich.edu}.}
}
\date{May 2026}
\begin{document}
\maketitle
\begin{abstract}
We study a principal-agent problem with adverse selection, where the principal does not know the agent's true cost but must design a contract to optimize a specific criterion. Unlike standard screening frameworks that allow for self-selection, we assume the principal can only offer a unique contract. We show that the agent's optimization problem can be reformulated as a stochastic target problem. After characterizing the credible domain of this target problem, we show that the principal's objective can be solved as a stochastic optimal control problem with partial information and state constraints. The description of the credible domain also allows us to obtain the value of screening contracts. 
\end{abstract}
\section{Introduction}

A continuous-time principal–agent problem is a bilevel optimization: 
the principal designs a terminal-payment contract $\xi$, given as a 
measurable function of the output process $X$, and the agent chooses 
an action profile to maximize her expected utility under $\xi$. The 
principal's payoff depends on the agent's optimal response, so the 
bilevel structure rarely admits a tractable reduction. The seminal work of Holmstr\"om and Milgrom \cite{holmstrom1987aggregation} shows that, in a continuous time setting with exponential utility, optimal contracts are linear in the output process. Beyond this tractable setting, however, the problem becomes significantly more complex due to the infinite-dimensional nature of the contract space and the incentive compatibility constraints. We refer to \cite{cvitanic2012contract} for a textbook treatment 
of continuous-time contract theory.

A subsequent foundational contribution by Sannikov 
\cite{sannikov2008continuous}, extended in the BSDE/stochastic-control 
framework by Cvitani\'c, Possama\"i, and Touzi \cite{cvitanic2018dynamic}, 
resolves this difficulty by reparameterizing the contract space. 
Rather than optimizing directly over a terminal payment $\xi$, the principal optimizes over a pair $(y_0, Z)$ consisting of an initial continuation utility $y_0$ for the agent and a process $Z$ encoding the sensitivity 
of the contract to $X$. The contract is then recovered as 
$\xi = Y^{y_0, Z}_T$, the terminal value of a forward stochastic 
differential equation starting at $y_0$ and controlled by $Z$. Under 
this reparametrization, the agent's optimal response is given pointwise 
by maximizing a Hamiltonian, and the principal's problem becomes a 
standard stochastic control problem on the joint state $(X, Y)$. This 
methodology has become a standard approach in the continuous-time contracting 
literature.

Sannikov's reparametrization presupposes that the principal knows 
the agent's cost function. In the present paper, we develop a new 
methodology for continuous-time contracting problems in which the agent's cost depends on a binary 
parameter $\Theta\in\{0,1\}$, observed by the agent at time zero but 
not by the principal, who knows only the prior distribution of 
$\Theta$. The principal commits at time zero to a single contract and 
cannot offer a menu. This contrasts with the classical \emph{screening} 
formulation of dynamic adverse-selection contracting 
\cite{cvitanic2013dynamics,santibanez2020bank}, in which the principal 
offers a menu of two contracts and the two types self-select at the 
acceptance stage, revealing the type to the principal. In our setting 
the type is never revealed directly, and the principal must update her 
beliefs about $\Theta$ from observations of $X$ alone. 

A natural reparametrization in this single-contract setting is to 
track, separately for each possible type $\theta\in\{0,1\}$, the 
agent's continuation utility process 
$Y^{\theta, y_\theta, Z^\theta}$ under her own type-specific 
Hamiltonian, with type-specific initial value $y_\theta$ and 
sensitivity process $Z^\theta$. Since the contract itself cannot 
depend on the type, the two type-indexed continuation utilities are 
forced to share the same terminal value:
\[
\frac{Y^{0, y_0, Z^0}_T}{\beta_0} \;=\; \frac{Y^{1, y_1, Z^1}_T}{\beta_1},\quad \mathbb P\text{-a.s.}
\]
This is a \emph{stochastic target} constraint in the sense of Soner 
and Touzi \cite{soner2002dynamic,soner2002stochastic,
bouchard2010stochastic,bouchard2010optimal}, and it induces a 
non-trivial set of pairs $(y_0, y_1)$ of initial continuation utilities 
that can be simultaneously delivered by a single contract: the 
\emph{credible set} $\mathcal E(X_0)$. Our approach is related to the recent stochastic-target reformulation of closed-loop Stackelberg games introduced in \cite{hernandez2024closed}, where the follower’s continuation value is used as a controlled state variable in the leader’s problem. However, their setting does not feature adverse selection which requires in our case to filter the private information and add the filter as an additional state variable. \cite{hernandez2024closed} considers a single follower with no privately known type, and the target constraint enforces the dynamic consistency of the follower’s continuation value.

Our first main result characterizes $\mathcal E(X_0)$ explicitly as a 
strip in the plane: the gap $y_0 - y_1$ lies between two boundary 
functions $\underline W(0, X_0)$ and $\overline W(0, X_0)$ that solve 
a pair of fully nonlinear PDEs. This turns the stochastic target 
constraint into a \emph{state constraint} on the trajectory of the 
gap process $Y^0 - Y^1$, in the spirit of the state-constrained 
optimal control literature \cite{soner1986optimalI,
soner1986optimalII,capuzzodolcettalions1990,ishii2002class,
flemingsoner2006}.

In our framework where the principal can only offer a unique contract, a complete reformulation of the principal's problem requires one 
further state variable: her conditional belief $p_t$ that $\Theta = 0$ 
given the path of $X$ up to time $t$. Because our state-constraint 
formulation keeps track of both type-indexed continuation utilities 
simultaneously, the belief dynamics arise naturally from the Wonham 
filter, driven by the innovation of $X$ and the dynamics of $(X,p)$ can be obtained under the information of the principal. Then, the principal's value 
function admits a representation as the supremum, over admissible 
pairs $(y_0, y_1) \in \mathcal E(X_0)$, of the value of a 
state-constrained stochastic optimal control problem in the augmented 
state $(X_t, Y^0_t, Y^1_t, p_t)$. In fact, the principal's problem can be 
stated as an optimal control problem with partial information and 
state constraints.

The state-constrained HJB equation arising from this formulation does 
not satisfy the standard inward-pointing condition of 
\cite{soner1986optimalI,ishii2002class}. To handle this, we adapt the 
methodology of Bouchard, \'Elie, and Imbert \cite{bouchard2010optimal} 
and show that, under additional structure (when the data are 
independent of $X$, the control set is bounded, and the principal is 
risk-neutral), the state-constrained problem reduces to a Dirichlet 
boundary value problem in which the principal receives an explicit 
boundary utility upon hitting the lateral boundary. The resulting 
equation falls within the standard viscosity-solution framework: we 
characterize the principal's value as the unique viscosity solution 
to a Hamilton--Jacobi--Bellman equation, which makes the problem 
amenable to numerical solution.

We then provide a verification theorem. When the state-constrained 
value function admits sufficient regularity, the optimal contract is 
given by a feedback rule in the augmented state 
$(X_t, Y^0_t, Y^1_t, p_t)$. Optimal contracts therefore depend on the 
output process, the continuation utility of each type, and the 
principal's belief about the type. The contract is non-trivial in a 
strong sense: it cannot be expressed as a function of $X$ alone, as in 
the moral-hazard-only case of \cite{cvitanic2018dynamic}, nor as a 
deterministic function of the principal's belief, as in pure filtering 
models, but jointly couples all three.

Our state-constrained stochastic-control formulation provides a common framework for comparing alternative participation requirements. Participation may be imposed conditionally or unconditionally.  In the conditional problem, the contract must satisfy a participation constraint for each possible type: after learning her type, the agent must obtain at least the corresponding type-specific reservation utility. In the unconditional problem, participation is imposed only ex ante, in expectation under the prior distribution over types. Conditional rationality is a stronger constraint, since every realized type must be willing to participate. Unconditional rationality is weaker since participation is required before the type is realized.

We compare these single-contract problems with the classical screening formulation of \cite{cvitanic2013dynamics}. In a screening model, the principal offers a menu of contracts, one for each type, and the agent self-selects. This choice can reveal the type immediately and eliminates the need for filtering. However, the menu must satisfy incentive-compatibility constraints: each type must prefer the contract intended for her over the contract intended for the other type. Our framework allows us to write the conditional single-contract problem, the unconditional single-contract problem, and the screening problem in a common stochastic-control framework. As shown formally in Theorem \ref{thm:HJB_principal_screens}, screening dominates the conditional single-contract problem, whereas the unconditional single-contract problem can dominate screening because it imposes a weaker participation requirement. Thus, screening is a useful benchmark, but it is not always the relevant feasible contracting environment. In some applications, the principal may not be able to offer a menu that reliably reveals the agent’s type at time zero. In that case, the type remains hidden and the principal must learn about it dynamically from the output process.

We also use numerical methods to study the comparison between the conditional single-contract, unconditional single-contract, and screening formulations. The corresponding solution of the state-constrained HJB equation is approximated using the deep Galerkin method of \cite{sirignano_spiliopoulos}. We consider both dominated and non-dominated type structures. In the dominated case, the two types can be ranked uniformly: one type has a lower cost of effort, and therefore a higher Hamiltonian for all relevant actions. In the non-dominated case, the cost functions are not uniformly ordered, so neither type is globally more efficient than the other. The numerical results confirm the ordering of the principal values: $V_{p,c} \leq V_{p,s}\leq V_{p,uc}$. They also show how the conditional, unconditional, and screening values vary with the prior belief $p$, and how this dependence differs between dominated and non-dominated structures. The simulations further illustrate how the optimal initial promised utilities for each agent's type $(y_0,y_1)$ adjust across formulations.

\paragraph{Outline.} 
Sections~\ref{subsec:prob_setup} and~\ref{section.PA.problem} set up the probabilistic framework and agent
 and principal's problems. Section~\ref{s:firstred} 
reduces the principal's problem to a stochastic target problem and 
proves the credible-set characterization. 
Section~\ref{s:statec} introduces the filter $p_t$ and reformulates 
the problem as a state-constrained optimal control problem with the 
corresponding HJB equation. 
Section~\ref{s:screening} recasts the screening problem of 
\cite{cvitanic2013dynamics} in our framework and develops the 
comparison between single-contract and screening formulations. 
Section~\ref{sec:dirichlet} carries out 
the Dirichlet reduction in the $X$-independent risk-neutral setting 
and proves the uniqueness of a constrained viscosity solution. 
Section~\ref{s:optcontract} establishes the verification theorem 
and discusses the structure of optimal contracts and Section~\ref{sec:benchmark} contains the numerical results.

\section{Probabilistic setup}\label{subsec:prob_setup}

Let $T>0$ be a fixed terminal time and $d\geq 1$ the dimension. We denote by $
\Omega_X := C([0,T];\RR^d)$
the canonical path space representing the agent's output process. 
We work on the product space
$
\Omega := \{0,1\} \times \Omega_X,$
equipped with its canonical coordinates $(\Theta,X)$, i.e.\ for $\omega=(\theta,x)\in\Omega$,
\[
\Theta(\omega)=\theta,\qquad X_t(\omega)=x_t,\ \ t\in[0,T],
\]
where $\Theta$ represents the agent's type (e.g., profitability of the project and/or cost to the agent).

We fix a continuous function
$$\sigma:[0,T]\times \RR^d\mapsto \RR^{d\times d},$$
such that $\sigma(t,x)$ is invertible for all $(t,x) \in [0,T]\times \RR^d$, 
and 
$$\sup_{t\in [0,T]}\sup_{x\neq y}\left\{\frac{|\sigma(t,x)-\sigma(t,y)|}{|x-y|}+|\sigma(t,x)|+|\sigma^{-1}(t,x)|\right\}<\infty.$$

We denote by ${\bar P}$ the probability\footnote{We take the convention that probability distributions with double-struck letters such as $\PP$ refer to measures on $\Omega$ whereas plain roman letters such as $P$ refer to probability measures on $\Omega_X$.} measure on $\Omega_X$ defined as the distribution of the solution of the following SDE
\begin{align}\label{sde:sigma}
    X_t=X_0+\int_0^t\sigma(s,X_s)dB^{{\bar P}}_s
\end{align}
where $B^{{\bar P}}$ is a ${\bar P}$ Brownian motion and $X_0\in \RR^d$ a fixed initial condition. Note that, in general, ${\bar P}$ may depend on $X_0$, but we omit this dependence to simplify the notation. Under our assumptions, both weak and strong uniqueness hold for \eqref{sde:sigma}.

We endow $\Omega$ with the filtration $\FF=(\cF_t)_{0\le t\le T}$ defined by
\[
\cF_t := \sigma(\Theta)\ \vee\ \cF_t^X,
\qquad
\cF_t^X := \sigma(X_s:\, s\le t)\vee \mathcal{N}({\bar P}),
\]
where $\mathcal{N}({\bar P})$ are the null sets of ${\bar P}$ (which is a measure on $\Omega_X$), $\FF^X$ represents the information available to the principal, and $\FF$ represents the information available to the agent.

We fix a probability measure\footnote{Extending our methodology to the case where $\theta$ can take more than two different values requires smooth solutions to a fully nonlinear anisotropic curvature-flow equation. Relaxing this smoothness assumption is left for future work.} $\mu\in Prob(\{0,1\})$
\[
p_0 := \mu\bigl(\{\Theta=0\}\bigr) \in(0,1), \mbox{ and }1-p_0:= \mu\bigl(\{\Theta=1\}\bigr)
\]
which is the distribution of $\Theta$.
We define the reference probability $\bar \PP$ on $\Omega$ by the product measure
\[
\bar \PP := \mu\otimes {\bar P},
\]
so that, under $\bar \PP$, $\Theta\sim\mu$, $X$ solves \eqref{sde:sigma} where $B^{{\bar P}}$ is a $\FF$-Brownian motion,
and $\Theta$ is independent of $X$.

A control process is an $\FF$-adapted process $\alpha=(\alpha_t)_{0\le t\le T}$ with values in a
compact subset $A$ of a finite-dimensional space, and we denote by $\mathcal{A}$ the set of all such controls. We also denote $\cA^X$ the set of $\FF^X$-adapted processes with values in $A$. Since $\Theta\in\{0,1\}$ and $\cF_t=\sigma(\Theta)\vee\cF_t^X$, the map
\begin{equation}\label{eq:control_bijection}
\alpha\ \longleftrightarrow\ (\alpha^0,\alpha^1),\qquad \alpha_t=\alpha^0_t\1_{\{\Theta=0\}}+\alpha^1_t\1_{\{\Theta=1\}}=\alpha^\Theta_t
\end{equation}
is a bijection between $\mathcal{A}$ and $\mathcal{A}^X\times\mathcal{A}^X$.

We fix a bounded measurable function
\[
\lambda:[0,T]\times \RR^d\times \{0,1\} \times A \to \RR^d
\]
such that $\alpha\mapsto\lambda(t,x,\theta,\alpha)$ is continuous, uniformly in $(t,x,\theta)\in[0,T]\times\RR^d\times\{0,1\}$.

\paragraph{Controlled measures via Girsanov's theorem.}
For $\theta\in\{0,1\}$ and $\alpha\in\mathcal{A}^X$, define the Dol\'eans--Dade exponential
\begin{equation}\label{eq:density}\notag
Z_T^{\alpha,\theta}\ :=\ \exp\!\left(\int_0^T\!\lambda(s,X_s,\theta,\alpha_s)^\top dB^{\bar P}_s-\tfrac12\int_0^T\!|\lambda(s,X_s,\theta,\alpha_s)|^2\,ds\right),
\end{equation}
which is a $(\FF^X,{\bar P})$-martingale by the boundedness of $\lambda$, and set
\begin{equation}\label{eq:Ptilde}\notag
\frac{dP^{\alpha,\theta}}{d{\bar P}}\bigg|_{\cF_T^X}:=Z_T^{\alpha,\theta}.
\end{equation}
By Girsanov's theorem and weak uniqueness for \eqref{sde:sigma}, $P^{\alpha,\theta}$ is the unique element of $Prob(\Omega_X)$ absolutely continuous with respect to ${\bar P}$ under which
\begin{equation}\label{eq:X_dynamics_small}
X_t = X_0 + \int_0^t \sigma(s,X_s)\Bigl(\lambda(s,X_s,\theta,\alpha_s)\,ds + dB^{P^{\alpha,\theta}}_s\Bigr),
\qquad  P^{\alpha,\theta}\text{-a.s.}
\end{equation}
for some $d$-dimensional $(\FF^X,P^{\alpha,\theta})$-Brownian motion $B^{P^{\alpha,\theta}}$.

For $\alpha\in\mathcal{A}$ with components $(\alpha^0,\alpha^1)$ as in \eqref{eq:control_bijection}, we define the joint measure $\PP^\alpha\in Prob(\Omega)$ by
\begin{equation}\label{eq:disintegration}
\PP^\alpha(\{\theta\}\times A)\ :=\ \mu(\{\theta\})\,P^{\alpha^\theta,\theta}(A),
\qquad A\in\cF_T^X,\ \theta\in\{0,1\},
\end{equation}
so that under $\PP^\alpha$, $\Theta\sim\mu$ and the conditional law of $X$ given $\Theta=\theta$ is $P^{\alpha^\theta,\theta}$. Its $X$-marginal is
\begin{equation}\label{eq:defpp}
P^{\PP^\alpha}\ :=\ p_0\,P^{\alpha^0,0}+(1-p_0)\,P^{\alpha^1,1}\ \in\ Prob(\Omega_X).
\end{equation}
Under $\PP^\alpha$, the process $X$ satisfies
\begin{equation}\label{eq:X_dynamics_big}
X_t = X_0 + \int_0^t \sigma(s,X_s)\Bigl(\lambda(s,X_s,\Theta,\alpha_s)\,ds + dB_s^{\PP^\alpha}\Bigr),
\qquad 0\le t\le T,\ \PP^\alpha\text{-a.s.},
\end{equation}
for some $d$-dimensional $(\FF,\PP^\alpha)$-Brownian motion $B^{\PP^\alpha}$.

\begin{remark}
Similarly to \cite{cvitanic2018dynamic}, one could introduce \emph{control models}: pairs $(\PP,\alpha)\in Prob(\Omega)\times \mathcal{A}$ satisfying \eqref{eq:X_dynamics_big} with $\PP(\,\cdot\mid\Theta=\theta)\ll{\bar P}$. However, thanks to weak uniqueness of \eqref{sde:sigma} and Girsanov's theorem, such a pair is entirely determined by the control: $\PP=\PP^\alpha$. We therefore work directly with $\alpha\in\mathcal{A}$ (or equivalently with the pair $(\alpha^0,\alpha^1)\in(\mathcal{A}^X)^2$) and the induced measure $\PP^\alpha$ throughout.
\end{remark}

The boundedness of $\lambda$ ensures that for every $\alpha\in\mathcal{A}$, all conditional distributions $P^{\alpha^\theta,\theta}$ and marginal laws $P^{\PP^\alpha}$ are absolutely continuous with respect to ${\bar P}$.

\section{The contracting problem}\label{section.PA.problem}

\subsection{Agent's problem}

We are given a cost function
\begin{equation}\label{eq:cost}\notag
c:[0,T]\times\RR^d\times\{0,1\}\times A\to\RR
\end{equation}
that is bounded, measurable in $(t,x)$ for each $(\theta,\alpha)$, and continuous in $\alpha$ for each $(t,x,\theta)$. Since $A$ is a compact subset of a finite-dimensional Euclidean space, these conditions make $c$ jointly Borel measurable, so that for every $\alpha\in\mathcal{A}$ the process $t\mapsto c(t,X_t,\Theta,\alpha_t)$ is automatically $\FF$-progressively measurable.

We denote by $\mathcal{C}_a$ the set of contracts, i.e.\ $\cF_T^X$-measurable random variables $\xi$ satisfying
\begin{equation}\label{def:cont}\notag
\mathcal{C}_a
:=
\left\{
\xi\ \ \cF_T^X\text{-measurable}:
\EE^{\bar P}\bigl[|\xi|^2\bigr]<\infty
\right\}.
\end{equation}

At time $t=0$, the principal offers a contract $\xi\in\mathcal{C}_a$.
At time $t=0+$, the agent learns her private information $\Theta$, which is not communicated to the principal.
For a given $\beta:\{0,1\}\mapsto (0,\infty)$, if the realization is $\theta \in \{0,1\}$, the agent solves
\begin{equation}\label{eq:agent_problem_k}
\sup_{\alpha\in\mathcal{A}^X} J_a(\alpha;\xi,\theta),
\end{equation}
where
\begin{equation}\label{eq:agent_value_k}\notag
J_a(\alpha;\xi,\theta)
:=
\EE^{P^{\alpha,\theta}}\!\left[
e^{-\int_0^T \kappa(r,X_r)\,dr}\,\beta_\theta\xi
-
\int_0^T e^{-\int_0^s \kappa(r,X_r)\,dr}\,c\bigl(s,X_s,\theta,\alpha_s\bigr)\,ds
\right],
\end{equation}
and $\kappa:[0,T]\times\RR^d\to\RR$ is a given bounded continuous function. Note that for fixed $\theta$ this is exactly the problem solved in \cite{cvitanic2018dynamic}.

Equivalently, using the joint formulation on $\Omega$, for $\alpha\in\mathcal{A}$ with components $(\alpha^0,\alpha^1)$ as in \eqref{eq:control_bijection} we define
\begin{equation}\label{eq:agent_value_joint}\notag
J_a(\alpha;\xi)
:=
\EE^{\PP^\alpha}\!\left[
e^{-\int_0^T \kappa(r,X_r)\,dr}\,\beta_\Theta\xi
-
\int_0^T e^{-\int_0^s \kappa(r,X_r)\,dr}\,c\bigl(s,X_s,\Theta,\alpha_s\bigr)\,ds
\right],
\end{equation}
and the disintegration \eqref{eq:disintegration} yields
\[
J_a(\alpha;\xi)
=
p_0\,J_a(\alpha^0;\xi,0)+(1-p_0)\,J_a(\alpha^1;\xi,1).
\]
The agent's optimization can therefore be written as
\begin{equation}\label{eq:agent_problem_k2}\notag
\sup_{\alpha\in\mathcal{A}} J_a(\alpha;\xi)=p_0\sup_{\alpha^0\in\mathcal{A}^X} J_a(\alpha^0;\xi,0)+(1-p_0)\sup_{\alpha^1\in\mathcal{A}^X} J_a(\alpha^1;\xi,1).
\end{equation}
For any $\xi\in\mathcal{C}_a$, we denote by
\[
\mathcal{A}^\ast(\xi)
:=
\left\{
(\alpha^0,\alpha^1)\in(\mathcal{A}^X)^2:
\ \alpha^\theta\text{ attains }\sup_{\alpha\in\mathcal{A}^X}J_a(\alpha;\xi,\theta)\ \text{for }\theta\in\{0,1\}
\right\}
\]
the set of optimal responses of the agent to a contract given different types.
We fix $R\in\RR$ and $R_\theta\in\RR$ and define two notions of rationality.

\begin{definition}\label{def:rationality}
We say that $\xi\in\mathcal{C}_a$ is \emph{unconditionally rational} if
\[
\sup_{\alpha\in\mathcal{A}} J_a(\alpha;\xi)=p_0\sup_{\alpha^0\in\mathcal{A}^X}J_a(\alpha^0;\xi,0)+(1-p_0)\sup_{\alpha^1\in\mathcal{A}^X}J_a(\alpha^1;\xi,1)
\ \ge\ R.
\]
We say that $\xi\in\mathcal{C}_a$ is \emph{individually (or conditionally) rational} if
\[
\sup_{\alpha\in\mathcal{A}^X}
J_a(\alpha;\xi,\theta)\ \ge\ R_\theta,
\qquad \theta\in\{0,1\}.
\]
We denote by $\mathcal{CR}\subset\mathcal{C}_a$ (resp.\ $\mathcal{UCR}\subset\mathcal{C}_a$) the set of
conditionally rational (resp.\ unconditionally rational) contracts.
\end{definition}
In order to compare our results with screening contracts, we use the definition in \cite[Definition 5.5]{cvitanic2013dynamics}.

\begin{definition}\label{def:screening}
A family $(\xi_\theta)_{\theta\in\{0,1\}}\in \mathcal{C}_a^2$ is called \emph{incentive-compatible screening contracts}
(termed \emph{initially individually rational} in \cite{cvitanic2013dynamics})
if, for all $\theta,\theta'\in\{0,1\}$,
\[
\sup_{\alpha\in\mathcal{A}^X} J_a(\alpha;\xi_\theta,\theta)
\ \ge\
\sup_{\alpha\in\mathcal{A}^X} J_a(\alpha;\xi_{\theta'},\theta)\vee R_\theta.
\]
That is, each type weakly prefers the contract designed for her over the alternative.
We denote by $\mathcal{S}\subset \mathcal{C}_a^2$ the set of all such families.
\end{definition}
We compare below the different notions of rationality and incentive compatibility after defining the principal's problem.

\subsection{Principal's problems}
Unlike screening models, in our framework, the principal offers a single contract (not a menu), implying that the principal does not initially learn $\Theta$ and only knows its prior distribution $\mu$.
We fix a concave, non-decreasing utility function $U_p$ for the principal, and a Lipschitz continuous liquidation function $\Gamma:\RR\mapsto \RR$ so that the terminal utility of the principal is $U_p(\Gamma (X_T)-\xi)$. Note that in this expression, both $X_T$ and $\xi \in \cC_a$ are $\cF^X_T$-measurable. Thus, the expected utility of the principal does not depend on $\PP^\alpha$ for $\alpha\in \mathcal{A}$ but only on the $X$-marginal $P^{\PP^\alpha}$ defined in \eqref{eq:defpp}.
We now define the principal's values for
conditionally, and unconditionally rational contracts, respectively.
\begin{equation}\label{eq:Vp_conditional}
V_{p,c} := \sup_{\xi\in\mathcal{CR}}\ \sup_{(\alpha^0,\alpha^1)\in\mathcal{A}^\ast(\xi)}
\EE^{\PP^{\alpha}}\bigl[U_p(\Gamma (X_T)-\xi)\bigr],
\end{equation}
\begin{equation}\label{eq:Vp_unconditional}
V_{p,uc} := \sup_{\xi\in\mathcal{UCR}}\sup_{(\alpha^0,\alpha^1)\in\mathcal{A}^\ast(\xi)}
\EE^{\PP^\alpha}\bigl[U_p(\Gamma (X_T)-\xi)\bigr],
\end{equation}
where, in both expressions, $\alpha\in\mathcal{A}$ denotes the control associated to $(\alpha^0,\alpha^1)$ via \eqref{eq:control_bijection}. Recalling \eqref{eq:defpp}, these values can also be written as
\begin{equation}\label{eq:Vp_conditional2}
V_{p,c} := \sup_{\xi\in\mathcal{CR}}\ \sup_{(\alpha^0,\alpha^1)\in\mathcal{A}^\ast(\xi)}
\Bigl\{p_0\EE^{P^{\alpha^0,0}}\bigl[U_p(\Gamma (X_T)-\xi)\bigr]+(1-p_0)\EE^{P^{\alpha^1,1}}\bigl[U_p(\Gamma (X_T)-\xi)\bigr]\Bigr\},
\end{equation}
\begin{equation}\label{eq:Vp_unconditional2}
V_{p,uc} := \sup_{\xi\in\mathcal{UCR}}\sup_{(\alpha^0,\alpha^1)\in\mathcal{A}^\ast(\xi)}
\Bigl\{p_0\EE^{P^{\alpha^0,0}}\bigl[U_p(\Gamma (X_T)-\xi)\bigr]+(1-p_0)\EE^{P^{\alpha^1,1}}\bigl[U_p(\Gamma (X_T)-\xi)\bigr]\Bigr\}.
\end{equation}

Similarly, for incentive-compatible screening contracts, we define
\begin{align}\label{eq:Vp_screening}
V_{p,s} := \sup_{(\xi_\theta)_{\theta\in\{0,1\}}\in\mathcal{S}}
&p_0\,
\sup_{(\alpha^0,\alpha^1)\in\mathcal{A}^\ast(\xi_0)}
\EE^{P^{\alpha^0,0}}\bigl[U_p(\Gamma (X_T)-\xi_0)\bigr]\\
+&(1-p_0)\,
\sup_{(\alpha^0,\alpha^1)\in\mathcal{A}^\ast(\xi_1)}
\EE^{P^{\alpha^1,1}}\bigl[U_p(\Gamma (X_T)-\xi_1)\bigr].\notag
\end{align}

\begin{remark}
\begin{itemize}

  \item[(i)] For any initially individually rational screening contracts $(\xi_0,\xi_1)\in \cS$, the agent of type $\theta$ will choose the contract $\xi_\theta$. As such, at time $t=0+$, the principal learns the type of the agent.
\item[(ii)] In view of \eqref{eq:Vp_conditional2}–\eqref{eq:Vp_unconditional2}, the principal’s value is convex in $p_0$. Meanwhile, \eqref{eq:Vp_conditional}–\eqref{eq:Vp_unconditional} provide more convenient expressions for the dynamic formulation of the problem.
\end{itemize}
\end{remark}
The participation (rationality) constraints can be interpreted as follows. At time $t=0$, the principal offers a contract $\xi$, and the agent knows the prior distribution $\mu$ of her type. Immediately after, at time $t=0+$, the agent observes the realization of $\Theta$ and may condition her actions on this information; accordingly, her control is given by a family of processes $\{(\alpha_t^\theta)\colon \theta\in\{0,1\}\}$.

The \emph{unconditionally rational condition} requires that the agent's expected utility, taken with respect to the prior $\mu$, is no less than her reservation utility. This reflects the requirement that the agent be incentivized ex ante, before the realization of her type.

The \emph{individually (or conditionally) rational condition}, in the sense of \cite{cvitanic2013dynamics}, requires that for each realized type, the agent's utility conditional on this information is no less than her reservation utility. This corresponds to a setting in which the contract must deliver sufficient utility for each realized type, once these are observed.

In the screening framework described in \cite{cvitanic2013dynamics}, the principal offers a \emph{menu} of contracts $(\xi_0,\xi_1)$ that satisfy the initial individual rationality/incentive compatibility, so that an agent of type $\theta$ optimally selects $\xi_\theta$ at time $t=0+$ (self-selection). Hence, the principal learns the realized type $\Theta$ immediately from the chosen contract, and the subsequent contracting problem becomes type-by-type with \emph{no remaining hidden parameter}. In particular, there is no need to filter $\Theta$ from the output process $X$ after $t=0+$, since the informational asymmetry is resolved at the contract-selection time. By contrast, in our setting the principal offers a \emph{single} contract (not a menu) and therefore does not observe $\Theta$ at $t=0+$; the principal must keep track of the $F^X$-conditional law of $\Theta$ over time, which leads to a genuine filtering (belief-update) component in the dynamic formulation.
Our stochastic target formulation below allows us to use the methodology of Sannikov \cite{sannikov2008continuous,cvitanic2018dynamic} in the presence of information asymmetry which has the main advantage of allowing us to perform this filtering procedure.

\section{First reduction of the problem of the principal}\label{s:firstred}
In this section, we exhibit state constraints that allow us to reduce \eqref{eq:Vp_conditional} and \eqref{eq:Vp_unconditional} to a state constraint control problem.

\subsection{Control problem of the agent for each type}

At time $0$ the agent learns $\Theta$ and can condition her controls on this information. Given the control problem \eqref{eq:agent_problem_k} and the state dynamics \eqref{eq:X_dynamics_small}, we define the Hamiltonian by
\begin{align}\label{def:hk}
    H^\theta(t,x,z):=\sup_{\a\in A}\{z^{\top} \sigma(t,x) \lambda(t,x,\theta,\a)  -c(t,x,\theta,\a) \},
\end{align}
and, as in \cite{cvitanic2018dynamic}, for all $(\theta,y_\theta)\in \{0,1\}\times \RR$ and every $\FF^X$-adapted $\RR^d$-valued process $Z^\theta$, we introduce the forward process
$Y^\theta=Y^{\theta,0,X_0,y_\theta,Z^\theta}$ defined under $(\FF^X,\bar P)$ by
\begin{align}
    \label{fsde}
    Y^\theta_s=y_\theta-\int_0^s \bigl[H^\theta(r,X_r,Z_r^\theta)-\kappa(r,X_r)Y^\theta_r\bigr]dr+\int_0^s (Z^\theta_r)^\top dX_r.
\end{align}
This process plays the role of the agent's continuation value. All stochastic equations below are stated under $\bar P$; since $P^{\alpha,\theta}\ll\bar P$ with densities in every $L^q$ (by boundedness of $\lambda$), the corresponding statements hold $P^{\alpha,\theta}$-a.s.\ as well for every $\alpha\in\mathcal{A}^X$.

For an $\FF^X$-adapted process $\phi$, we use the standard BSDE norms \cite{el1997backward}
\begin{align}
    \|\phi\|^2_{\HH_2(\bar P)}&:=\EE^{\bar P}\!\left[\int_0^T |\phi_s|^2\,ds\right], &
    \|\phi\|^2_{\bS_2(\bar P)}&:=\EE^{\bar P}\!\left[\sup_{0\leq s\leq T} |\phi_s|^2\right].\notag
\end{align}

\begin{definition}\label{def:response}
We denote by $\cV=\cV(X_0)$\footnote{The set depends on $X_0$ because of the dependence of $\bar P$ on $X_0$. The dependence will be omitted if there is no confusion.} the set of $\FF^X$-adapted $\RR^d$-valued processes $Z$ with $\|Z\|_{\HH_2(\bar P)}<\infty$. For $(t,x)\in[0,T]\times\RR^d$, we write $\cV(t,x)$ for the analogous space of $\FF^X$-adapted processes on $[t,T]$ associated to the reference measure of the SDE \eqref{sde:sigma} started at $X_t=x$.
\end{definition}

By \cite[Proposition~3.3]{cvitanic2018dynamic}, for fixed $y_\theta\in\RR$ and $Z\in\cV$, if the contract is $\xi=Y^{\theta,0,X_0,y_\theta,Z}_T$ and $\Theta=\theta$, then the optimal control of the agent at time $t$ is to choose any control in $A_\theta^*(t,X_t,Z_t)$ defined by
\begin{align}\label{eq:defAstar}
    A_\theta^*(t,x,z):=\argmax_{\a\in A} \{z^{\top} \sigma(t,x) \lambda(t,x,\theta,\a)  -c(t,x,\theta,\a) \},
\end{align}
which is non-empty by continuity of the integrand in $\alpha$ and compactness of $A$.

\begin{remark}\label{rem:selection}
Under our standing assumptions on $\sigma,\lambda,\kappa,c$ and compactness of $A$, for every $Z\in\cV$ there exists $(\alpha^0,\alpha^1)\in(\cA^X)^2$ such that, for each $\theta\in\{0,1\}$,
\[
H^\theta(t,X_t,Z_t)= Z^\top_t\sigma(t,X_t)\lambda(t,X_t,\theta,\alpha^\theta_t)-c(t,X_t,\theta,\alpha^\theta_t),\qquad dt\times d\bar P\text{-a.s.,}
\]
by the Kuratowski--Ryll-Nardzewski measurable selection theorem applied to the non-empty correspondence $A^*_\theta$. Moreover, $Y^{\theta,0,X_0,y,Z}\in\bS_2(\bar P)$ for every $y\in\RR$ by Gr\"onwall's inequality applied to the linear-in-$Y$ dynamics \eqref{fsde}.
\end{remark}

In our framework with information asymmetry, the contract cannot depend on the type $\Theta$, which means that $\frac{Y^{\theta,0,X_0,y_\theta,Z^\theta}_T}{\beta_\theta}$ must be independent of $\theta$. We formulate this constraint as a target problem in the spirit of \cite{soner2002dynamic}. Denote by $G:=\{(\beta_0y,\beta_1y):y\in\RR\}$ the subspace of $\RR^2$ in the direction $(\beta_0,\beta_1)$.

\begin{definition}
Following the terminology of \cite{cvitanic2013dynamics}, we define the \emph{credible set}
\begin{align*}
    \cE(X_0)&:=\bigl\{(y_0,y_1)\in \RR^2: \exists (Z^0,Z^1)\in \cV^2 \text{ s.t.\ } \frac{Y^{0,0,X_0,y_0,Z^0}_T}{\beta_0}=\frac{Y^{1,0,X_0,y_1,Z^1}_T}{\beta_1},\, \text{a.s.}\bigr\}\\
    &=\bigl\{(y_0,y_1)\in \RR^2: \exists (Z^0,Z^1)\in \cV^2 \text{ s.t.\ }(Y^{0,0,X_0,y_0,Z^0}_T,Y^{1,0,X_0,y_1,Z^1}_T) \in G,\, \text{a.s.}\bigr\},\notag
\end{align*}
where the almost sure statements are with respect to $\bar P$. For $(y_0,y_1)\in\cE(X_0)$, we  define the set of admissible control pairs realizing the target condition 
\begin{align}\label{eq:contz}
\cV(y_0,y_1)&:=\bigl\{(Z^0,Z^1)\in \cV^2: \frac{Y^{0,0,X_0,y_0,Z^0}_T}{\beta_0}=\frac{Y^{1,0,X_0,y_1,Z^1}_T}{\beta_1},\, \bar P\text{-a.s.}\bigr\}.
\end{align}
\end{definition}

The common value of $\frac{Y^{0,0,X_0,y_0,Z^0}_T}{\beta_0}=\frac{Y^{1,0,X_0,y_1,Z^1}_T}{\beta_1}$ is the contract (a measurable function of the paths of $X$ but independent of type), and on the event $\{\Theta=\theta\}$, classical optimal control arguments characterize the best response of the agent via \eqref{eq:defAstar}. The term ``credible'' is justified by the fact that $(y_0,y_1)\in\cE(X_0)$ represents the pairs of continuation (or promised) utilities for the two types that can be simultaneously induced by some contract.

To make the description of the credible set explicit, for $(t,x,\theta,\xi)\in [0,T]\times \RR^d\times\{0,1\}\times \cC_a$ we introduce the BSDE
\begin{align}\label{eq:bsdek}
    \cY^{\theta,0,X_0,\xi}_s=\beta(\theta)\xi+\int_s^T\bigl[H^\theta(r,X_r,\cZ^{\theta,0,X_0,\xi}_r)-\kappa(r,X_r) \cY^{\theta,0,X_0,\xi}_r\bigr]dr-\int_s^T (\cZ^{\theta,0,X_0,\xi}_r)^\top dX_r,\, \bar P\text{-a.s.}
\end{align}
The well-posedness of \eqref{eq:bsdek} and the resulting parameterization of $\cE(X_0)$ are provided in the following lemma.

\begin{lemma}\label{lem:bsde_wp}
Under our standing assumptions on $\sigma,\lambda,\kappa,c$ and compactness of $A$, for each $\theta\in\{0,1\}$ and every $\xi\in\cC_a$, the BSDE \eqref{eq:bsdek} admits a unique solution
\[
(\cY^{\theta,0,X_0,\xi},\cZ^{\theta,0,X_0,\xi})\in\bS_{2}(\bar P)\times\HH_{2}(\bar P)
\]
under $(\FF^X,\bar P)$. The solution map is Lipschitz continuous in the terminal condition: there exists a constant $C>0$, depending only on $T$ and on $\|\sigma\|_\infty,\|\sigma^{-1}\|_\infty,\|\lambda\|_\infty,\|\kappa\|_\infty$, such that for all $\xi,\xi'\in\cC_a$,
\[
\|\cY^{\theta,0,X_0,\xi}-\cY^{\theta,0,X_0,\xi'}\|_{\bS_{2}(\bar P)}+\|\cZ^{\theta,0,X_0,\xi}-\cZ^{\theta,0,X_0,\xi'}\|_{\HH_{2}(\bar P)}\ \le\ C\,\|\xi-\xi'\|_{L^{2}(\bar P)}.
\]
Moreover,
\begin{equation}\label{eq:cE_rep}
\cE(X_0)=\bigl\{(\cY^{0,0,X_0,\xi}_0,\cY^{1,0,X_0,\xi}_0):\xi\in \cC_a\bigr\},
\end{equation}
{and $\cE(X_0)$ is a connected set}.
\end{lemma}

Lemma~\ref{lem:bsde_wp} shows that, rather than optimizing \eqref{eq:Vp_conditional} and \eqref{eq:Vp_unconditional} over contracts $\xi\in\cC_a$, the principal can equivalently optimize over pairs $(y_0,y_1)\in\cE(X_0)$ and control pairs $(Z^0,Z^1)\in\cV (y_0,y_1)$ satisfying the target condition
\begin{align}\label{targetcondition}\notag
    \frac{Y^{0,0,X_0,y_0,Z^0}_T}{\beta_0}=\frac{Y^{1,0,X_0,y_1,Z^1}_T}{\beta_1},\quad \bar P\text{-a.s.}
\end{align}
This equivalence is useful because for $(y_0,y_1)\in\cE(X_0)$, and control pairs $(Z^0,Z^1)\in\cV (y_0,y_1)$, we can easily solve the optimization problem of each type.
\begin{proposition}\label{prop:principal_target}
Let $(y_0,y_1)\in \cE(X_0)$ and $(Z^0,Z^1)\in \cV (y_0,y_1)$, and set
\[
\xi:=\frac{Y^{0,0,X_0,y_0,Z^0}_T}{\beta_0}=\frac{Y^{1,0,X_0,y_1,Z^1}_T}{\beta_1},\qquad \bar P\text{-a.s.}
\]
Then $\xi\in\cC_a$ and, for each $\theta\in\{0,1\}$,
\[
y_\theta=\sup_{\alpha\in\mathcal{A}^X} J_a(\alpha;\xi,\theta).
\]
Moreover, a pair $(\alpha^0,\alpha^1)\in(\mathcal{A}^X)^2$ belongs to $\mathcal{A}^\ast(\xi)$ if and only if
\[
\a^\theta_t\in A_\theta^*(t,X_t,Z^\theta_t),\qquad dt\times d\bar P\text{-a.s.,}\quad \theta\in\{0,1\}.
\]
Consequently, the principal's values admit the representations
\begin{align}\label{eq:Vpc_target}
V_{p,c}&=\sup\bigl\{\EE^{\PP^{\alpha}}[U_p(\Gamma(X_T)-\beta_0^{-1}Y^{0,0,X_0,y_0,Z^0}_T)]\,:\ (y_0,y_1)\in\cE(X_0),\ y_\theta\ge R_\theta,\ \theta\in\{0,1\};\notag\\
&\qquad\qquad\qquad\ (Z^0,Z^1)\in\cV (y_0,y_1),\ \alpha^\theta_t\in A_\theta^*(t,X_t,Z^\theta_t),\ \theta\in\{0,1\}\bigr\},\\[0.5em]
\label{eq:Vpuc_target}
V_{p,uc}&=\sup\bigl\{\EE^{\PP^{\alpha}}[U_p(\Gamma(X_T)-\beta_0^{-1}Y^{0,0,X_0,y_0,Z^0}_T)]\,:\ (y_0,y_1)\in\cE(X_0),\ p_0 y_0+(1-p_0)y_1\ge R;\notag\\
&\qquad\qquad\qquad\ (Z^0,Z^1)\in\cV (y_0,y_1),\ \alpha^\theta_t\in A_\theta^*(t,X_t,Z^\theta_t),\ \theta\in\{0,1\}\bigr\},
\end{align}
where $\beta_0^{-1}Y^{0,0,X_0,y_0,Z^0}_T=\beta_1^{-1}Y^{1,0,X_0,y_1,Z^1}_T$, $\alpha\in\mathcal{A}$ is associated to $(\alpha^0,\alpha^1)$ via \eqref{eq:control_bijection}, and $\PP^\alpha$ is defined as in \eqref{eq:disintegration}.
\end{proposition}

The key observation is that a single $\cF^X_T$-measurable contract $\xi$ admits the two representations $$\xi=\beta_0^{-1}Y^{0,0,X_0,y_0,Z^0}_T=\beta_1^{-1}Y^{1,0,X_0,y_1,Z^1}_T.$$ For each $\theta\in\{0,1\}$, the representation $\xi=\beta_\theta^{-1}Y^{\theta,0,X_0,y_\theta,Z^\theta}_T$ places the agent's type-$\theta$ problem directly in the setting of \cite[Proposition~3.3]{cvitanic2018dynamic}, and the first two statements follow by applying that proposition $\theta$-by-$\theta$. The representations of the principal's values then follow from the parametrization \eqref{eq:cE_rep} in Lemma~\ref{lem:bsde_wp}. We therefore omit the proof.

\subsection{Description of the credible set}

Our methodology consists in keeping track of all responses of agents of different types. To write the principal's problem as a standard control problem, we need a more explicit description of $\cE(X_0)$ together with the admissible control pairs $(Z^0,Z^1)\in\cV (y_0,y_1)$.

The forward dynamics \eqref{fsde} of $Y^\theta$ are linear in $Y^\theta$ with a coefficient $-\kappa$ that is independent of $\theta$. It follows that, for every $c\in\RR$, $(y_0,y_1)\in\cE(X_0)$ if and only if $(y_0+\beta_0 c,y_1+\beta_1 c)\in\cE(X_0)$. To characterize $\cE(X_0)$, it therefore suffices to study the projection of the set on the orthogonal to $(\beta_0,\beta_1)$ and we introcuce the process
\[
\Delta_s:=\Delta^{0,X_0,\delta,Z^0,Z^1}_s:=\beta_1 Y^0_s-\beta_0 Y^1_s,
\]
whose dynamics under $(\FF^X,\bar P)$ is
\begin{align}
    \label{fsdediff}\notag
    \Delta_s&=\delta-\int_0^s \bigl[\beta_1 H^0(r,X_r,Z_r^0)-\beta_0 H^1(r,X_r,Z_r^1)-\kappa(r,X_r)\Delta_r\bigr]dr\\
    &\quad+\int_0^s (\beta_1 Z^0_r-\beta_0 Z^1_r)^\top dX_r.\notag
\end{align}
More generally, for $(t,x)\in[0,T]\times\RR^d$,  and $\delta\in\RR$, we define $\Delta_s=\Delta^{t,x,\delta,Z^0,Z^1}_s$ on $\{X_t=x\}$ for $s\in[t,T]$ by
\begin{align*}
    \Delta_s&=\delta-\int_t^s \bigl[\beta_1 H^0(r,X_r,Z_r^0)-\beta_0H^1(r,X_r,Z_r^1)-\kappa(r,X_r)\Delta_r\bigr]dr\\
    &\quad+\int_t^s (\beta_1 Z^0_r-\beta_0 Z^1_r)^\top dX_r,\,\bar P\text{-a.s.}\notag
\end{align*}
By Lemma~\ref{lem:bsde_wp}, $\cE(X_0)$ is connected, so the set
\[
\bigl\{\delta\in\RR: \exists (Z^0,Z^1)\in(\cV(t,x))^2 \text{ s.t.\ }\Delta^{t,x,\delta,Z^0,Z^1}_T=0 \text{ on }\{X_t=x\}\bigr\}
\]
is an interval of $\RR$, which we now characterize.

Following \cite{bouchard2010stochastic,soner2002stochastic}, we expect the smallest and largest values of $\delta$ in this interval to be $\underline W(t,x)$ and $\overline W(t,x)$, defined as the solutions to
\begin{align}\label{eq:pdeuwk}
    -\pa_t \underline W-\tfrac{1}{2}\mathrm{tr}(\sigma^2(t,x)\pa_{xx}\underline W)-\underline H(t,x,\underline W,\pa_x\underline W)&=0,\\
    -\pa_t \overline W-\tfrac{1}{2}\mathrm{tr}(\sigma^2(t,x)\pa_{xx}\overline W)-\overline H(t,x,\overline W,\pa_x\overline W)&=0,\label{eq:pdeowk}\\
    \underline W(T,x)=\overline W(T,x)&=0,\notag
\end{align}
where, for $(t,x,y,z)\in[0,T]\times\RR^d\times\RR\times\RR^d$,
\begin{align}
\underline H(t,x,y,z)&:=-\kappa(t,x)\,y+\inf_{z_1\in\RR^d}\left(\beta_1 H^0\bigl(t,x,\frac{\beta_0 z_1+z}{\beta_1}\bigr)-\beta_0 H^1\bigl(t,x,z_1\bigr)\right)\label{low_H},\\
\overline H(t,x,y,z)&:=-\kappa(t,x)\,y+\sup_{z_1\in\RR^d}\left(\beta_1 H^0\bigl(t,x,\frac{\beta_0 z_1+z}{\beta_1}\bigr)-\beta_0 H^1\bigl(t,x,z_1\bigr)\right).\label{high_H}
\end{align}
When $\underline W$ and $\overline W$ are smooth, setting
\[
(\underline Y_s,\underline Z_s):=(\underline W(s,X_s),\pa_x\underline W(s,X_s))\quad\text{and}\quad (\overline Y_s,\overline Z_s):=(\overline W(s,X_s),\pa_x\overline W(s,X_s)),
\]
It\^o's formula shows that these processes satisfy the BSDEs
\begin{align}
    \label{eq:bsdeup}
    \overline Y_s&=\int_s^T\overline H\bigl(r,X_r,\overline Y_r,\overline Z_r\bigr)\,dr-\int_s^T \overline Z_r^\top\,dX_r,\\
    \label{eq:bsdedown}
    \underline Y_s&=\int_s^T\underline H\bigl(r,X_r,\underline Y_r,\underline Z_r\bigr)\,dr-\int_s^T \underline Z_r^\top\,dX_r,
\end{align}
under $(\FF^X,\bar P)$.

To give a self-contained proof of the characterization of $\cE(X_0)$, we make the following assumption on the solutions of \eqref{eq:pdeuwk}--\eqref{eq:pdeowk}.

\begin{ass}[Boundary dynamics]\label{ass:boundary}
\begin{enumerate}
    \item[(i)] The PDEs \eqref{eq:pdeuwk} and \eqref{eq:pdeowk} admit unique $C^{1,2}([0,T)\times\RR^d)\cap C^0([0,T]\times\RR^d)$ solutions, and the sets
    \begin{align}\label{eq:defcv}
        \underline\cV(t,x,p)&:=\argmin_{z_1\in\RR^d}\left(\beta_1 H^0\bigl(t,x,\frac{\beta_0 z_1+z}{\beta_1}\bigr)-\beta_0 H^1\bigl(t,x,z_1\bigr)\right),\\
        \overline\cV(t,x,p)&:=\argmax_{z_1\in\RR^d}\left(\beta_1 H^0\bigl(t,x,\frac{\beta_0 z_1+z}{\beta_1}\bigr)-\beta_0 H^1\bigl(t,x,z_1\bigr)\right)\notag
    \end{align}
    are non-empty for all $(t,x,p)\in[0,T)\times\RR^d\times\RR^d$.
    \item[(ii)] For any $Z^1\in\cV$ such that $Z^1_s\in\underline\cV(s,X_s,\underline Z_s)$ (resp.\ $Z^1_s\in\overline\cV(s,X_s,\overline Z_s)$) holds $dt\times d\bar P$-a.s., we have
    \[
   \frac{\beta_0 Z^1+\underline Z}{\beta_1}\in\cV(X_0)\quad\text{(resp.\ }\frac{\beta_0 Z^1+\overline Z}{\beta_1}\in\cV(X_0)\text{).}
    \]
    \item[(iii)] The BSDEs \eqref{eq:bsdeup}--\eqref{eq:bsdedown} satisfy the strict comparison principle: if $(Y,Z)$ satisfies
    \[
    Y_s=\int_s^T H\bigl(r,X_r,Y_r,Z_r\bigr)\,dr-\int_s^T Z_r^\top\,dX_r
    \]
    with $\underline H\le H\le\overline H$, then $\underline Y_0\le Y_0\le\overline Y_0$. Moreover, if $\underline Y_0=Y_0$ (resp.\ $Y_0=\overline Y_0$), then $\underline Y_t=Y_t$ for all $t\in[0,T]$ (resp.\ $\overline Y_t=Y_t$ for all $t\in[0,T]$) and
    \[
    \underline H(t,X_t,\underline Y_t,\underline Z_t)=H(t,X_t,\underline Y_t,\underline Z_t)\quad\text{(resp.\ }\overline H(t,X_t,\overline Y_t,\overline Z_t)=H(t,X_t,\overline Y_t,\overline Z_t)\text{).}
    \]
\end{enumerate}
\end{ass}

\begin{remark}[Interpretation of the finiteness of $\underline W$ and $\overline W$]
The finiteness of the solutions to \eqref{eq:pdeuwk} and \eqref{eq:pdeowk} admits a concrete economic interpretation. The function $\underline W(t,x)$ (resp.\ $\overline W(t,x)$) is the smallest (resp.\ largest) gap $y_0-y_1$ in the agents' continuation utilities that is compatible with \emph{simultaneously} implementing both types under a \emph{single} contract at state $(t,x)$. Equivalently, $\underline W(t,x)$ is the threshold below which no admissible contract can incentivize both types at once, and $\underline W(t,x)$ thus measures a \emph{feasibility boundary} for implementation under a single contract.

If $\underline W(t,x)=-\infty$ there is no finite utility gap ruling out simultaneous implementation: for every pair of target continuation utilities $(y_0,y_1)\in\RR^2$ small enough, one can find an admissible contract $\xi\in\cC_a$ (possibly depending on $(t,x,y_0,y_1)$) such that, when offered at $(t,x)$, type $\theta$ attains continuation utility $y_\theta$ for $\theta\in\{0,1\}$ and $\xi$ induces the intended incentive-compatible actions for both types. In this case there is no binding lower feasibility constraint and the principal can tailor a single admissible contract to deliver arbitrary finite promised utilities. 
\end{remark}

We are now ready to characterize $\cE(X_0)$ and the admissible control pairs on its boundary.

\begin{theorem}\label{thm:domain}
Under Assumption~\ref{ass:boundary}, for any $X_0\in\RR^d$ the following two statements hold. 
\begin{enumerate}
    \item[(a)] $\cE(X_0)=\bigl\{(y_0,y_1)\in\RR^2: \underline W(0,X_0)\le \beta_1 y_0-\beta_0 y_1\le\overline W(0,X_0)\bigr\}$. In particular, $\cE(X_0)$ is closed.
    \item[(b)] For any $(y_0,y_1)\in\cE(X_0)$ and $(Z^0,Z^1)\in(\cV(X_0))^2$, with $Y^{\theta,0,X_0,y_\theta,Z^\theta}$ given by \eqref{fsde} for $\theta\in\{0,1\}$, the following are equivalent:
    \begin{enumerate}
        \item[(i)] For all $s\in[0,T] $ we have $\underline W(s,X_s)\le \beta_1Y^{0,0,X_0,y_0,Z^0}_s-\beta_0Y^{1,0,X_0,y_1,Z^1}_s\le\overline W(s,X_s)$.
        \item[(ii)] For all $s\in[0,T]$ on the event $\{\underline W(s,X_s)=\beta_1Y^{0,0,X_0,y_0,Z^0}_s-\beta_0Y^{1,0,X_0,y_1,Z^1}_s\}$ we have for $r\in[s,T]$
        \[
        Z^1_r\in\underline\cV(r,X_r,\underline Z_r),\,Z^0_r=\frac{\beta_0 Z^1_r+\underline Z_r}{\beta_1}\mbox{ and }\underline W(r,X_r)=\beta_1Y^{0,0,X_0,y_0,Z^0}_r-\beta_0Y^{1,0,X_0,y_1,Z^1}_r.
        \]
        (resp. on the event $\{\overline W(s,X_s)=\beta_1Y^{0,0,X_0,y_0,Z^0}_s-\beta_0Y^{1,0,X_0,y_1,Z^1}_s\}$  we have for $r\in[s,T]$
      $
        Z^1_r\in\overline\cV(r,X_r,\overline Z_r),\,Z^0_r=\frac{\beta_0 Z^1+\overline Z_r}{\beta_1}\mbox{ and }\overline W(r,X_r)=\beta_1Y^{0,0,X_0,y_0,Z^0}_r-\beta_0Y^{1,0,X_0,y_1,Z^1}_r
      $).

        \item[(iii)] There exists $\xi\in\cC_a$ such that $(\cY^{\theta,0,X_0,\xi},\cZ^{\theta,0,X_0,\xi})=(Y^{\theta,0,X_0,y_\theta,Z^\theta},Z^\theta)$ for $\theta\in\{0,1\}$.
        \item[(iv)] $(Z^0,Z^1)\in\cV (y_0,y_1)$.
    \end{enumerate}
\end{enumerate}
\end{theorem}
\begin{remark}
\begin{enumerate}
\item[(i)]The first item of the Theorem is a characterization for $(y_0,y_1)\in\cE(X_0)$ whereas the second item is a characterization for $(Z^0,Z^1)\in\cV (y_0,y_1)$.
    \item[(ii)] Theorem~\ref{thm:domain} can be proven under slightly weaker hypotheses formulated directly on the BSDEs \eqref{eq:bsdeup}--\eqref{eq:bsdedown}.

    \item[(iii)] The existence of smooth solutions to \eqref{eq:pdeuwk}--\eqref{eq:pdeowk} is stronger than the viscosity-based domain description of 
    \cite{soner2002dynamic,soner2002stochastic,bouchard2010stochastic}. The additional regularity is used here not to characterize $\cE(X_0)$ itself, but to define the gradient processes $\underline Z$ and $\overline Z$ that enter the sets $\underline\cV(t,X_t,\underline Z_t)$ and $\overline\cV(t,X_t,\overline Z_t)$, which are needed for the optimal contract problem.
    \item[(iv)] Part (b) of Theorem~\ref{thm:domain} says that whenever $\beta_1Y^0_s-\beta_0Y^1_s$ touches the boundary of $\cE(X_0)$, the controls must match the optimizers in $\underline\cV$ or $\overline\cV$, which keep $\beta_1Y^0-\beta_0Y^1$ on the boundary until the terminal time. Among all admissible pairs $(Z^0,Z^1)$, this amounts to matching the diffusion coefficient of $\beta_1Y^0-\beta_0Y^1$ with that of $\underline W(\cdot,X_\cdot)$ or $\overline W(\cdot,X_\cdot)$; otherwise, $\beta_1Y^0_t-\beta_0Y^1_t-\overline W(t,X_t)$ would change sign due to Brownian noise.
\end{enumerate}
\end{remark}

\subsection{Interpretation of Theorem \ref{thm:domain}}\label{ss.inter}

In the benchmark setting without adverse selection, i.e., when $\mu$ is a Dirac mass and there is a single type, one of the main contributions of \cite{sannikov2008continuous,cvitanic2018dynamic} is to show that, instead of optimizing over $\xi$ in \eqref{eq:Vp_conditional}--\eqref{eq:Vp_unconditional}, the principal can optimize over the initial value $y$ and a control process $Z$. The process $Z$ is required only to satisfy an integrability condition. The principal then considers the family of contracts $\{\xi=\beta_0^{-1}Y^{0,X_0,y,Z}_T:y\in\RR,\,Z\in\cV\}$, which leads to a stochastic control problem for the principal in which the control $Z$ carries no additional constraints and the state is $(X_t,Y_t)$.

Theorem~\ref{thm:domain} extends this principle to the adverse-selection setting. The equivalence (ii)$\Leftrightarrow$(iii) replaces the contract $\xi$ by the pair of initial promises $(y_0,y_1)\in\RR^2$ and the control pair $(Z^0,Z^1)\in\cV (y_0,y_1)$, while (i) pins down the state constraint on $\beta_1Y^0-\beta_0Y^1$ that emerges from the information asymmetry  (iii) shows that the constraint on $Z^0,Z^1$ is vacuous in the interior of the credible set but binding on its boundary.

As illustrated in Figure \ref{fig:credible-set}, the forward processes $Y^{\theta,0,X_0,y_\theta,Z^\theta}$ are continuous, so at any time $s$ at which the promise gap lies strictly inside the feasible band,
\[
\underline W(t,X_t)<\beta_1 Y^{0,0,X_0,y_0,Z^0}_t-\beta_0 Y^{1,0,X_0,y_1,Z^1}_t<\overline W(t,X_t),
\]
there is no instantaneous restriction on the current values of $(Z^0_t,Z^1_t)$ beyond integrability. Consequently, the Hamiltonian for the principal's control problem will be maximized over $(z^0,z^1)\in\RR^d\times\RR^d$.

By contrast, when the promise gap hits the boundary, the set of admissible values collapses. More precisely, if
\[
\beta_1 Y^{0,0,X_0,y_0,Z^0}_t-\beta_0 Y^{1,0,X_0,y_1,Z^1}_t=\underline W(t,X_t),
\]
then necessarily
\[
Z^1_t\in\underline\cV(t,X_t,\underline Z_t)\qquad\text{and}\qquad \frac{\beta_0 Z^1_t+\underline Z_t}{\beta_1}.
\]
Similarly, if
\[
\beta_1Y^{0,0,X_0,y_0,Z^0}_t-\beta_0Y^{1,0,X_0,y_1,Z^1}_t=\overline W(t,X_t),
\]
then necessarily
\[
Z^1_t\in\overline\cV(t,X_t,\overline Z_t)\qquad\text{and}\qquad \frac{\beta_0 Z^1_t+\overline Z_t}{\beta_1}.
\]
In other words, the interior of the domain corresponds to a ``free'' choice of the current controls, whereas on the boundary, the controls must satisfy a compatibility condition encoded by $\underline\cV$ or $\overline\cV$ together with the gradient correction linking $Z^0$ and $Z^1$. This observation will be used to derive a dynamic characterization of the principal's value function.

\begin{figure}[ht]
\centering
\begin{tikzpicture}[scale=0.85,>=Latex]
  \draw[->] (-0.2,0)--(12,0) node[right,font=\small]{$y_1$};
  \draw[->] (0,-2)--(0,7.5) node[above,font=\small]{$y_0$};
  
  \fill[pattern=north east lines, pattern color=black!25]
    (0,-1)--(10,4)--(10,7)--(0,2)--cycle;
  \draw[thick] (0,2)--(10,7);
  \draw[thick,dashed] (0,-1)--(10,4);
  \node at (2,3.5) [rotate=26.565,font=\scriptsize,fill=white,inner sep=1pt]
    {$\beta_1 Y^0_t-\beta_0 Y^1_t=\overline{W}(t,X_t)$};
  \node at (5.2,1.2) [rotate=26.565,font=\scriptsize,fill=white,inner sep=1pt]
    {$\beta_1 Y^0_t-\beta_0 Y^1_t=\underline{W}(t,X_t)$};
  \node[font=\scriptsize,align=left,anchor=west] at (3.5,3)
    {no condition on};
  \node[font=\scriptsize,align=left,anchor=west] at (3.75,2.5)
    {$(Z^0_t,Z^1_t)$};
  \draw[->, thick] (7.5, 4.5) -- (9, 5.25);
  \node[font=\scriptsize, anchor=west] at (8, 4.5) {$(\beta_0,\beta_1)$};
  \fill (6,5) circle (1.8pt);
  \draw[->,shorten >=2pt] (4.2,6.4) to[out=-30,in=150] (5.92,5.08);
  \node[font=\scriptsize,align=left,anchor=west] at (3.2,7.7) {upper hit:};
  \node[font=\scriptsize,align=left,anchor=west] at (3.2,7.2)
    {$Z^1_t\!\in\!\overline{\mathcal{V}}(t,X_t,\overline{Z}_t)$};
  \node[font=\scriptsize,align=left,anchor=west] at (3.2,6.7)
    {$Z^0_t=\frac{\beta_0 Z^1_t+\overline{Z}_t}{\beta_1}$};
  \draw[fill=white,thick] (8,3) circle (2.0pt);
  \draw[->,shorten >=2pt] (9.2,2) to[out=150,in=-30] (8,2.94);
  \node[font=\scriptsize,align=left,anchor=west] at (7.5,2) {lower hit:};
  \node[font=\scriptsize,align=left,anchor=west] at (7.5,1.5)
    {$Z^1_t\!\in\!\underline{\mathcal{V}}(t,X_t,\underline{Z}_t)$};
  \node[font=\scriptsize,align=left,anchor=west] at (7.5,1)
    {$Z^0_t=\frac{\beta_0 Z^1_t+\underline{Z}_t}{\beta_1}$};
\end{tikzpicture}
\caption{The credible set $\{(y_0,y_1):\underline{W}(t,X_t)\le\beta_1 y_0-\beta_0 y_1\le\overline{W}(t,X_t)\}$ 
for general $\beta_0, \beta_1 > 0$. The constraint on $\beta_1 y_0-\beta_0 y_1$ defines a strip unbounded in the 
direction $(\beta_0,\beta_1)$ by additive invariance. The boundaries have slope $\beta_0/\beta_1$ (here shown for 
$\beta_0=1, \beta_1=2$, giving slope $1/2$). Inside the strip, $(Z^0_t,Z^1_t)$ is unconstrained; on the boundaries, 
the matching conditions via $\underline{\mathcal{V}}, \overline{\mathcal{V}}$ as $Z^0_t=\frac{\beta_0 Z^1_t \pm Z_t}{\beta_1}$ apply.}
\label{fig:credible-set}
\end{figure}
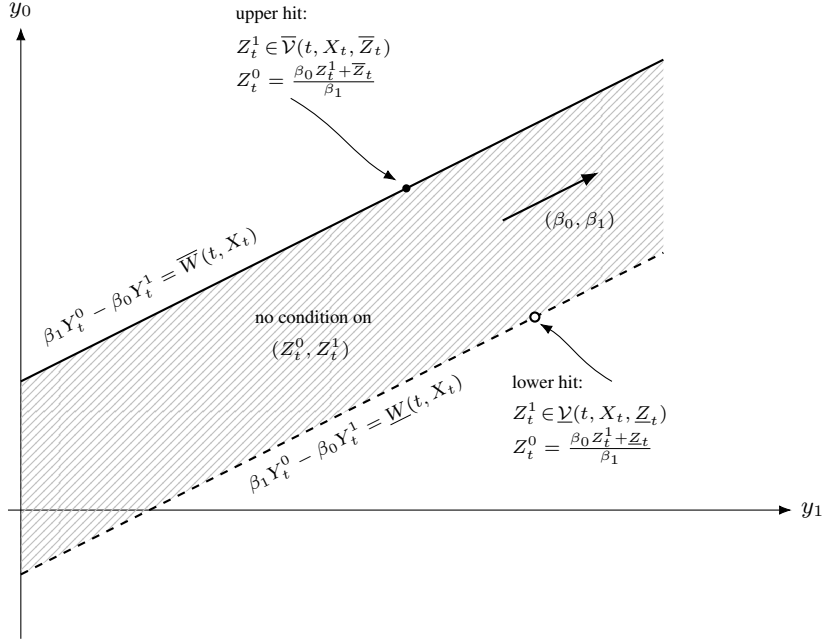

\section{Reduction to a state constraint optimal control problem}\label{s:statec}

Theorem~\ref{thm:domain} exhibits the family of control pairs $(Z^0,Z^1)\in\cV (y_0,y_1)$ as admissible controls for an optimal control problem that we now formalize. Fix $(y_0,y_1)\in\cE(X_0)$ and $(Z^0,Z^1)\in\cV (y_0,y_1)$, and set
$\xi:=\beta^{-1}_0Y^{0,0,X_0,y_0,Z^0}_T=\beta_1^{-1}Y^{1,0,X_0,y_1,Z^1}_T$. By Proposition~\ref{prop:principal_target}, any $(\alpha^0,\alpha^1)\in\mathcal{A}^\ast(\xi)$ satisfies $\alpha^\theta_t\in A^*_\theta(t,X_t,Z^\theta_t)$, $dt\times d\bar P$-a.s.

A key advantage of our stochastic target formulation --- $Y^0$ and $Y^1$ are controlled separately with a terminal target constraint --- is that it yields an explicit description of the $\cF^X_t$-conditional law of $\Theta$ under $\PP^\alpha$, where $\alpha\in\mathcal{A}$ is associated to $(\alpha^0,\alpha^1)$ via \eqref{eq:control_bijection}. Moreover, a crucial simplification in \eqref{eq:Vpc_target}--\eqref{eq:Vpuc_target} is that the integrand $U_p( \Gamma (X_T)-\xi)$ is $\cF^X_T$-measurable, so the expectation $\EE^{\PP^\alpha}[U_p( \Gamma (X_T)-\xi)]$ depends on $\PP^\alpha$ only through its $X$-marginal $P^{\PP^\alpha}$ defined in \eqref{eq:defpp}. We now characterize $P^{\PP^\alpha}$ via a filter.

For compactness, introduce
\begin{align*}
\bar\lambda(t,x,p,\a_0,\a_1)&:=p\,\lambda(t,x,0,\a_0)+(1-p)\,\lambda(t,x,1,\a_1)\in\RR^d,\\
\Delta\lambda(t,x,\a_0,\a_1)&:=\lambda(t,x,0,\a_0)-\lambda(t,x,1,\a_1)\in\RR^d.
\end{align*}
The boundedness of $\lambda$ together with the Lipschitz continuity and the non-degeneracy of $\sigma$ yield the following lemma.

\begin{lemma}[Wonham filter]\label{lem:filter}
Let $(y_0,y_1)\in\cE(X_0)$, $(Z^0,Z^1)\in\cV (y_0,y_1)$, and set $\xi:=\beta^{-1}_0 Y^{0,0,X_0,y_0,Z^0}_T$. Let $(\alpha^0,\alpha^1)\in\mathcal{A}^\ast(\xi)$, and let $\alpha\in\mathcal{A}$ be the associated control via \eqref{eq:control_bijection}, with induced measure $\PP^\alpha\in Prob(\Omega)$ defined in \eqref{eq:disintegration} and $X$-marginal $P^{\PP^\alpha}$ defined in \eqref{eq:defpp}. Define the $\cF^X_t$-conditional probability
\[
p_t\;:=\;\PP^\alpha\bigl(\Theta=0\mid\cF^X_t\bigr),\qquad 0\le t\le T.
\]
Then, the \emph{innovation process}
\begin{equation}\label{eq:filter}\notag
dB^{\PP^\alpha,X}_t
:=
\sigma(t,X_t)^{-1}\,dX_t
-
\bar\lambda\bigl(t,X_t,p_t,\alpha^0_t,\alpha^1_t\bigr)\,dt
\end{equation}
is an $(\FF^X,P^{\PP^\alpha})$-Brownian motion, and $p_t$ satisfies the Kushner--Stratonovich equation
\begin{equation}\label{eq:wonham}
dp_t\;=\;p_t(1-p_t)\,\Delta\lambda\bigl(t,X_t,\alpha^0_t,\alpha^1_t\bigr)^{\!\top}dB^{\PP^\alpha,X}_t,
\qquad p_0=\mu(\{0\}).
\end{equation}
Consequently, the family $(X,p,Y^0,Y^1)$ satisfies the controlled dynamics
\begin{align}\label{eq:xstatec}
dX_t&=\sigma(t,X_t)\bigl(\bar\lambda(t,X_t,p_t,\alpha^0_t,\alpha^1_t)\,dt+dB^{\PP^\alpha,X}_t\bigr),\\
\label{eq:pstatec}
dp_t&=p_t(1-p_t)\,\Delta\lambda(t,X_t,\alpha^0_t,\alpha^1_t)^\top dB^{\PP^\alpha,X}_t,\\
\label{eq:ystatec}
dY^\theta_t&=\Bigl(-H^\theta(t,X_t,Z^\theta_t)+\kappa(t,X_t)Y^\theta_t+(Z^\theta_t)^\top\sigma(t,X_t)\bar\lambda(t,X_t,p_t,\alpha^0_t,\alpha^1_t)\Bigr)dt\notag\\
&\quad+(Z^\theta_t)^\top\sigma(t,X_t)\,dB^{\PP^\alpha,X}_t,\qquad \theta\in\{0,1\}.
\end{align}
\end{lemma}

\begin{proof}
We write $\lambda^\theta_t:=\lambda(t,X_t,\theta,\alpha^\theta_t)$ and $\PP:=\PP^\alpha$ for brevity. We observe that,  under $\PP$, \eqref{eq:X_dynamics_big} gives
\[
\sigma(t,X_t)^{-1}\,dX_t=\lambda^\Theta_t\,dt+dB^{\PP}_t,
\]
where $B^{\PP}$ is a $\PP$-Brownian motion; boundedness of $\sigma^{-1}$ is used. Since $\alpha^0,\alpha^1$ and $X$ are $\FF^X$-adapted, the $\cF^X_t$-conditional expectation of $\lambda^\Theta_t$ is
\[
\widehat\lambda_t:=\EE^\PP\!\bigl[\lambda^\Theta_t\mid\cF^X_t\bigr]=p_t\,\lambda^0_t+(1-p_t)\,\lambda^1_t=\bar\lambda(t,X_t,p_t,\alpha^0_t,\alpha^1_t).
\]
Next, we introduce the innovation process. Define
\[
B^{\PP,X}_t:=\int_0^t\sigma(s,X_s)^{-1}\,dX_s-\int_0^t\widehat\lambda_s\,ds=\int_0^t(\lambda^\Theta_s-\widehat\lambda_s)\,ds+B^{\PP}_t.
\]
Boundedness of $\lambda$ makes the drift term absolutely continuous with bounded density, so $B^{\PP,X}$ is $\FF^X$-adapted and continuous. Moreover, for any bounded $\cF^X_s$-measurable $\eta$ and $s\le t$,
\[
\EE^\PP\!\bigl[\eta\bigl(B^{\PP,X}_t-B^{\PP,X}_s\bigr)\bigr]
=\EE^\PP\!\left[\eta\int_s^t(\lambda^\Theta_u-\widehat\lambda_u)\,du\right]=0
\]
by the tower property, so $B^{\PP,X}$ is an $(\FF^X,\PP)$-martingale. Its quadratic variation coincides with that of $B^{\PP}$, namely $tI_d$. The characterization of L\'evy identifies $B^{\PP,X}$ as an $(\FF^X,\PP)$-Brownian motion. 

Finally, we introduce the  Kushner equation. For a bounded $f:\{0,1\}\to\RR$, set $\pi_t(f):=\EE^\PP[f(\Theta)\mid\cF^X_t]$. Since $\Theta$ is $\cF_0$-measurable and time-invariant, the Fujisaki--Kallianpur--Kunita filtering equation reduces to
\[
d\pi_t(f)=\bigl(\pi_t(f\,(\lambda^{\cdot}_t)^\top)-\pi_t(f)\,\pi_t((\lambda^{\cdot}_t)^\top)\bigr)\,dB^{\PP,X}_t.
\]
Taking $f(\theta)=\mathbf 1_{\{\theta=0\}}$, so that $\pi_t(f)=p_t$, gives
\[
\pi_t(f\,(\lambda^{\cdot}_t)^\top)=p_t\,(\lambda^0_t)^\top,\qquad \pi_t((\lambda^{\cdot}_t)^\top)=\bar\lambda^\top,
\]
and therefore
\[
dp_t=\Bigl(p_t(\lambda^0_t)^\top-p_t\bigl[p_t(\lambda^0_t)^\top+(1-p_t)(\lambda^1_t)^\top\bigr]\Bigr)\,dB^{\PP,X}_t=p_t(1-p_t)(\lambda^0_t-\lambda^1_t)^\top\,dB^{\PP,X}_t,
\]
which is \eqref{eq:wonham}. The dynamics \eqref{eq:xstatec} and \eqref{eq:ystatec} follow by substituting the decomposition $\sigma(t,X_t)^{-1}dX_t=\bar\lambda\,dt+dB^{\PP,X}_t$ into \eqref{eq:X_dynamics_big} and \eqref{fsde}. \qed
\end{proof}
Together with Theorem~\ref{thm:domain}, Lemma~\ref{lem:filter} reformulates the principal's problem as an optimal control problem with state constraints. We introduce the parabolic strip and its boundaries
\begin{align*}
    \cD&:=\bigl\{(t,x,y_0,y_1,p)\in[0,T)\times\RR^d\times\RR\times\RR\times(0,1): \underline W(t,x)<\beta_1 y_0-\beta_0 y_1<\overline W(t,x)\bigr\},\\
    \cD_d&:=\bigl\{(t,x,y_0,y_1,p)\in[0,T)\times\RR^d\times\RR\times\RR\times(0,1): \beta_1 y_0-\beta_0 y_1=\underline W(t,x)\bigr\},\\
    \cD_u&:=\bigl\{(t,x,y_0,y_1,p)\in[0,T)\times\RR^d\times\RR\times\RR\times(0,1): \beta_1 y_0-\beta_0 y_1=\overline W(t,x)\bigr\},
\end{align*}
and write $\mathrm{cl}_y(\cD):=\cD\cup\cD_d\cup\cD_u\cup (\{T\}\times \RR^d\times G\times (0,1))$ which is the closure of $\cD$ in $(t,x,y_0,y_1)$, but not in $p$.

Given the initial data $(0,X_0,y_0,y_1,p_0)\in\mathrm{cl}_y(\cD)$, the controls $(Z^0,Z^1)\in(\cV(X_0))^2$, and $\alpha^\theta_t\in A^*_\theta(t,X_t,Z^\theta_t)$ for $\theta\in\{0,1\}$, let $P\in Prob(\Omega_X)$ carry a $(P,\cF^X)$-Brownian motion $B$ such that $(X,p,Y^0,Y^1)$ is the solution of
\begin{align}\label{eq:xstatec1}
dX_t&=\sigma(t,X_t)\bigl(\bar\lambda(t,X_t,p_t,\alpha^0_t,\alpha^1_t)\,dt+dB_t\bigr),\\
\label{eq:pstatec1}
dp_t&=p_t(1-p_t)\,\Delta\lambda(t,X_t,\alpha^0_t,\alpha^1_t)^\top dB_t,\\
\label{eq:ystatec1}
dY^\theta_t&=\Bigl(-H^\theta(t,X_t,Z^\theta_t)+\kappa(t,X_t)Y^\theta_t+(Z^\theta_t)^\top\sigma(t,X_t)\bar\lambda(t,X_t,p_t,\alpha^0_t,\alpha^1_t)\Bigr)dt\notag\\
&\quad+(Z^\theta_t)^\top\sigma(t,X_t)\,dB_t,\qquad \theta\in\{0,1\}.
\end{align}

Define the value function
\begin{equation}\label{eq:Vsc}
V_{sc}(0,X_0,y_0,y_1,p_0)\;:=\;\sup\EE^{P}\!\bigl[U_p(\Gamma (X_T)-\beta^{-1}_0 Y^0_T)\bigr],
\end{equation}
where the supremum is taken over such probability measures satisfying $P$-a.s. the state constraint
\begin{align}\label{state:c}
   \underline W(t,X_t)\le \beta_1 Y^0_t-\beta_0 Y^1_t \le\overline W(t,X_t),\qquad t\in[0,T]
\end{align}
which is equivalent to $(t,X_t,Y^0_t,Y^1_t,p_t)\in \mathrm{cl}_y(\cD)$, $P$-a.s.

More generally, for any $(t,x,y_0,y_1,p)\in\mathrm{cl}_y(\cD)$, we define $V_{sc}(t,x,y_0,y_1,p)$ as the analogous supremum where the controlled system \eqref{eq:xstatec1}--\eqref{eq:ystatec1} is initialized at time $t$ with $(X_t,Y^0_t,Y^1_t,p_t)=(x,y_0,y_1,p)$, and the state constraint \eqref{state:c} is required to hold $P$-a.s.\ for all $s\in[t,T]$.

\begin{theorem}\label{thm:sc_rep}
The principal's values admit the representations
\begin{align}
V_{p,c}&=\sup\bigl\{V_{sc}(0,X_0,y_0,y_1,p_0)\ :\ y_\theta\ge R_\theta,\ \theta\in\{0,1\}\bigr\},\label{eq:repvaluestatic1sc}\\[0.5ex]
V_{p,uc}&=\sup\bigl\{V_{sc}(0,X_0,y_0,y_1,p_0)\ :\ p_0 y_0+(1-p_0)y_1\ge R\bigr\}.\label{eq:repvaluestaticsc_uc}
\end{align}
\end{theorem}
\begin{proof}[Proof of Theorem \ref{thm:sc_rep}]
The theorem is a consequence of the following claim. Given $(Z^0,Z^1)\in \mathcal{V}(X_0)^2$, the controlled state
$(X,p,Y^0,Y^1)$ defined in 
    \eqref{eq:xstatec1}--\eqref{eq:ystatec1} satisfies the state constraint \eqref{state:c} if and only if $(Z^0,Z^1)\in \cV (y_0,y_1)$. We prove the two implications separately. Throughout, for $(Z^0,Z^1)\in(\cV(X_0))^2$ we fix, via Remark~\ref{rem:selection} and the Kuratowski--Ryll-Nardzewski theorem, an $\FF^X$-progressively measurable selection $\alpha^\theta_t\in A^*_\theta(t,X_t,Z^\theta_t)$ for $\theta\in\{0,1\}$, and we write $\bar\lambda:=\bar\lambda(t,X_t,p_t,\alpha^0_t,\alpha^1_t)$ and $\Delta\lambda:=\Delta\lambda(t,X_t,\alpha^0_t,\alpha^1_t)$ for brevity.

\smallskip
Let $(Z^0,Z^1) \in \mathcal{V}(y_0,y_1)$. By Lemma \ref{lem:filter}, the controlled state $(X,p,Y^0,Y^1)$ satisfies the dynamics  \eqref{eq:xstatec1}--\eqref{eq:ystatec1}. Moreover, by Theorem \ref{thm:domain}, $(Y^0,Y^1)$ satisfy the state constraint \eqref{state:c}. 

\smallskip
Next, we show the reverse implication. Let $(Z^0,Z^1) \in \mathcal{V}(X_0)^2$, and suppose that $(X,p,Y^0,Y^1)$ satisfy \eqref{eq:xstatec1}--\eqref{eq:ystatec1} $P$-a.s., together with the state constraint \eqref{state:c} $P$-a.s.. Let $(\alpha^0,\alpha^1)$ be the selection, $\PP^\alpha$ the associated joint measure, and set $P^*:=P^{\PP^\alpha}$ both defined as in \eqref{eq:disintegration}, and  \eqref{eq:defpp}. Under $P^*$, using Lemma \ref{lem:filter}, we obtain that $(X,p)$ solves the coupled SDE
\[
dX_t=\sigma(t,X_t)\bigl(\bar\lambda\,dt+dB^*_t\bigr),\qquad dp_t=p_t(1-p_t)\,\Delta\lambda^\top dB^*_t,
\]
with initial condition $(X_0,p_0)$, for some $(P^*,\FF^X)$-Brownian motion $B^*$. The same coupled system holds under $P$ by hypothesis. Since $\bar\lambda$ is bounded, and $\sigma$ is non-degenerate and Lipschitz, Girsanov's Theorem together with the weak uniqueness for the equation \eqref{sde:sigma} forces the law of $(X,p)$ under $P$ to coincide with its law under $P^*$ on $\cF_T^X$. Hence $P=P^*$ on $\cF_T^X$, and in particular $P\sim\bar P$.

Now define $\widetilde Y^\theta:=Y^{\theta,0,X_0,y_\theta,Z^\theta}$ as the $\bar P$-solution of \eqref{fsde}. The Girsanov substitution $dB^{\bar P}_t=\bar\lambda\,dt+dB_t$ shows that $(\widetilde Y^0,\widetilde Y^1)$ satisfies \eqref{eq:ystatec1} under $P$ with initial conditions $(y_0,y_1)$. The SDE \eqref{eq:ystatec1} is linear in $Y^\theta$ with bounded coefficients, so its solution is unique in $\bS_2$. Therefore $\widetilde Y^\theta=Y^\theta$ both $P$- and $\bar P$-a.s. In particular,
\[
\underline W(t,X_t)\le \beta_1 Y^{0,0,X_0,y_0,Z^0}_t- \beta_0 Y^{1,0,X_0,y_1,Z^1}_t\le\overline W(t,X_t),\qquad t\in[0,T],\ \bar P\text{-a.s.,}
\]
which is condition (i) of Theorem~\ref{thm:domain}(b). Finally, the equivalence (i) and (iv) in Theorem \ref{thm:domain}(b) yields $(Z^0,Z^1)\in\cV (y_0,y_1)$, as required. \qed

\end{proof}
\subsection{Hamilton--Jacobi--Bellman equation for the state-constrained problem}\label{subsec:HJB}

The viscosity characterization of state-constrained optimal control originates with Soner~\cite{soner1986optimalI,soner1986optimalII} for first-order Hamilton--Jacobi equations and was developed for the second-order stochastic setting by Lasry--Lions~\cite{lasrylions1989}, Capuzzo-Dolcetta--Lions~\cite{capuzzodolcettalions1990}, Ishii--Loretti~\cite{ishii2002class} and Katsoulakis~\cite{katsoulakis1994representation}; see also Fleming--Soner~\cite[Sec.~IV.5]{flemingsoner2006}. As established in these works, the value function of a supremum-type problem with state constraint is characterized as the unique \emph{constrained viscosity solution} of the associated HJB equation, in the sense that the subsolution inequality is required on the closure of the constraint set $\mathrm{cl}_y(\cD):=\cD\cup\cD_d\cup\cD_u$, whereas the supersolution inequality is required only on its (relative) interior $\cD$. The asymmetry reflects the geometry of supremum problems at the lateral boundary, on which trajectories may be absorbed but cannot be pushed outward. We now derive this HJB equation.

\paragraph{State, controls and diffusion matrix.}
 By Theorem~\ref{thm:domain}, the set $\cV (y_0,y_1)$ defined in \eqref{eq:contz} admits the pathwise representation
\[
\cV (y_0,y_1)=\bigl\{(Z^0,Z^1)\in\cV^2: \underline W(s,X_s)\le \beta_1 Y^{0,0,X_0,y_0,Z^0}_s- \beta_0 Y^{1,0,X_0,y_1,Z^1}_s\le\overline W(s,X_s),\,\forall s\in[0,T]\bigr\}.
\]
The state of the principal's problem is $(t,x,y_0,y_1,p)\in\mathrm{cl}_y(\cD)$, and the controls are $(z^0,z^1,\alpha^0,\alpha^1)\in(\RR^d)^2\times A\times A$ subject to the incentive-compatibility constraint
\begin{equation}\label{eq:adm_set}\notag
\alpha^\theta\in A^*_\theta(t,x,z^\theta),\qquad \theta\in\{0,1\}.
\end{equation}
Recalling
\begin{align*}
\bar\lambda(t,x,p,\alpha^0,\alpha^1)&:=p\,\lambda(t,x,0,\alpha^0)+(1-p)\,\lambda(t,x,1,\alpha^1),\\
\Delta\lambda(t,x,\alpha^0,\alpha^1)&:=\lambda(t,x,0,\alpha^0)-\lambda(t,x,1,\alpha^1),
\end{align*}
the dynamics \eqref{eq:xstatec1}--\eqref{eq:ystatec1} are driven by the $d$-dimensional Brownian motion $B$ through the $(d+3)\times d$ diffusion matrix
\begin{equation}\label{eq:Sigma_full}\notag
\Sigma(t,x,p,z^0,z^1,\alpha^0,\alpha^1):=
\begin{pmatrix}
\sigma(t,x)\\[0.3em]
(z^0)^\top\sigma(t,x)\\[0.3em]
(z^1)^\top\sigma(t,x)\\[0.3em]
p(1-p)\,\Delta\lambda(t,x,\alpha^0,\alpha^1)^\top
\end{pmatrix},
\end{equation}
acting on the augmented state $(X,y^0,y^1,p)\in\RR^d\times\RR\times\RR\times[0,1]$, and by the drift vector
\begin{equation}\label{eq:drift_full}\notag
b(t,x,y_0,y_1,p,z^0,z^1,\alpha^0,\alpha^1):=
\begin{pmatrix}
\sigma(t,x)\bar\lambda(t,x,p,\alpha^0,\alpha^1)\\[0.3em]
-H^0(t,x,z^0)+\kappa(t,x)\,y_0+(z^0)^\top\sigma(t,x)\bar\lambda(t,x,p,\alpha^0,\alpha^1)\\[0.3em]
-H^1(t,x,z^1)+\kappa(t,x)\,y_1+(z^1)^\top\sigma(t,x)\bar\lambda(t,x,p,\alpha^0,\alpha^1)\\[0.3em]
0
\end{pmatrix}.
\end{equation}

\paragraph{Generator and Hamiltonian.}
For a smooth test function $\varphi:\mathrm{cl}_y(\cD)\to\RR$, with gradient $$\nabla\varphi=(\nabla_x\varphi,\partial_{y_0}\varphi,\partial_{y_1}\varphi,\partial_p\varphi)\in\RR^{d+3}$$ and Hessian $\nabla^2\varphi\in\bS_{d+3}$, define the controlled generator
\begin{equation}\label{eq:generator_full}
L^{z^0,z^1,\alpha^0,\alpha^1}\varphi
:=b\cdot\nabla\varphi+\tfrac12\,\mathrm{tr}\bigl(\Sigma\Sigma^\top\nabla^2\varphi\bigr),
\end{equation}
where $b$ and $\Sigma$ are evaluated at $(t,x,y_0,y_1,p,z^0,z^1,\alpha^0,\alpha^1)$. The associated Hamiltonian is
\begin{equation}\label{eq:H_full}\notag
H(t,x,y_0,y_1,p,q,M)
:=\inf_{(z^0,z^1)\in(\RR^d)^2}\ \inf_{(\alpha^0,\alpha^1)\in A^*_0(t,x,z^0)\times A^*_1(t,x,z^1)}\bigl(-L^{z^0,z^1,\alpha^0,\alpha^1}(t,x,y_0,y_1,p;q,M)\bigr),
\end{equation}
where on the right $L^{z^0,z^1,\alpha^0,\alpha^1}(t,x,y_0,y_1,p;q,M)$ denotes the right-hand side of \eqref{eq:generator_full} with $\nabla\varphi$ replaced by $q$ and $\nabla^2\varphi$ by $M$. The HJB equation associated to $V_{sc}$ is
\begin{equation}\label{eq:hjb_full}
-\partial_t V_{sc}(t,x,y_0,y_1,p)+H\bigl(t,x,y_0,y_1,p,\nabla V_{sc}(t,x,y_0,y_1,p),\nabla^2 V_{sc}(t,x,y_0,y_1,p)\bigr)=0.
\end{equation}
for $(t,x,y_0,y_1,p)\in \cD$ whereas we require 
\begin{equation}\label{eq:hjb_fullsub}
-\partial_t V_{sc}(t,x,y_0,y_1,p)+H\bigl(t,x,y_0,y_1,p,\nabla V_{sc}(t,x,y_0,y_1,p),\nabla^2 V_{sc}(t,x,y_0,y_1,p)\bigr)\leq0.
\end{equation}
for $(t,x,y_0,y_1,p)\in \mathrm{cl}_y(\cD)-\{t=T\}$.

 Since the slice of $\mathrm{cl}_y(\cD)$ at $t=T$ is $\{(T,x,y,y,p):x\in\RR^d,y\in\RR,p\in[0,1]\}$, on which the principal's terminal payoff is $U_p(x-y)$, the natural terminal condition is
\begin{equation}\label{eq:terminal_full}\notag
V_{sc}(T,x,y,\frac{\beta_1}{\beta_0}y,p)=U_p(\Gamma (x) - \beta^{-1}_0y),\quad x\in\RR^d,\ y\in\RR,\ p\in[0,1].
\end{equation}

Note that we do not specify the value of $V_{sc}$ at $p=0$ or $p=1$. This comes from the fact that the diffusion of $p$ degenerates to $0$ at $p=0$ and $p=1$ and $p$ in fact never reaches this boundary. Thus, we do not need a boundary value to have the wellposedness of $V_{sc}$.

Broadly, the literature offers two approaches to the well-posedness of the
associated Hamilton--Jacobi--Bellman equations. The first requires that at
every boundary point there exist controls pushing the state back into the
interior of the domain~\cite{soner1986optimalII,ishii2002class}; this
inward-pointing condition fails in our setting, since on the boundaries
$\cD_d$ and $\cD_u$ the admissible controls are determined by
Theorem~\ref{thm:domain}~(b) and cannot be chosen to drive
$(X_t,\beta_1 Y^0_t- \beta_0 Y^1_t)$ strictly inside $\cD$. The second approach, developed
in~\cite{capuzzodolcettalions1990} for problems of the form
\eqref{eq:Vsc}, circumvents the inward-pointing condition but requires a
priori continuity of the value function up to the lateral boundary. This
last property is far from automatic: as already observed in
\cite[p.~502]{katsoulakis1994representation}, the value function of a
state-constrained stochastic control problem need not be continuous in
general, and additional structural assumptions on the dynamics, the
running cost, or the geometry of the constraint set are needed to rule
out boundary discontinuities. In our setting, where the controls on
$\cD_d$ and $\cD_u$ are pinned down by the credible-set
characterisation of Theorem~\ref{thm:domain}~(b), we are not aware of a
direct argument that delivers such regularity, and we cannot readily
verify the continuity hypothesis required by
\cite{capuzzodolcettalions1990}.

To circumvent these difficulties, we follow a third route. Adapting the
arguments of~\cite{bouchard2010optimal}, we show that, under suitable
structural conditions on the problem, the state-constrained control
problem can be recast as a Dirichlet boundary value problem --- i.e.,
the principal receives a boundary utility when the state exits the
domain. Specifically, we assume that the data are independent of $X$,
the effort set is bounded, and the principal is risk-neutral. Under
these assumptions, the lateral boundaries $\cD_d$ and $\cD_u$ become
absorbing for the gap process $\beta_1 Y^0- \beta_0 Y^1$ in the sense of
Theorem~\ref{thm:domain}~(b), and the principal's continuation value
upon hitting either boundary is a function of $(t,x,p)$ alone. These
structural conditions also play the role of the additional regularity
hypotheses alluded to above: they are precisely what allows us to
establish, rather than postulate, the continuity of the value function
at the lateral boundary. The Dirichlet reformulation then falls within
the standard viscosity-solution framework, for which existence,
uniqueness, and stability follow from the comparison principles
of~\cite{flemingsoner2006,crandall1992user}. We carry out this reduction
in Section~\ref{sec:dirichlet}.
\section{Connection to screening contracts}\label{s:screening}

The main contribution of this section is to show that our 
stochastic-target/state-constrained-control framework provides a 
natural language for screening contracts, recasting the credible-set 
characterization of \cite{cvitanic2013dynamics} in the finite-horizon 
two-type setting.

Results of this type are known in the literature in various forms. 
Credible-set-type objects --- feasible sets of jointly implementable 
type-indexed continuation utilities --- appear in the persistent-shocks 
literature \cite{williams2011persistent,zhang2009dynamic,bloedel2025insurance,santibanez2020bank}, 
and the screening counterpart, formulated in terms of 
continuation/temptation value pairs, appears in 
\cite{cvitanic2013dynamics} for the infinite-horizon discounted 
problem. Our contribution is to give a rigorous derivation of 
the credible-set characterization in the finite-horizon screening 
setting via Theorem~\ref{thm:domain}, yielding the explicit strip 
representation 
$\underline W(0,X_0)\le \beta_1 y_0-\beta_0 y_1\le\overline W(0,X_0)$ in terms of 
the gap PDEs \eqref{eq:pdeuwk}--\eqref{eq:pdeowk}. This places the 
screening value $V_{p,s}$ in the same analytical framework as the 
single-contract values $V_{p,c}$ and $V_{p,uc}$ of earlier sections, 
supporting the comparison carried out in Section~\ref{sec:benchmark}.

Following \cite{cvitanic2013dynamics}, we introduce \emph{temptation 
values}: for a given pair of contracts 
$(\xi_\theta)_{\theta\in\{0,1\}}\in\mathcal{C}_a^2$, define, for each 
$\theta\in\{0,1\}$,
\[
y_\theta^c:=\sup_{\alpha\in\mathcal{A}^X} J_a(\alpha;\xi_{1-\theta},\theta),
\qquad
y_\theta:=\sup_{\alpha\in\mathcal{A}^X} J_a(\alpha;\xi_{\theta},\theta).
\]
Here $y_\theta^c$ is the value a type-$\theta$ agent obtains by 
deviating to the contract $\xi_{1-\theta}$ designed for the other 
type, while $y_\theta$ is the value she derives from her own 
contract $\xi_\theta$. The principal designs the menu so that each 
type's truthful value weakly exceeds her temptation value, ensuring 
that each agent chooses the contract intended for her 
(self-selection).

For notational convenience, we use the superscript ``$c$'' on a control process to indicate that the corresponding type's forward equation \eqref{fsde} is initialized at her \emph{temptation} value rather than her truthful value. Thus $Z^{1,c}$ governs the type-$1$ forward process initialized at $y^c_1$ (i.e.\ the path of type $1$'s continuation value when she deviates to $\xi_0$), and $Z^{0,c}$ governs the type-$0$ forward process initialized at $y^c_0$ (the path under deviation to $\xi_1$).

We denote by $\cD_4(X_0)$ the set of all quadruples $(y_0,y_1^c,y_0^c,y_1)\in\RR^4$ arising from some $(\xi_\theta)_{\theta\in\{0,1\}}\in\mathcal{C}_a^2$, that is,
\begin{align*}
\cD_4(X_0):=\Bigl\{(y_0,y_1^c,y_0^c,y_1)\in\RR^4:\ \exists (\xi_\theta)_{\theta\in\{0,1\}}\in\mathcal{C}_a^2\text{ such that, for all }\theta\in\{0,1\},\\
y_\theta^c=\sup_{\alpha\in\mathcal{A}^X} J_a(\alpha;\xi_{1-\theta},\theta),\quad y_\theta=\sup_{\alpha\in\mathcal{A}^X} J_a(\alpha;\xi_{\theta},\theta)\Bigr\}.
\end{align*}

For notational convenience we set
\[
y^{(0)}:=(y_0,y_1^c)\qquad\text{and}\qquad y^{(1)}:=(y_0^c,y_1),
\]
so that $y^{(0)}$ collects the values generated by $\xi_0$ for the two types (with type $0$ first, type $1$ second), and $y^{(1)}$ does the same for $\xi_1$. Theorem~\ref{thm:domain} immediately yields
\[
\cD_4(X_0)=\cE(X_0)\times\cE(X_0).
\]
The incentive-compatibility condition of \cite{cvitanic2013dynamics} (Definition~\ref{def:screening}) requires in addition that $(\xi_\theta)_{\theta\in\{0,1\}}\in\cS$, i.e.\ each type weakly prefers her own contract: $y_\theta\ge y_\theta^c$ for $\theta\in\{0,1\}$. Accordingly, define
\begin{align*}
\cD_\cS(X_0):=\Bigl\{(y_0,y_1^c,y_0^c,y_1)\in\RR^4:\ \exists (\xi_\theta)_{\theta\in\{0,1\}}\in\cS\text{ such that, for all }\theta\in\{0,1\},\\
y_\theta^c=\sup_{\alpha\in\mathcal{A}^X} J_a(\alpha;\xi_{1-\theta},\theta),\quad y_\theta=\sup_{\alpha\in\mathcal{A}^X} J_a(\alpha;\xi_{\theta},\theta)\Bigr\}.
\end{align*}
Then
\[
\cD_\cS(X_0)=\bigl(\cE(X_0)\times\cE(X_0)\bigr)\cap\bigl\{(y_0,y_1^c,y_0^c,y_1)\in\RR^4:\ y_\theta\ge y_\theta^c,\ \theta\in\{0,1\}\bigr\}.
\]

Consequently, Theorem~\ref{thm:domain} implies that, instead of optimizing over incentive-compatible screening contracts $(\xi_\theta)_{\theta\in\{0,1\}}\in\cS\subset\mathcal{C}_a^2$, the principal may optimize over $(y_0,y_1^c,y_0^c,y_1)\in\cE(X_0)\times\cE(X_0)$ subject to
\[
y_\theta\ge y_\theta^c,\quad
(Z^0,Z^{1,c})\in\cV(y_0,y_1^c),\quad
(Z^{0,c},Z^1)\in\cV(y_0^c,y_1).
\]
The contracts are then recovered by
\[
\xi_0=\frac{Y^{0,0,X_0,y_0,Z^0}_T}{\beta_0}=\frac{Y^{1,0,X_0,y_1^c,Z^{1,c}}_T}{\beta_1}
\quad\text{and}\quad
\xi_1=\frac{Y^{0,0,X_0,y_0^c,Z^{0,c}}_T}{\beta_0}=\frac{Y^{1,0,X_0,y_1,Z^1}_T}{\beta_1}.
\]

Under the screening contract $\xi_\theta$, the principal observes the type $\theta$ at time $t=0+$, so the subsequent contracting problem is type-by-type without any remaining hidden parameter. In particular, there is no filtering component, and the principal's expectation under $\xi_\theta$ is taken under $P^{\alpha^\theta,\theta}$ with $\alpha^\theta_t\in A^*_\theta(t,X_t,Z^\theta_t)$. Hence the principal's value under screening contracts \eqref{eq:Vp_screening} reads
\begin{align}\label{principal_val_screen}
V_{p,s}=\sup\Bigl\{\,
&p_0\,\EE^{P^{\alpha^0,0}}\!\bigl[U_p( \Gamma (X_T)-\beta^{-1}_0 Y^{0,0,X_0,y_0,Z^0}_T)\bigr]\notag\\
&+(1-p_0)\,\EE^{P^{\alpha^1,1}}\!\bigl[U_p( \Gamma (X_T)-\beta^{-1}_1 Y^{1,0,X_0,y_1,Z^1}_T)\bigr]\;:\notag\\
&(Z^0,Z^{1,c})\in\cV(y_0,y_1^c),\ \ (Z^{0,c},Z^1)\in\cV(y_0^c,y_1),\notag\\
&\alpha^\theta_t\in A^*_\theta(t,X_t,Z^\theta_t),\ \ \theta\in\{0,1\},\notag\\
&\underline W(0,X_0)\le \beta_1 y_0- \beta_0 y_1^c\le\overline W(0,X_0),\notag\\
&\underline W(0,X_0)\le \beta_1 y_0^c-\beta_0 y_1\le\overline W(0,X_0),\notag\\
&y_\theta\ge\max\{y_\theta^c,R_\theta\},\ \ \theta\in\{0,1\}\,\Bigr\}.
\end{align}

Owing to the similarity between the two terms, we present the methodology through the principal's control problem associated with $\xi_0$, which depends only on $y^{(0)}=(y_0,y_1^c)$ and leads to the value function $V_0$ we describe below.

Unlike in Section \ref{sec:dirichlet} where we do not assume the state constraint control problem to have a smooth value, in this section, we aim to only provide verification theorem therefore we directly assume that the value function is smooth. Given that the Theorem \ref{thm:domain}, if $\beta_1 Y^0_t- \beta_0 Y^{1,c}_t$ and $\beta_1 Y^{0,c}_t- \beta_0 Y^{c}_t$ touch their respective boundaries, they stay on this boundary and the state constraint problem at the boundary is a Dirichlet boundary value problem. We first characterize the boundary of this problem. 

\paragraph{Boundary value functions.} For each $(t,x,y)\in[0,T]\times\RR^d\times\RR$, we define the value functions on the boundaries of the credible region: the upper-boundary value
\begin{align*}
\overline V_0\!\left(t,x,\tfrac{\overline W(t,x)+\beta_0 y}{\beta_1},y\right)&:=\sup\Bigl\{\EE^{P^{\alpha^0,0}}\!\bigl[U_p( \Gamma (X_T)-\beta^{-1}_0Y^{0,t,x,(\overline W(t,x)+\beta_0 y)/\beta_1,(\beta_0 Z^{1,c}+\overline Z)/\beta_1}_T)\mid X_t=x\bigr]:\notag\\
&\qquad Z^{1,c}_\cdot\in\overline\cV(\cdot,X_\cdot,\overline Z_\cdot),\ \alpha^0_t\in A^*_0\!\left(t,X_t,\tfrac{\beta_0 Z^{1,c}_t+\overline Z_t}{\beta_1}\right),\\
&\qquad \alpha^1_t\in A^*_1(t,X_t,Z^{1,c}_t)\Bigr\},\notag
\end{align*}
and the lower-boundary value
\begin{align*}
\underline V_0\!\left(t,x,\tfrac{\underline W(t,x)+\beta_0 y}{\beta_1},y\right)&:=\sup\Bigl\{\EE^{P^{\alpha^0,0}}\!\bigl[U_p( \Gamma (X_T)-\beta^{-1}_0 Y^{0,t,x,(\underline W(t,x)+\beta_0 y)/\beta_1,(\beta_0 Z^{1,c}+\underline Z)/\beta_1}_T)\mid X_t=x\bigr]:\notag\\
&\qquad Z^{1,c}_\cdot\in\underline\cV(\cdot,X_\cdot,\underline Z_\cdot),\ \alpha^0_t\in A^*_0\!\left(t,X_t,\tfrac{\beta_0 Z^{1,c}_t+\underline Z_t}{\beta_1}\right),\\
&\qquad \alpha^1_t\in A^*_1(t,X_t,Z^{1,c}_t)\Bigr\}.\notag
\end{align*}
These represent the principal's continuation values once the gap $\beta_1 Y^0-\beta_0Y^1$ reaches one of the boundaries of the credible region, evaluated under the type-$0$ measure since type $0$ is the type who actually accepts $\xi_0$ at equilibrium.

\paragraph{Interior exit-time problem.} In the interior of the spacetime strip $\mathrm{cl}_y(\cD)$, define
\begin{align*}
\underline\tau^0&:=\inf\bigl\{s\ge t:\beta_1 Y^{0,t,x,y_0,Z^0}_s-\beta_0 Y^{1,t,x,y_1^c,Z^{1,c}}_s=\underline W(s,X_s)\bigr\},\\
\overline\tau^0&:=\inf\bigl\{s\ge t:\beta_1 Y^{0,t,x,y_0,Z^0}_s-\beta_0 Y^{1,t,x,y_1^c,Z^{1,c}}_s=\overline W(s,X_s)\bigr\},
\end{align*}
and $\tau^0:=\underline\tau^0\wedge\overline\tau^0\le T$, with $\tau^0=T$ on the event that the trajectory remains strictly inside the strip on $[t,T)$ (using $\underline W(T,\cdot)=\overline W(T,\cdot)=0$). The interior control problem is the exit-time problem
\begin{align*}
V_0(t,x,y_0,y_1^c)&:=\sup_{(Z^0,Z^{1,c})}\sup_{\alpha^\theta_\cdot\in A^*_\theta(\cdot,X_\cdot,Z^\theta_\cdot)}\Bigl(\EE^{P^{\alpha^0,0}}\!\bigl[\overline V_0(\tau^0,X_{\tau^0},Y^0_{\tau^0},Y^1_{\tau^0})\mathbf{1}_{\{\overline\tau^0\le\underline\tau^0\}}\mid X_t=x\bigr]\notag\\
&\qquad+\EE^{P^{\alpha^0,0}}\!\bigl[\underline V_0(\tau^0,X_{\tau^0},Y^0_{\tau^0},Y^1_{\tau^0})\mathbf{1}_{\{\overline\tau^0>\underline\tau^0\}}\mid X_t=x\bigr]\Bigr),
\end{align*}
where by convention $Z^0$ is the control attached to type $0$ under the truthful contract $\xi_0$ and $Z^{1,c}$ is the control attached to the temptation contract of type $1$ deviating to $\xi_0$. This is a standard stochastic optimal control problem for the state process $(X_s,Y^0_s,Y^1_s)$.
\paragraph{HJB formulation.} Since the principal learns the type at $t=0+$ under a screening contract, the dynamics under $(\FF^X,P^{\alpha^0,0})$ (i.e.\ conditional on $\{\Theta=0\}$) are, by \eqref{eq:X_dynamics_small} and \eqref{fsde},
\begin{align*}
dX_t&=\sigma(t,X_t)\bigl(\lambda(t,X_t,0,\alpha^0_t)\,dt+dB^{P^{\alpha^0,0}}_t\bigr),\\
dY^0_t&=\Bigl(-H^0(t,X_t,Z^0_t)+\kappa(t,X_t)Y^0_t+(Z^0_t)^\top\sigma(t,X_t)\lambda(t,X_t,0,\alpha^0_t)\Bigr)dt\\
&\quad+(Z^0_t)^\top\sigma(t,X_t)\,dB^{P^{\alpha^0,0}}_t,\\
dY^1_t&=\Bigl(-H^1(t,X_t,Z^{1,c}_t)+\kappa(t,X_t)Y^1_t+(Z^{1,c}_t)^\top\sigma(t,X_t)\lambda(t,X_t,0,\alpha^0_t)\Bigr)dt\\
&\quad+(Z^{1,c}_t)^\top\sigma(t,X_t)\,dB^{P^{\alpha^0,0}}_t.
\end{align*}

For a smooth test function $\varphi$ on $[0,T]\times\RR^d\times\RR^2$, write its gradient as
\[
D\varphi=(D_x\varphi,\pa_{y_0}\varphi,\pa_{y_1}\varphi),
\]
with Hessian $D^2\varphi$ a $(d+2)\times(d+2)$ block matrix whose blocks are $\pa_{xx}\varphi\in\RR^{d\times d}$, $\pa_{xy_i}\varphi\in\RR^d$, $\pa_{y_iy_j}\varphi\in\RR$. Writing $\lambda^0:=\lambda(t,x,0,\alpha^0)$, define the generator
\begin{align*}
\cL_{s_0}^{z_0,z_1,\alpha^0,\alpha^1}\varphi&:=\langle\sigma(t,x)\lambda^0,D_x\varphi\rangle+\sum_{\theta\in\{0,1\}}\Bigl(-H^\theta(t,x,z_\theta)+\kappa(t,x)y_\theta+z_\theta^\top\sigma(t,x)\lambda^0\Bigr)\pa_{y_\theta}\varphi\\
&\quad+\tfrac12\|\sigma(t,x)^\top z_0\|^2\pa_{y_0y_0}\varphi+\tfrac12\|\sigma(t,x)^\top z_1\|^2\pa_{y_1y_1}\varphi+z_0^\top\sigma(t,x)\sigma(t,x)^\top z_1\,\pa_{y_0y_1}\varphi\\
&\quad+\tfrac12\Tr\bigl(\sigma(t,x)\sigma(t,x)^\top\pa_{xx}\varphi\bigr)+z_0^\top\sigma(t,x)\sigma(t,x)^\top\pa_{xy_0}\varphi+z_1^\top\sigma(t,x)\sigma(t,x)^\top\pa_{xy_1}\varphi,
\end{align*}
and the Hamiltonian
\[
\cH_{s_0}(t,x,y_0,y_1;D\varphi,D^2\varphi):=\sup_{z_0,z_1\in\RR^d}\ \sup_{\substack{\alpha^0\in A^*_0(t,x,z_0)\\ \alpha^1\in A^*_1(t,x,z_1)}}\cL_{s_0}^{z_0,z_1,\alpha^0,\alpha^1}\varphi(t,x,y_0,y_1).
\]

\begin{theorem}[HJB equation for screening problems and comparison with single contract case]\label{thm:HJB_principal_screens}
Assume Assumption~\ref{ass:boundary} holds. Then,
\begin{align}\label{eq:comparesc}
    V_{p,s}\geq V_{p,c}.
\end{align}

For each $\theta\in\{0,1\}$, denote by $\cD^\theta$, $\cD_d^\theta$, $\cD_u^\theta$ the analogues of $\cD$, $\cD_d$, $\cD_u$ for the coordinates $y^{(\theta)}$. Assume also that the HJB equation
\begin{equation*}\label{eq:HJB_principal_screen_n}
-\pa_t V_\theta(t,x,y^{(\theta)})-\cH_{s_\theta}\!\bigl(t,x,y^{(\theta)};DV_\theta(t,x,y^{(\theta)}),D^2V_\theta(t,x,y^{(\theta)})\bigr)=0\qquad\text{on }\cD^\theta,
\end{equation*}
with boundary conditions
\begin{align*}
V_\theta(t,x,y^{(\theta)})&=\underline V_\theta(t,x,y^{(\theta)}),&&\text{on }\cD_d^\theta,\\
V_\theta(t,x,y^{(\theta)})&=\overline V_\theta(t,x,y^{(\theta)}),&&\text{on }\cD_u^\theta,
\end{align*}
admits a classical solution $V_\theta\in C^{1,2}(\cD^\theta)\cap C(\overline{\cD^\theta})$. Then the principal's screening value \eqref{principal_val_screen} satisfies
\begin{align*}
V_{p,s}
=
\sup\Bigg\{
&\,p_0 V_0(0,X_0,y_0,y_1^c)
 +(1-p_0)V_1(0,X_0,y_0^c,y_1) :
 (y_0,y_1^c,y_0^c,y_1)\in\RR^4, \\
&\, y_0\geq \max\{y_0^c,R_0\},\;
   y_1\geq \max\{y_1^c,R_1\},\\
&\, \underline W(0,X_0)
   \leq \beta_1 y_0-\beta_0 y_1^c
   \leq \overline W(0,X_0),\\
&\, \underline W(0,X_0)
   \leq \beta_1 y_0^c-\beta_0 y_1
   \leq \overline W(0,X_0)
\Bigg\}.
\end{align*}

\end{theorem}
\begin{remark}
     Similarly to the Section \ref{sec:dirichlet}, in particular cases, $V_0,V_1$ can be characterized as value functions of optimal control problems with Dirichlet boundary value conditions. 
\end{remark}

\begin{proof}
    Given the smoothness asssumption on $V_\theta$, the proof of the representation of $V_{p,s}$ is the standard verification theorem of stochastic optimal control. Thus, we only prove \eqref{eq:comparesc}. 
Fix $y_\theta \geq R_\theta$ for $\theta \in \{0,1\}$ with $(y_0, y_1) \in \mathcal{E}(X_0)$ and $(Z^0, Z^1) \in \mathcal{V}(y_0, y_1)$, and consider the contract
\[
\xi := \frac{Y^{0,0,X_0,y_0,Z^0}_T}{\beta_0} = \frac{Y^{1,0,X_0,y_1,Z^1}_T}{\beta_1},\; \bar{P}\text{-a.s.}
\]
By Proposition~\ref{prop:principal_target}, $\sup_{\alpha\in\mathcal{A}^X} J_a(\alpha;\xi,\theta) = y_\theta \ge R_\theta$ for $\theta\in\{0,1\}$. Thus the pair $(\xi_0,\xi_1):=(\xi,\xi)\in\mathcal{S}$, since the incentive compatibility constraint $y_\theta \ge y_\theta^c$ holds trivially with equality (both contracts are identical). Moreover, the principal's payoff under the menu $(\xi,\xi)$ coincides with her payoff under the single contract $\xi$. Taking the supremum over admissible $(y_0,y_1,Z^0,Z^1)$ yields
\[
V_{p,s} \geq V_{p,c}. \qed
\]
\end{proof}

\section{A third reduction in particular cases}\label{sec:dirichlet}

Unlike the Section \ref{s:screening}, in this section, we do not assume that the value function of the state constraint problem is continuous but prove this point. Indeed, we show that by adapting the proofs in \cite{bouchard2010optimal}, the state-constrained control problem in Theorem \ref{thm:sc_rep} can be recast as a Dirichlet boundary value problem --- i.e., the principal receives a boundary utility when the state exits the domain --- provided the data of the problem are independent of $X$, the effort set is bounded and the principal is risk-neutral.

\begin{ass}\label{assumption:cd}
The volatility is normalized to $
    \sigma(t,x)=1,$ $U_p(x)=\Gamma(x)=x$, $\beta_0=\beta_1=1$
and the output drift satisfies\[
    \lambda(s,x,\theta,\alpha)=\alpha,
    \qquad \alpha\in A:=[a_{\min},a_{\max}].\]
Moreover, for each
$\theta\in\{0,1\}$, the cost function satisfies
$
    c(\theta,\cdot)\in C^2(A),
$
and there exists $\rho>0$ such that
\[
    \partial_{\alpha\alpha}c(\theta,\alpha)\geq \rho,
    \qquad (\theta,\alpha)\in\{0,1\}\times A.
\]
Finally, the function
$c(0,\cdot)-c(1,\cdot)$
is non-constant on $A$.
\end{ass}

Throughout this section the simplification $\beta_0=\beta_1=1$ is in force; the general case extends with only notational adjustments to the coefficients in the change of variables below.

Under these hypotheses, \eqref{def:hk} reduces to
\begin{align*}
    H^\theta(z):=\sup_{\a\in A}\{\a z  -c(\theta,\a) \},
\end{align*}
and for each $z\in\RR$ by strong convexity of $c$ we have a Lipschitz continuous maximizer \[A^\theta(z)=\argmax_{\alpha\in A}\{\alpha z-c(\theta,\alpha)\}.\]
Setting
\begin{align}\label{eq:n0}
   N_0 := \max\bigl\{|c'(\theta,\alpha)|+|\a|:\theta\in\{0,1\},\,\alpha\in A\bigr\},
 \end{align}

the map $z_1\mapsto H^0(z_1+z)-H^1(z_1)$ is piecewise affine outside a bounded region, constant on each half-line $(-\infty,-N_0-|z|]$ and $[N_0+|z|,+\infty)$. Consequently, for every $z\in\RR$,
\[
\inf_{z_1\in\RR}\bigl[H^0(z_1+z)-H^1(z_1)\bigr]=\inf_{z_1\in K(z)}\bigl[H^0(z_1+z)-H^1(z_1)\bigr]
\]
with $K(z):=[-N_0-|z|,N_0+|z|]$ compact, and similarly for the supremum. In particular, both are finite.

Moreover, for each fixed $z_1\in\RR$, the map $z\mapsto H^0(z_1+z)-H^1(z_1)$ is Lipschitz with constant $N_0$, since $H^0$ is a supremum of affine functions with slopes in $A$. As the infimum (resp.\ supremum) of a family of $L$-Lipschitz functions is $L$-Lipschitz whenever finite,
\[
z\mapsto\inf_{z_1\in\RR}\bigl[H^0(z_1+z)-H^1(z_1)\bigr]
\qquad\text{and}\qquad
z\mapsto\sup_{z_1\in\RR}\bigl[H^0(z_1+z)-H^1(z_1)\bigr]
\]
are $N_0$-Lipschitz on $\RR$. Hence $\underline H$ and $\overline H$ are Lipschitz in $(y,z)$ with constants depending only on $\|\kappa\|_\infty$ and $N_0$, and the PDEs \eqref{eq:pdeuwk}--\eqref{eq:pdeowk} admit unique solutions, which are independent of $x$ by the $x$-independence of the data.

Setting
\begin{equation}\label{eq:infw}
\underline a:=\inf_{z\in\RR}\bigl[H^0(z)-H^1(z)\bigr]=\underline H(0,0)<
\overline a:=\sup_{z\in\RR}\bigl[H^0(z)-H^1(z)\bigr]=\overline H(0,0),
\end{equation}
the solutions $(\underline W,\overline W)$ of \eqref{eq:pdeuwk}--\eqref{eq:pdeowk} are $X$ independent and given by the ODE system
\begin{equation}\label{def:w}
\begin{cases}
-\underline W'(t)+\kappa\underline W(t)=\underline a,& \underline W(T)=0,\\[2pt]
-\overline W'(t)+\kappa\overline W(t)=\overline a,& \overline W(T)=0,
\end{cases}
\end{equation}
so that $\underline W(t)= \frac{\underline a}{\kappa} \left( 1 - e^{-\kappa(T-t)} \right)<\overline W(t)=\frac{\overline a}{\kappa} \left( 1 - e^{-\kappa(T-t)} \right)$ for all $t\in[0,T)$.

We further introduce the drift and diffusion of the filter--promise-gap system which is $x$ independent
\begin{align}
\bar\lambda(p,z^0,z^1)&:=p\,\alpha^*_0(z^0)+(1-p)\,\alpha^*_1(z^1),\notag\\
\Sigma(p,z^0,z^1)&:=\bigl(z^0-z^1,\;p(1-p)(\alpha^*_0(z^0)-\alpha^*_1(z^1))\bigr)^\top,.\label{eq:defSigma}
\end{align}
Since $\argmax_{\alpha\in A}\{\alpha z-c(\theta,\alpha)\}$ is single valued, the generator only depends on $z^0,z^1$ and is
\begin{align}\label{eq:defLinside}
L^{z^0,z^1}(t,y,p;q,N)
&:=\Bigl[-H^0(z^0)+H^1(z^1)+\kappa y+\bar\lambda(p,z^0,z^1)(z^0-z^1)\Bigr]q\notag\\
&\quad+\tfrac12\,\mathrm{tr}\bigl(\Sigma\Sigma^\top N\bigr)+\ell(t,z^0,z^1,p),
\end{align}
and the running reward is
\begin{equation}\label{eq:effective_running_reward}
\ell(s,z^0,z^1,p):=\bar\lambda(p,z^0,z^1)+\frac{e^{\kappa(T-s)}}{2}\Bigl[H^0(z^0)+H^1(z^1)-\bar\lambda(p,z^0,z^1)(z^0+z^1)\Bigr].
\end{equation}
In this simple setting the state dynamics \eqref{eq:xstatec}-\eqref{eq:ystatec} is
\begin{align*}
dX_t&=\bar\lambda(p_t,Z^0_t,Z^1_t)dt+dB^{\mathbb{P}}_t\\
dp_t &= p_t(1-p_t)\bigl(A^0(Z^0_t)-A^1(Z^1_t)\bigr)\,dB^{\mathbb{P}}_t,\\
dY^0_t &= \Bigl(-H^0(Z^0_t)+\kappa Y^0_t+
  \bar\lambda(p_t,Z^0_t,Z^1_t) Z^0_t\Bigr)dt 
  + Z^0_t\,dB^{\mathbb{P}}_t,\\
dY^1_t &= \Bigl(-H^1(Z^1_t)+\kappa Y^1_t+
 \bar\lambda(p_t,Z^0_t,Z^1_t)Z^1_t\Bigr)dt 
  + Z^1_t\,dB^{\mathbb{P}}_t
\end{align*}
and the state constraint optimal control problem \eqref{eq:Vsc} is
\begin{align}
    &V_{sc}(t,x,y^0,y^1,p):=x+\sup_{(Z^0,Z^1)\in\cV (t,y^0-y^1)} \left\{ \EE^{\PP}\left[\int_t^T\bar\lambda(p_s,Z^0_s,Z^1_s)ds - Y^0_T\right]\right\}\notag
\end{align}
where     
$$\cV(t,y^0-y^1):=\{(Z^0,Z^1)\in \cV^2:\underline W(s)\leq Y^0_s-Y^1_s\leq \overline W(s),\,\forall s\in[t,T]\}$$
which depends only on $y^0-y^1$ due to the linearity of the dynamics in $Y^0,Y^1$. Then the value of the principal can be computed using Theorem \ref{thm:sc_rep}.

\textbf{Change of Variables} We now consider the change of variable $(Y,S)=(Y^0-Y^1,Y^0+Y^1)$ and note that the dynamics of $Y=Y^0-Y^1$ does not depend on $Y^0+Y^1$. Thus, we make the following ansatz on the value function
\begin{equation*}
V_{sc}(t,x,y^0,y^1,p) = x-\frac{e^{\kappa(T-t)}}{2}\,(y^0+y^1) + w(t,y^0-y^1,p).
\end{equation*}
The $Y^0+Y^1$-dependent terms cancel identically, and in the sense of \eqref{eq:hjb_full}-\eqref{eq:hjb_fullsub}, $w$ satisfies
\begin{equation}\label{eq:hjb}
-\partial_t w + H(t,y,p,\nabla w,\nabla^2 w)=0,
\end{equation}
where for $(t,y,p, q ,A)\in \mathrm{cl}(\mathcal{D})\times \RR\times \cS_2$ define
\begin{equation}\label{eq:defHinside}
H(t,y,p, q ,A) 
:= \inf_{(z^0,z^1)\in\RR^2}\bigl\{-L^{z^0,z^1}(t,y,p, q ,A)\bigr\}.
\end{equation}

\begin{lemma}\label{lem:verify_cd}
Let Assumption \ref{assumption:cd} holds, then the following properties hold.
\begin{enumerate}[leftmargin=*]
    \item The functions $H^0,H^1,A^0,A^1:\mathbb{R}\to\mathbb{R}$ are
    Lipschitz continuous. Moreover, $A^0$ and $A^1$ are bounded, and
    \[
        \|A^0\|_\infty\vee\|A^1\|_\infty\leq N_0.
    \]
    In addition, $H^0$ and $H^1$ satisfy \eqref{eq:infw}.

    \item There exists a constant $C\geq 0$ such that, for every
    $(p,Z^0,Z^1)\in[0,1]\times\mathbb{R}^2$,
    \begin{align*}
        &\bigl|H^1(Z^1)-H^0(Z^0)\bigr|
        +\Bigl|
            H^0(Z^0)+H^1(Z^1)
            -\bar\lambda(p,Z^0,Z^1)(Z^0+Z^1)
        \Bigr|  \\
        &\hspace{4cm}
        \leq C\bigl(1+|Z^0-Z^1|\bigr).
    \end{align*}
    If
    $
        |Z^0+Z^1|\geq C\bigl(1+|Z^0-Z^1|\bigr),
    $
    then $\mathcal{L}^{Z^0,Z^1}(t,0,p;q,N)$ is independent of $p$.

    \item Let $c\in\{\underline a,\overline a\}$ and let
    $(z^n)_{n\geq 1}\subset\mathbb{R}$ be such that
    \[
        H^0(z^n)-H^1(z^n)\longrightarrow c.
    \]
    Then there exist a subsequence $(z^{n_k})_{k\geq 1}$ and a sequence
    $(\tilde z^k)_{k\geq 1}\subset\mathbb{R}$ such that
    \[
        H^0(\tilde z^k)-H^1(\tilde z^k)=c,
        \qquad
        z^{n_k}-\tilde z^k\longrightarrow 0.
    \]
\end{enumerate}
\end{lemma}

\subsection{Viscosity characterization of the value function}

In the $x-$independent setting of this section, we drop the dependence of the domains in $x$ and denote
\begin{align*}
    \mathcal{D}&:=\{(t,y,p)\in [0,T)\times\mathbb{R}\times (0,1):\underline{W}(t)<y<\overline{W}(t)\},\\
    \mathcal{D}_d&:=\{(t,y,p)\in [0,T)\times\mathbb{R}\times (0,1):y=\underline{W}(t)\},\\
    \mathcal{D}_u&:=\{(t,y,p)\in [0,T)\times\mathbb{R}\times (0,1):y=\overline{W}(t)\},
\end{align*}
and $\mathrm{cl}_y(\mathcal{D})=\mathcal{D}\cup\mathcal{D}_d\cup\mathcal{D}_u\cup (\{T\}\times \{0\}\times (0,1))$.
For a function $u:\mathcal{D}\to\mathbb{R}$, its upper and lower semicontinuous
envelopes on $\mathrm{cl}(\mathcal{D})$ are defined as limits from the interior:
\begin{align*}
    u^*(t,y,p)&:=\limsup_{\substack{(s,\gamma,q)\to(t,y,p)\\(s,\gamma,q)\in\mathcal{D}}} u(s,\gamma,q),\\[4pt]
    u_*(t,y,p)&:=\liminf_{\substack{(s,\gamma,q)\to(t,y,p)\\(s,\gamma,q)\in\mathcal{D}}} u(s,\gamma,q),
\end{align*}
for $(t,y,p)\in\mathrm{cl}(\mathcal{D})$, so that $u^*:\mathrm{cl}(\mathcal{D})\to\mathbb{R}$
is upper semicontinuous, $u_*:\mathrm{cl}(\mathcal{D})\to\mathbb{R}$ is lower semicontinuous,
and $u_*\leq u\leq u^*$ on $\mathcal{D}$.

We have the following a priori bound on the value function which is a consequence of a verification theorem. 
\begin{proposition}\label{prop:supersol}
Define
\begin{align*}
\overline C&=\bigl((C+N_0)^2-2\kappa\bigr)_+(\overline a+|\underline a|)^2T^2+\Bigl[2C+\tfrac{(C+N_0)C}{2}e^{\kappa T}\Bigr](\overline a+|\underline a|)T\notag\\
&\quad+\tfrac{C}{2}e^{\kappa T}+\tfrac{C^2}{16}e^{2\kappa T}+N_0+1\\
\underline C&=-\Bigl[2\kappa(\overline a+|\underline a|)^2T^2+2C(\overline a+|\underline a|)T+\tfrac{e^{\kappa T}C}{2}+N_0+1\Bigr]\notag
\end{align*}
where $C,N_0$ are the constants in Lemma \ref{lem:verify_cd}  and \eqref{eq:n0} and define $\phi,\psi:[0,T]\times\RR\to\RR$ by
\[
\phi(t,y):=\overline C(T-t)-y^2,\qquad\psi(t,y):=\underline C(T-t)-y^2.
\]
Then,
\begin{align}\label{eq:supersol-ineq}
-\partial_t\phi(t,y)+H\bigl(t,y,p,\partial_y\phi(t,y),\partial_{yy}\phi(t,y)\bigr)&\;\ge\;0\\
-\partial_t\psi(t,y)+H\bigl(t,y,p,\partial_y\psi(t,y),\partial_{yy}\psi(t,y)\bigr)&\;\le\;0\notag
\end{align}
and
\begin{equation}\label{eq:prioribound}
\underline C(T-t)-y^2\;\le\;w(t,y,p)\;\le\;\overline C(T-t)-y^2
\end{equation}
for every $(t,y,p)\in\mathrm{cl}_y(\mathcal{D})$.
\end{proposition}

Since \eqref{eq:prioribound} shows that all viscosity solution properties below has terminal condition $0$ at $t=T$.




\subsubsection{Viscosity solution properties}

Denote the lower and upper semicontinuous envelopes:
\begin{align*}
H_*(t,y,p, q ,A)
&= \liminf_{\substack{(\tilde t,\tilde y ,\tilde p,
  \tilde  q ,\tilde A)\\
  \to\,(t, y ,p, q ,A)}}
  H(\tilde t,\tilde y ,\tilde p,\tilde  q ,\tilde A),\\
H^*(t, y ,p, q ,A)
&= \limsup_{\substack{(\tilde t,\tilde y ,\tilde p,
  \tilde  q ,\tilde A)\\
  \to\,(t, y ,p, q ,A)}}
  H(\tilde t,\tilde y ,\tilde p,\tilde  q ,\tilde A).
\end{align*}

\begin{definition}[State-constraint viscosity solution]\label{eq:defstatc}
\noindent A locally bounded function $v:\mathcal{D}\to\mathbb{R}$ is a 
\textbf{viscosity supersolution} of \eqref{eq:hjb} if, for every test function 
$\phi\in C^{1,2}(\mathrm{cl}_y(\mathcal{D}))$ and every local minimum point 
$(t_0,y_0,p_0)\in\mathcal{D}$ of $v_*-\phi$:
\begin{equation}\label{eq:super_int}\notag
-\partial_t\phi(t_0,y_0,p_0)
+H^*\!\left(t_0,y_0,p_0,\nabla\phi,\nabla^2\phi\right)
\;\ge\; 0.
\end{equation}
\noindent A locally bounded function $u:\mathcal{D}\to\mathbb{R}$ is a 
\textbf{viscosity subsolution} of \eqref{eq:hjb} if, for every test function 
$\phi\in C^{1,2}(\mathrm{cl}_y(\mathcal{D}))$ and every local maximum point 
$(t_0,y_0,p_0)\in\mathrm{cl}_y(\mathcal{D})\cap\{t_0<T\}$ of $u^*-\phi$ (relative to $\mathrm{cl}_y(\mathcal{D})$):
\begin{equation}\label{eq:sub}\notag
-\partial_t\phi(t_0,y_0,p_0)
+H_*\!\left(t_0,y_0,p_0,\nabla\phi,\nabla^2\phi\right)
\;\le\; 0.
\end{equation}
\noindent A locally bounded function $w:\mathcal{D}\to\mathbb{R}$ is a 
\textbf{viscosity solution} if it is simultaneously a viscosity supersolution 
and a viscosity subsolution.
\end{definition}

The proof of the following results can be proven as in \cite{bouchard2010optimal}.  
\begin{theorem}[Viscosity property]\label{thm:viscosity}
Let Assumption~\ref{assumption:cd} hold. Then:
\begin{enumerate}[leftmargin=*]
  \item \textbf{Supersolution property.}
    The function $w_*$ is a viscosity supersolution of \eqref{eq:hjb} in $\mathcal{D}$.

  \item \textbf{Subsolution property.}
    The function $w^*$ is a viscosity subsolution of \eqref{eq:hjb} on
    $\mathrm{cl}_y(\mathcal{D})\cap\{t<T\}$.
\end{enumerate}
\end{theorem}
Thanks to Proposition \ref{prop:supersol} and the fact that $\overline W(T)=\underline W(T)=0$, the terminal condtition for $w$ is $w^*(T,0,p)=w_*(T,0,p)=0.$

We now provide the following DPP that can be proven as in \cite{bouchard2012weak}. 
\begin{proposition}[Weak Dynamic Programming Principle for $w$]\label{prop:DPP_w}
Let $(t,y,p) \in \mathcal{D}$. For any stopping time $\theta$ taking values in $[t,T]$, the value function $w(t,y,p)$ satisfies:
\begin{enumerate}[label=\roman*), leftmargin=*]
    \item \textbf{Upper bound:}
    \begin{equation}\label{eq:DPP_w_upper}\notag
        w(t,y,p) \leq \sup_{(Z^0,Z^1)\in\cV (t,y)} \EE^{\PP}\left[ \int_t^\theta \ell(s, Z^0_s, Z^1_s, p_s)\,ds + w^*(\theta,Y_\theta,p_\theta) \right].
    \end{equation}
    
    \item \textbf{Lower bound:}
    \begin{equation}\label{eq:DPP_w_lower}
        w(t,y,p) \geq \sup_{(Z^0,Z^1)\in\cV (t,y)} \EE^{\PP}\left[ \int_t^\theta \ell(s, Z^0_s, Z^1_s, p_s)\,ds + w_*(\theta,Y_\theta,p_\theta) \right].
    \end{equation}
\end{enumerate}
where $\ell$ is defined at \eqref{eq:effective_running_reward}.
\end{proposition}

\subsubsection{Value at the boundary}
Following the methodology of \cite{bouchard2010optimal}, we define the functions
\begin{align*}
    \overline w^*(t,p):=w^*(t,\overline W_t,p),\,\overline w_*(t,p):=w_*(t,\overline W_t,p)\\
    \underline w^*(t,p):=w^*(t,\underline W_t,p),\, \underline w_*(t,p):=w_*(t,\underline W_t,p)    
\end{align*}
and the optimal controls in \eqref{eq:defcv} which do not depend on $(t,x,p)$ and are give by
{\begin{align}\label{eq:defv}
     \overline \cV:=\{\overline z\in \RR: H_0(\overline z)-H_1(\overline z)=\overline a\},\,\underline \cV:=\{\underline z\in \RR: H_0(\underline z)-H_1(\underline z)=\underline a\}.
\end{align}
In order to characterize $\overline w^*$ and  $\underline w^*$, we define the operators
\begin{align}\label{eq:defuplowpde}
     \overline H(t,p,\gamma)&=\inf_{z\in \overline \cV}\, -L^{z,z}(t,0,\,p, \gamma)\\
      \underline H(t,p,\gamma)&=\inf_{z\in \underline \cV}\, -L^{z,z}(t,0,\,p, \gamma)\notag
\end{align}
and recall the standard viscosity property for a PDE of the form 
\begin{equation}\label{eq:test}
-\partial_t\phi(t,p)
+H\!\left(t,\,p,\,\phi_{pp}(t,p)\right)
=0.
\end{equation}

\begin{definition}[Bounded standard viscosity solution]\label{visco:classic}
\noindent A bounded function $u:[0,T]\times (0,1)\to\mathbb{R}$ is a 
\textbf{standard viscosity supersolution} of \eqref{eq:test} if, for every test function 
$\phi\in C^{1,2}([0,T]\times (0,1))$ and every local minimum point 
$(t_0,p_0)\in [0,T)\times (0,1)$ of $u_*-\phi$:
\begin{equation}\label{eq:super}\notag
-\partial_t\phi(t_0,p_0)
+{H}^*\!\left(t_0,\,p_0,\,\phi_{pp}(t_0,p_0)\right)
\;\ge\; 0.
\end{equation}

\noindent A bounded function $u:[0,T]\times (0,1)\to\mathbb{R}$ is a 
\textbf{standard viscosity subsolution} of \eqref{eq:test} if, for every test function 
$\phi\in C^{1,2}([0,T]\times (0,1))$ and every local maximum point 
$(t_0,p_0)\in [0,T)\times (0,1)$ of $u^*-\phi$:
\begin{equation}\label{eq:subs}\notag
-\partial_t\phi(t_0,p_0)
+{H}_*\!\left(t_0,\,p_0,\,\phi_{pp}(t_0,p_0)\right)
\;\le\; 0.
\end{equation}

\noindent Here ${H}^*$ and ${H}_*$ denote the upper and lower semicontinuous envelopes of ${H}$. A bounded function $u$ is a \textbf{standard viscosity solution} of \eqref{eq:hjb} if it is simultaneously a standard viscosity supersolution and a standard viscosity subsolution.
\end{definition}

\begin{proposition}\label{prop:boundary}
    Under Assumption \ref{assumption:cd}, $\overline w^*$ is a bounded standard viscosity subsolution and $\overline w_*$ is a bounded standard viscosity supersolution to 
        \begin{equation}\label{eq:pdeup}
-\partial_t u + \overline H (t,\,p, u_{pp})=0,
    \end{equation}
and $\underline w^*$ is a bounded viscosity subsolution and $\underline w_*$ is a bounded viscosity supersolution to         
\begin{equation}\label{eq:pdedown}
-\partial_t u + \underline H(t,\,p,u_{pp})=0,
    \end{equation}
    with terminal condition $u(T)=0$.
\end{proposition}
Clearly, if we have a comparison for \eqref{eq:pdedown} and \eqref{eq:pdeup}, we have that $\overline w^*=\overline w_*$ and 
$\underline w^*=\underline w_*.$
This implies that for all $t\in[0,T]$ we have 
\begin{align}\label{eq:bd}
    w^*=w_* \mbox{ on }\cD_d\cup \cD_u.
\end{align}

This equality finally allows us to obtain the continuity of $w$ at the lateral boundary, which is the crucial step to characterize it as the unique viscosity solution to a Dirichlet boundary value problem. The common values $\overline w^*=\overline w_*$ on $\cD_u$ and $\underline w^*=\underline w_*$ on $ \cD_d$ are the utility that the principal gets when he reaches the boundary of the domain $\cD$. 
\begin{theorem}\label{thm:compinside}
Assume Assumptions \eqref{assumption:cd} and that \eqref{eq:pdeup} and \eqref{eq:pdedown} admits a comparison and unique solution $\overline w(t,p)$ and $\underline w(t,p)$. 
Then, $w$ is the unique continuous bounded viscosity solution to \eqref{eq:hjb} in $\cD$ and 
$w(t,\overline W(t),p)=\overline w(t,p)$
on $\cD_u$ and 
    $w(t,\underline W(t),p)=\underline w(t,p)$
on $\cD_d$.
with terminal condition $w(T,0,p)=0$.
\end{theorem}
Given the singularities of the Hamiltonian the viscosity property is to be understood in the sense of Definition \ref{eq:defstatc} but both sub and supersolution properties are required only on $\cD$. Thus, the viscosity property of $w$ is a direct consequence of the state constraint viscosity property where we do not use the subsolution property on the boundary.
The boundary values are also a direct consequence of Proposition \ref{prop:boundary}. Thus, the theorem only requires us to prove a comparison result for this PDE with given upper and lower semi continuous limits on the boundary, \eqref{eq:bd}. 

 The theorem states that the value of the principal can be described by an optimal control problem where the principal receives the utility $\overline w$ or $\underline w$ when the state reaches one of the boundaries, see Figure \ref{fig:trajectory}.

\begin{figure}[ht]
\centering
\begin{tikzpicture}[
  scale=0.78,
  >=Latex,
  declare function={
    Wup(\t)  = (1.15/0.5) * (1 - exp(-0.5*(8-\t)));
    Wlow(\t) = (-0.85/0.5) * (1 - exp(-0.5*(8-\t)));
  }
]
  \draw[->] (-0.15,0)--(9,0) node[right]{$t$};
  \draw[->] (0,-2.1)--(0,3.6) node[above]{$Y^0_t-Y^1_t$};
  \draw[dashed,gray!50] (8,-2.1)--(8,3.4);
  \node[below=1pt,font=\small] at (8,0) {$T$};
  \node[below left=-1pt,font=\small] at (0,0) {$0$};

  \draw[very thick,domain=0:8,samples=160,smooth]
    plot (\x,{Wup(\x)});
  \node[right,font=\small,fill=white,inner sep=1pt] at (6.6,1.25)
    {$\oW(t)=\tfrac{\bar a}{\kappa}(1{-}e^{-\kappa(T-t)})$};

  \draw[very thick,dashed,domain=0:8,samples=160,smooth]
    plot (\x,{Wlow(\x)});
  \node[right,font=\small,fill=white,inner sep=1pt] at (5.8,-1.25)
    {$\uW(t)=\tfrac{\underline a}{\kappa}(1{-}e^{-\kappa(T-t)})$};

  \draw[thick]
    (0.000,0.218) -- (0.036,0.245) -- (0.073,0.312) -- (0.109,0.277) --
    (0.145,0.189) -- (0.182,0.186) -- (0.218,0.220) -- (0.255,0.294) --
    (0.291,0.213) -- (0.327,0.128) -- (0.364,0.178) -- (0.400,0.211) --
    (0.436,0.315) -- (0.473,0.381) -- (0.509,0.374) -- (0.545,0.355) --
    (0.582,0.527) -- (0.618,0.505) -- (0.655,0.451) -- (0.691,0.374) --
    (0.727,0.380) -- (0.764,0.390) -- (0.800,0.367) -- (0.836,0.363) --
    (0.873,0.382) -- (0.909,0.428) -- (0.945,0.501) -- (0.982,0.514) --
    (1.018,0.393) -- (1.055,0.441) -- (1.091,0.346) -- (1.127,0.336) --
    (1.164,0.239) -- (1.200,0.242) -- (1.236,0.193) -- (1.273,0.210) --
    (1.309,0.434) -- (1.345,0.432) -- (1.382,0.489) -- (1.418,0.401) --
    (1.455,0.383) -- (1.491,0.429) -- (1.527,0.424) -- (1.564,0.450) --
    (1.600,0.459) -- (1.636,0.393) -- (1.673,0.433) -- (1.709,0.380) --
    (1.745,0.508) -- (1.782,0.475) -- (1.818,0.523) -- (1.855,0.524) --
    (1.891,0.589) -- (1.927,0.550) -- (1.964,0.618) -- (2.000,0.610) --
    (2.036,0.695) -- (2.073,0.736) -- (2.109,0.745) -- (2.145,0.693) --
    (2.182,0.744) -- (2.218,0.775) -- (2.255,0.884) -- (2.291,0.857) --
    (2.327,0.893) -- (2.364,0.925) -- (2.400,0.943) -- (2.436,0.955) --
    (2.473,0.965) -- (2.509,0.923) -- (2.545,0.854) -- (2.582,0.828) --
    (2.618,0.867) -- (2.655,0.968) -- (2.691,1.032) -- (2.727,1.101) --
    (2.764,1.150) -- (2.800,1.194) -- (2.836,1.203) -- (2.873,1.251) --
    (2.909,1.366) -- (2.945,1.339) -- (2.982,1.288) -- (3.018,1.415) --
    (3.055,1.426) -- (3.091,1.486) -- (3.127,1.538) -- (3.164,1.456) --
    (3.200,1.614) -- (3.236,1.725) -- (3.273,1.729) -- (3.309,1.772) --
    (3.345,1.772) -- (3.382,1.705) -- (3.418,1.690) -- (3.455,1.708) --
    (3.491,1.776) -- (3.527,1.816) -- (3.564,1.812) -- (3.600,1.895) --
    (3.636,1.836) -- (3.673,1.882) -- (3.709,1.948) -- (3.745,1.920) --
    (3.782,1.793) -- (3.818,1.726) -- (3.855,1.737) -- (3.891,1.767) --
    (3.927,1.865) -- (3.964,1.792) -- (4.000,1.787) -- (4.036,1.826) --
    (4.073,1.810) -- (4.109,1.647) -- (4.145,1.698) -- (4.182,1.850) --
    (4.218,1.912) -- (4.255,1.846) -- (4.291,1.814) -- (4.327,1.826) --
    (4.364,1.917) -- (4.400,1.910) -- (4.436,1.903) -- (4.473,1.896) --
    (4.509,1.887) -- (4.545,1.881) -- (4.582,1.833) -- (4.618,1.799) --
    (4.655,1.846) -- (4.691,1.850) -- (4.727,1.842) -- (4.764,1.834) --
    (4.800,1.712) -- (4.836,1.665) -- (4.873,1.626) -- (4.909,1.810) --
    (4.945,1.801) -- (4.982,1.791) -- (5.018,1.782) -- (5.055,1.773) --
    (5.091,1.763) -- (5.127,1.753) -- (5.164,1.743) -- (5.200,1.733) --
    (5.236,1.722) -- (5.273,1.712) -- (5.309,1.701) -- (5.345,1.690) --
    (5.382,1.679) -- (5.418,1.667) -- (5.455,1.656) -- (5.491,1.644) --
    (5.527,1.632) -- (5.564,1.620) -- (5.600,1.607) -- (5.636,1.595) --
    (5.673,1.582) -- (5.709,1.568) -- (5.745,1.555) -- (5.782,1.541) --
    (5.818,1.527) -- (5.855,1.513) -- (5.891,1.499) -- (5.927,1.484) --
    (5.964,1.469) -- (6.000,1.454) -- (6.036,1.438) -- (6.073,1.423) --
    (6.109,1.406) -- (6.145,1.390) -- (6.182,1.373) -- (6.218,1.356) --
    (6.255,1.339) -- (6.291,1.321) -- (6.327,1.303) -- (6.364,1.285) --
    (6.400,1.267) -- (6.436,1.248) -- (6.473,1.228) -- (6.509,1.209) --
    (6.545,1.189) -- (6.582,1.168) -- (6.618,1.147) -- (6.655,1.126) --
    (6.691,1.105) -- (6.727,1.083) -- (6.764,1.060) -- (6.800,1.038) --
    (6.836,1.015) -- (6.873,0.991) -- (6.909,0.967) -- (6.945,0.943) --
    (6.982,0.918) -- (7.018,0.892) -- (7.055,0.866) -- (7.091,0.840) --
    (7.127,0.813) -- (7.164,0.786) -- (7.200,0.758) -- (7.236,0.730) --
    (7.273,0.701) -- (7.309,0.672) -- (7.345,0.642) -- (7.382,0.612) --
    (7.418,0.581) -- (7.455,0.549) -- (7.491,0.517) -- (7.527,0.484) --
    (7.564,0.451) -- (7.600,0.417) -- (7.636,0.382) -- (7.673,0.347) --
    (7.709,0.311) -- (7.745,0.275) -- (7.782,0.238) -- (7.818,0.200) --
    (7.855,0.161) -- (7.891,0.122) -- (7.927,0.082) -- (7.964,0.041) --
    (8.000,0.000);
  \node[font=\small,anchor=west] at (0.2,1.05) {$Y^0_t-Y^1_t$};
  \draw[->,gray!60,shorten >=2pt] (1.15,1.0) to[out=-10,in=110] (1.5,0.45);

  \fill (4.909,1.810) circle (2pt);
  \draw[->,shorten >=3pt] (3.1,3.1) to[out=-5,in=125] (4.87,1.88);
  \node[font=\small,align=left,anchor=west] at (0.3,3.3)
    {boundary hit at $t=t^*$:};
  \node[font=\small,anchor=west] at (0.3,2.95)
    {$Z^1_{s}\!\in\!\oV,\;\;Z^0_{s}=Z^1_{s},\,\forall s\in[t^*,T]$};
\end{tikzpicture}
\caption{A typical trajectory of the gap process $Y^0_t-Y^1_t$ under
Theorem~\ref{thm:compinside} and Theorem~\ref{thm:domain}(b),
illustrated in the simpler $x$-independent setting of Section~\ref{sec:dirichlet}
(so that $\uW,\oW$ depend only on $t$ and are given by the explicit
formulas shown). The trajectory remains in the strip
$[\uW(t),\oW(t)]$ for all $t\in[0,T]$ and terminates at $0$ at $t=T$
(both boundaries vanish at $T$). At an interior time, $(Z^0_t,Z^1_t)\in\RR^d\times\RR^d$
is unconstrained; at the boundary-hitting time $t^*$ shown, the matching
condition activates: $Z^1_{s}\in\oV(s,X_{s},\oZ_{s})=\oV$ and
$Z^0_{s}=Z^1_{s}$ for all $s\in[t^*,T]$ (the symmetric condition with
$\uV$ holds at lower-boundary hits).}
\label{fig:trajectory}
\end{figure}
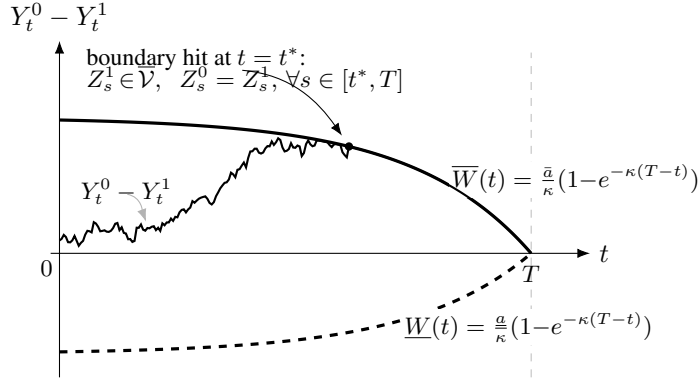
In Section \ref{sec:benchmark}, we numerically solve \eqref{eq:hjb} for the following two examples. 

\begin{example}[Dominated case]\label{ex:dom}
In this example, we consider the case where type $0$ is the good agent (with lower cost) and type $1$ is the bad agent (with higher cost).

For $\theta\in\{0,1\}$, let
\[
\lambda(t,x,\theta,\alpha)=\alpha,
\qquad
c(t,x,\theta,\alpha)=\frac{j_\theta}{2}\alpha^2,
\quad
j_\theta=\theta+1.
\]
Let the admissible action set be
\[
A=[0,\sqrt{2\overline a}],
\]
Then, for each $\theta\in\{0,1\}$, the Hamiltonian is
\[
H^\theta(z^\theta)
=
\sup_{\alpha\in A}
\left\{
z^\theta \alpha-\frac{j_\theta}{2}\alpha^2
\right\}.
\]
The corresponding optimizer is given by
\[
A^\theta(z^\theta)=\Pi_A\!\left(\frac{z^\theta}{j_\theta}\right),
\]
where $\Pi_A:\mathbb{R}\to A$ denotes the projection onto $A$, namely
\[
\Pi_A(r):=\min\{\max\{r,0\},\sqrt{2\overline a}\}, \qquad r\in\mathbb{R}.
\]
In other words, the optimal feedback is obtained by truncating the unconstrained maximizer $\frac{z^\theta}{j_\theta}$ to the admissible interval $A$.

Moreover, the boundary functions $\underline W$ and $\overline W$ are independent of $x$. Since
\[
\underline a=\inf_{z\in\mathbb{R}}\bigl(H^0(z)-H^1(z)\bigr)=0,
\qquad
\overline a=\sup_{z\in\mathbb{R}}\bigl(H^0(z)-H^1(z)\bigr),
\]
the ODEs in \eqref{def:w} yield
\[
\underline W(t)=0,
\qquad
\overline W(t)=\frac{\overline a}{\kappa}\bigl(1-e^{-\kappa(T-t)}\bigr).
\]
In particular,
\[
\underline W(t)<\overline W(t),
\qquad t\in[0,T).
\]
For this Example 
we have
\begin{align*}
\overline{H}(t,p,\gamma) &= -\sqrt{2 \bar{a}}+\frac{3}{2} \bar{a} e^{\kappa(T-t)} , \\
\underline{H}(t,p,\gamma) &= 0.
\end{align*}
and \eqref{eq:pdedown} and \eqref{eq:pdeup} admit the following solutions

\begin{align*}
\overline{w}(t,p)=\overline{w}(t) &= \sqrt{2 \bar{a}}(T-t)-\frac{3 \bar{a}}{2 \kappa}\left(e^{\kappa(T-t)}-1\right) ., \\
\underline{w}(t,p)=\underline{w}(t) &= 0.
\end{align*}

\end{example}
\begin{example}[Non-dominated case]\label{ex:ndom}
In this example, we consider the non-dominated case, meaning that neither type uniformly dominates the other. In particular, we have $\underline a<0<\overline a$.

For $\theta\in\{0,1\}$, let
\[
\lambda(t,x,\theta,\alpha)=\alpha,
\qquad
c(t,x,\theta,\alpha)=\frac{1}{2}\alpha^2+(-1)^{\theta+1}\alpha.
\]
Let the admissible action set be
\[
A=\left[\frac{\underline a}{2},\frac{\overline a}{2}\right],
\]
Then, for each $\theta\in\{0,1\}$, the Hamiltonian is given by
\[
H^\theta(z)
=
\sup_{\alpha\in A}
\left\{
z\alpha-\left(\frac{1}{2}\alpha^2+(-1)^{\theta+1}\alpha\right)
\right\}.
\]
The corresponding optimizer is
\[
A^\theta(z)=\Pi_A\!\left(z+(-1)^\theta\right),
\]
where $\Pi_A:\RR\to A$ denotes the projection onto $A$, namely
\[
\Pi_A(r):=\min\left\{\max\left\{r,\frac{\underline a}{2}\right\},\frac{\overline a}{2}\right\},
\qquad r\in\RR.
\]
In other words, the optimal feedback is obtained by truncating the unconstrained maximizer $z+(-1)^\theta$ to the admissible interval $A$.

Moreover, the boundary functions $\underline W$ and $\overline W$ are independent of $x$. Since
\[
\inf_{z\in\RR}\bigl(H^0(z)-H^1(z)\bigr)=\underline a,
\qquad
\sup_{z\in\RR}\bigl(H^0(z)-H^1(z)\bigr)=\overline a,
\]
the ODEs in \eqref{def:w} yield
\[
\underline W(t)=\frac{\underline a}{\kappa}\bigl(1-e^{-\kappa(T-t)}\bigr),
\qquad
\overline W(t)=\frac{\overline a}{\kappa}\bigl(1-e^{-\kappa(T-t)}\bigr).
\]
In particular,
\[
\underline W(t)<\overline W(t),
\qquad t\in[0,T).
\]
For this Example \ref{ex:ndom}
we have
\begin{align*}
\overline{H}(t,p,\gamma) &= -\frac{\bar{a}}{2} + \frac{\bar{a}^2}{8}\,e^{\kappa(T-t)}, \\
\underline{H}(t,p,\gamma) &= -\frac{\underline{a}}{2} + \frac{\underline{a}^2}{8}\,e^{\kappa(T-t)}.
\end{align*}
and \eqref{eq:pdedown} and \eqref{eq:pdeup} admits the following solutions
\begin{align*}
\overline{w}(t,p)=\overline{w}(t) &= \frac{\bar{a}}{2}\,(T-t) - \frac{\bar{a}^2}{8\kappa}\bigl(e^{\kappa(T-t)}-1\bigr), \\
\underline{w}(t,p)=\underline{w}(t) &= \frac{\underline{a}}{2}\,(T-t) - \frac{\underline{a}^2}{8\kappa}\bigl(e^{\kappa(T-t)}-1\bigr).
\end{align*}

\end{example}

\section{Structure of optimal contracts}\label{s:optcontract}
The Theorem \ref{thm:compinside} fully characterizes $V_{sc}(t,y_0,y_1,p)= -\frac{e^{\kappa(T-t)}}{2}\,(y_0+y_1) + w(t,y_0-y_1,p)$ where $w$ is the solution to  
\begin{align}
  -\partial_t w(t,y,p)
    + H\!\bigl(t,y,p,\nabla w(t,y,p),\nabla^2 w(t,y,p)\bigr)
    &\;=\; 0,
    && (t,y,p)\in \cD, \\[0.3em]
  -\partial_t w(t,y,p)
    + \overline H\bigl(t,p,\partial_{pp}w(t,y,p)\bigr)
    &= 0,
    && (t,p)\in [0,T)\times (0,1),\ y=\overline W(t), \\[0.3em]
  -\partial_t w(t,y,p)
    + \underline H\bigl(t,p,\partial_{pp}w(t,y,p)\bigr)
    &= 0,
    && (t,p)\in [0,T)\times (0,1),\ y=\underline W(t).
\end{align}
     with terminal condition $w(T,0,p)=0.$ 

Assume now that this PDE has a smooth solution and for all $(t,y,p)\in \mathrm{cl}_y(\cD)$ choose optimizer $$(Z^*_1(t,y,p),Z^*_2(t,y,p))$$ for $H\!\bigl(t,y,p,\nabla w(t,y,p),\nabla^2 w(t,y,p)\bigr)$ defined in \eqref{eq:defHinside}. Similarly, on the lateral boundaries, given the definition \eqref{eq:defuplowpde}, choose optimizer $Z^*(t,y,p)$ for $y=\overline W(t)$ or $y=\underline W(t)$ and extend the definition of $(Z^*_1,Z^*_2)$ to this lateral boundary by taking 
$(Z^*_1(t,y,p),Z^*_2(t,y,p))=(Z^*(t,y,p),Z^*(t,y,p))$ (the equality of the controls is needed to kill the noise of $Y$). Given the definition of $\overline H,\underline H$ and $\overline \cV,\underline \cV$, the controls at the boundary satisfy the state constraint and we have the following verification result that shows that the optimal contracts are highly non-trivial and in fact depend on the continuation utilities and the belief of the principal on the private information of the agent. 

\begin{theorem}[Verification Theorem] Assume that $w\in C^{1,2}(\mathrm{cl}_y(\cD))$, let $(y_0,y_1)$ be optimizers of \eqref{eq:repvaluestatic1sc} or \eqref{eq:repvaluestaticsc_uc}, and let $(Z^*_0(t,y,p),Z^*_1(t,y,p))$ be defined as above.
Then the contract
\[
\xi^*\;:=\;\frac{Y^{0,0,X_0,y_0,Z^*_0}_T}{\beta_0}\;=\;\frac{Y^{1,0,X_0,y_1,Z^*_1}_T}{\beta_1},
\]
with $Y^{\theta,0,X_0,y_\theta,Z^*_\theta}$ defined in \eqref{fsde}, is an optimal contract for the principal. 
\end{theorem}

The proof of the theorem is the standard verification result and is not provided.

\section{Numerical Results}\label{sec:benchmark}
By Theorem \ref{thm:compinside}, the HJB equation \eqref{eq:hjb} is analytically well posed but remains numerically challenging, since the optimizer $(z_0^*,z_1^*)$ may be unbounded and the effective domain must be tracked dynamically. A natural computational approach is to \emph{truncate} the sensitivity domain to $[-K,K]$ for some large parameter $K>0$. The resulting PDE then has a bounded control set and a classical (non-singular) Hamiltonian, which makes it amenable to standard numerical methods. In particular, we apply the Deep Galerkin Method introduced in \cite{sirignano_spiliopoulos} to analyze both the dominated and non-dominated cases and to explore the associated economic implications.

\subsection{Dominated case}
In this section, we present numerical results that illustrate Example~\ref{ex:dom}. Recall that, in this case, the agent is either good, meaning that he has a low cost of effort, or bad, meaning that he has a high cost of effort. 
Throughout the numerical experiments, we fix the following input values: $\overline{a}=1$, $\kappa=0.1$, $R = R_0 = R_1 = 0$ and $T=2$. 

\begin{figure}[h]
     \centering
     \begin{subfigure}[b]{0.45\textwidth}
         \centering
         \includegraphics[width=\textwidth]{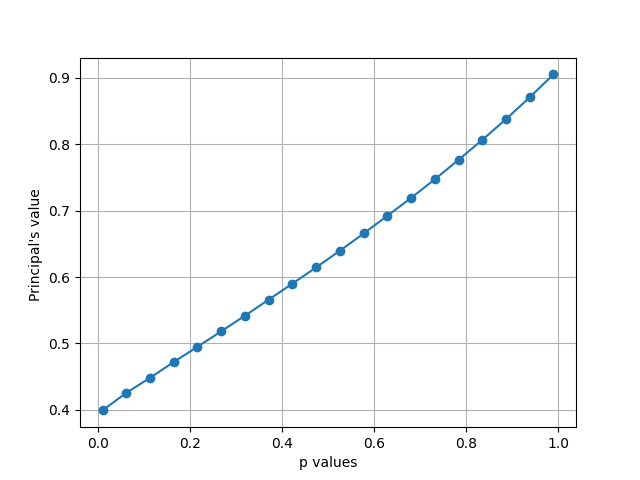}
         \caption{Principal's value}
         \label{fig:pv_bd_uc_d}
     \end{subfigure}
     \quad
     \begin{subfigure}[b]{0.45\textwidth}
         \centering
         \includegraphics[width=\textwidth]{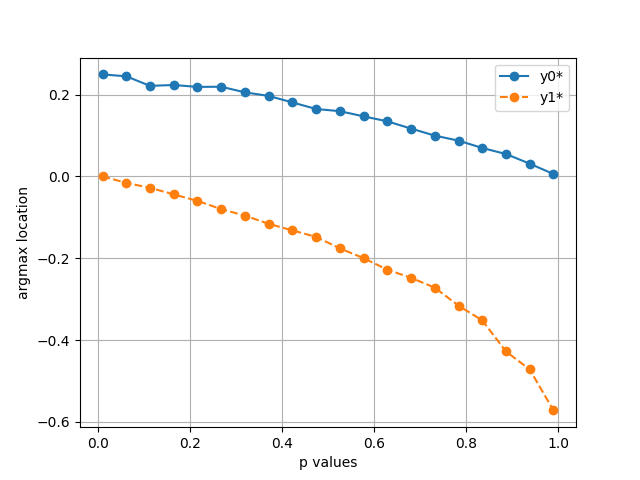}
         \caption{Argmax locations }
         \label{fig:argmax_bd_uc_d}
     \end{subfigure}
        \caption{Unconditionally rational}
    \label{fig:bd_uc_d}
\end{figure}

Figure \ref{fig:pv_bd_uc_d} illustrates the optimal value of the principal in terms of her initial belief that she is facing agent type 0 (the good agent). Figure \ref{fig:argmax_bd_uc_d} reports the optimal promised utilities offered to agent type 0 (blue), and agent type 1 (orange).

We observe that as the initial belief increases, the principal’s value also increases. Moreover, due to the domination relationship $H^0 \ge H^1$, the structure of the optimal contract changes with the belief level. When the initial belief $p_0$ is small, the principal benefits from promising a strictly higher utility to the potentially good agent. In contrast, when the initial belief $p_0$ is close to 1, the principal benefits from binding agent type 0 at the reservation utility while punishing the bad agent. This punishment corresponds to offering the bad agent a negative initial utility.

\begin{figure}[H]
     \centering
     \begin{subfigure}[b]{0.45\textwidth}
         \centering
         \includegraphics[width=\textwidth]{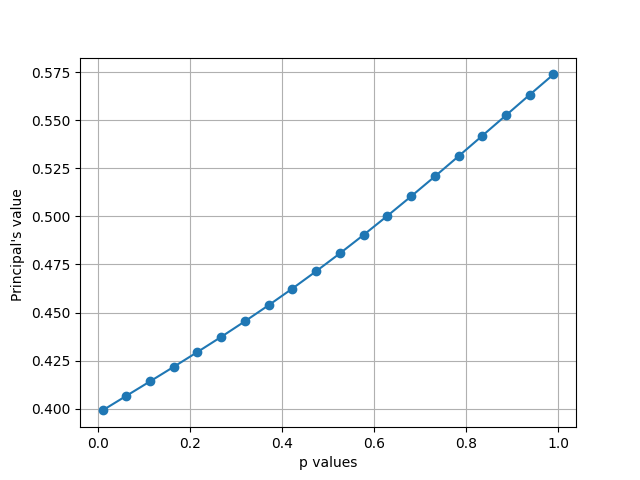}
         \caption{Principal's value}
         \label{fig:pv_bd_c_d}
     \end{subfigure}
     \quad
     \begin{subfigure}[b]{0.45\textwidth}
         \centering
         \includegraphics[width=\textwidth]{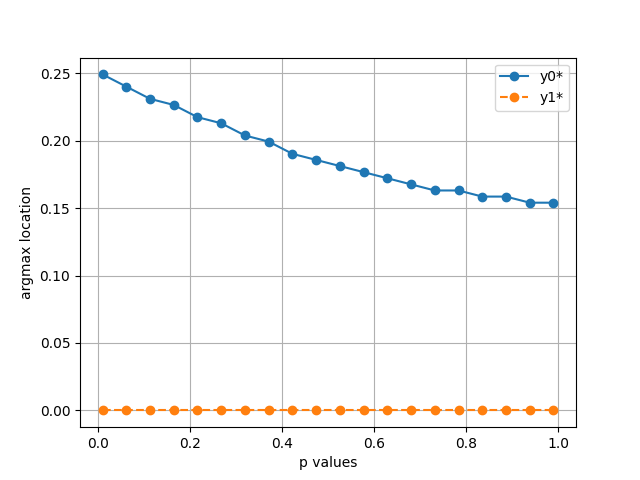}
         \caption{Argmax locations }
         \label{fig:argmax_bd_c_d}
     \end{subfigure}
        \caption{Individually (or conditionally) rational}
\label{fig:bd_c_d}
\end{figure}

Figure~\ref{fig:pv_bd_c_d} illustrates the principal’s optimal value as a function of the initial belief that the agent is of type 0 (the good agent). In this setting, the optimal value is obtained under separate participation constraints, namely $y_0 \geq R,y_1 \geq R$. Figure~\ref{fig:argmax_bd_c_d} reports the corresponding optimal promised utilities offered to agent type 0, and agent type 1.

As shown in Figure~\ref{fig:pv_bd_c_d}, the qualitative behavior is similar to that in Figure~\ref{fig:pv_bd_uc_d}. In particular, the principal’s value is monotone in the initial belief and exhibits convexity with respect to the prior. However, the structure of the optimal promised utilities differs from the unconstrained case displayed in Figure \ref{fig:argmax_bd_uc_d}. Under separate participation constraints, the promised utility to agent type 1 binds at the reservation level 
$R=0$ across the range of beliefs.

Moreover, as the initial belief increases, the principal optimally lowers the promised utility offered to agent type $0$. This is economically intuitive. When the game starts from a prior that places a higher probability on the good agent, the need to provide informational rents is reduced, which allows the principal to extract more surplus from agent type $0$ by lowering the promised utility.

\begin{figure}[H]
     \centering
     \begin{subfigure}[b]{0.47\textwidth}
         \centering
\includegraphics[width=\textwidth]{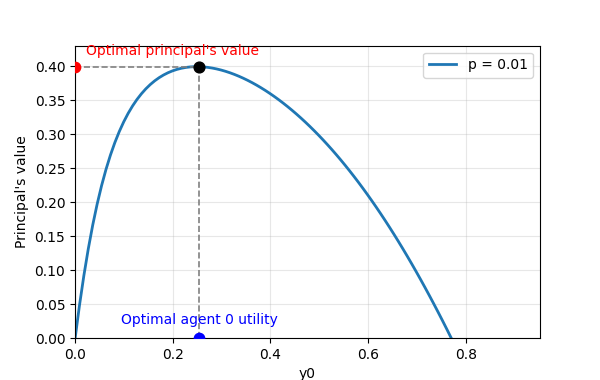}
         \caption{Principal's value  ($p =0.01$)}
\label{fig:slice_pv_fixed_type1_p0}
     \end{subfigure}
     \quad
     \begin{subfigure}[b]{0.47\textwidth}
         \centering
\includegraphics[width=\textwidth]{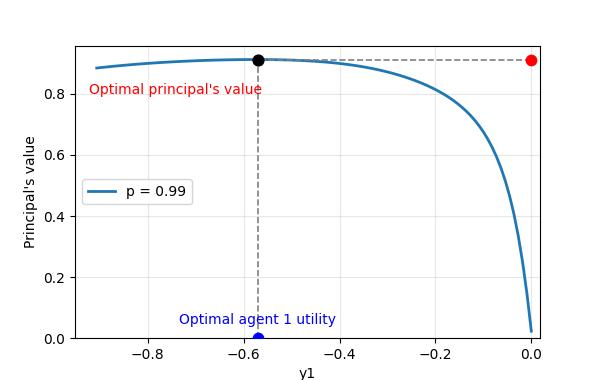}
         \caption{Principal's value ($p =0.99$)}
\label{fig:slice_pv_fixed_type0_p1}
     \end{subfigure}
        \caption{Cross-Sectional Slices of  $V\left(t, y_0, y_1, p\right)$  at Extreme Belief Levels (UR)}
\label{fig:infoy1vy2vs}
\end{figure}

Figure \ref{fig:slice_pv_fixed_type1_p0} shows a slice of the principal's value as a function of the promised utility of agent $0$, with the promised utility of agent $1$ fixed at its reservation level, at $t=0$ and initial belief $p_0=0.01$. In this case, the game starts from a prior that places a very small probability on the good agent (type $0$). The figure shows that increasing the promised utility of this unlikely type can raise the principal's value, which indicates the presence of informational rent.

By contrast, Figure \ref{fig:slice_pv_fixed_type0_p1} shows a slice of the principal's value as a function of the promised utility of agent $1$, with the promised utility of agent $0$ fixed at its reservation level, again at $t=0$, but now with initial belief $p_0=0.99$. In this case, the game starts from a prior that places a very small probability on the bad agent (type $1$). The figure indicates that allowing a strictly negative promised utility for this unlikely type can increase the principal's value, again revealing the presence of informational loss.

\begin{figure}[H]
    \centering
\includegraphics[width=0.5\linewidth]{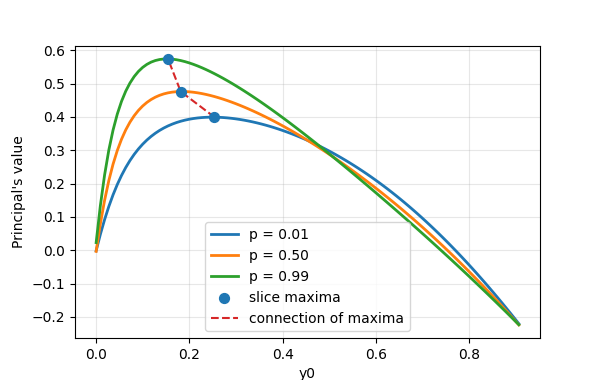}
    \caption{Cross-Sectional Slices of  $V\left(t, y_0, y_1, p\right)$ (CR)}
    \label{fig:slice_c_d}
\end{figure}

Figure \ref{fig:slice_c_d} presents slices of the principal's value in the case of conditional reservation utilities. In this figure, we fix $t=0$ and $y_1=0$, and consider several values of the initial belief $p_0$. The figure shows that, in the dominated case, where type $0$ uniformly generates more favorable incentives than type $1$ through the ordering $H^0\geq H^1$, the numerical results show that the optimal contract may leave a strictly positive rent to type $0$. Instead, the principal benefits from offering a strictly positive promised utility to agent 0. Moreover, as the initial belief $p_0$ increases, the principal benefits from offering a smaller positive promised utility to agent 0 , while the principal's optimal value increases with $p_0$.



\begin{figure}[H]
    \centering
\includegraphics[width=0.45\linewidth]{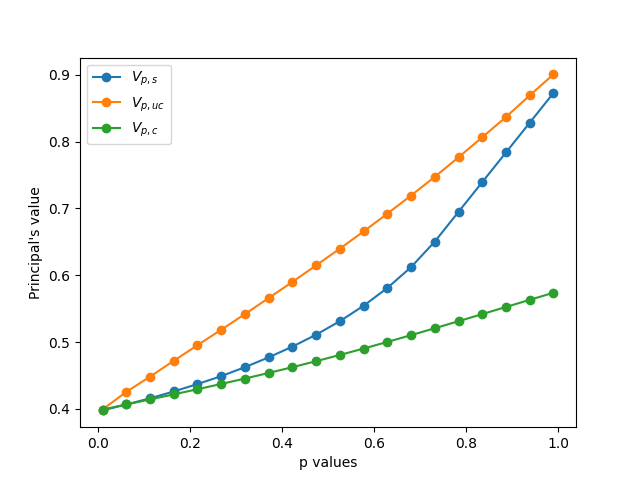}
    \caption{Comparison of the principal's values}
\label{fig:comp_principalvals_newconditiond}
\end{figure}

Figure \ref{fig:comp_principalvals_newconditiond} compares the principal's value as a function of the initial belief $p_0$. The three curves correspond to $V_{p,uc}$, the principal's value under the unconditional reservation utility constraint with a single contract, $V_{p,c}$, the principal's value under the conditional reservation utility constraint with a single contract, and $V_{p,s}$, the principal's value under the screening framework with a menu of contracts.

In Figure \ref{fig:comp_principalvals_newconditiond}, $V_{p,s}$ is obtained from the following static constrained optimization problem
\begin{align*}
V_{p,s}
=
\sup\Bigg\{
&\, p_0 V_{0}(0,y_0,y_1^c) + (1-p_0)V_{1}(0,y_0^c,y_{1}) : 
 (y_0,y_1^c,y_0^c,y_1)\in \RR^4,\\
&\quad\quad\quad \max\{y_0^c,R,y_{1}^c + \underline W (0)\} \leq y_0 \leq y_{1}^c + \overline W (0)\\
&\quad\quad\quad \max\{y_1^c,R,y_{0}^c - \overline W (0)\} \leq y_1 \leq y_{0}^c - \underline W(0)\;\;\;\;\;
\Bigg\}.
\end{align*}

First, $V_{p,uc}$ dominates both $V_{p,c}$ and $V_{p,s}$. Economically, this reflects the greater flexibility afforded by the unconditional reservation utility constraint, especially for extreme beliefs. Because the unconditional reservation utility is averaged across types, it gives the principal more room to adjust promised utilities and thereby improve her objective value.

Second, the screening value $V_{p,s}$ dominates $V_{p,c}$. This reflects the additional flexibility provided by the menu-of-contracts framework. Under $V_{p,c}$, the principal is restricted to offering a single contract, whereas under $V_{p,s}$ she can tailor the menu to the distinct incentive and participation requirements of the two types.

\subsection{Non-dominated case}

This section provides the numerical analysis for the non-dominated case in Example~\ref{ex:ndom}. Unlike the previous example, this environment does not admit a natural ranking of agent's types. The two types differ in their cost functions, but neither cost function uniformly dominates the other. The parameters are fixed as $\overline{a}=1$, $\underline{a}=-1$, $\kappa=0.1$, $R = R_0 = R_1 = 0$ and $T=2$.

\begin{figure}[H]
     \centering
     \begin{subfigure}[b]{0.47\textwidth}
         \centering
         \includegraphics[width=\textwidth]{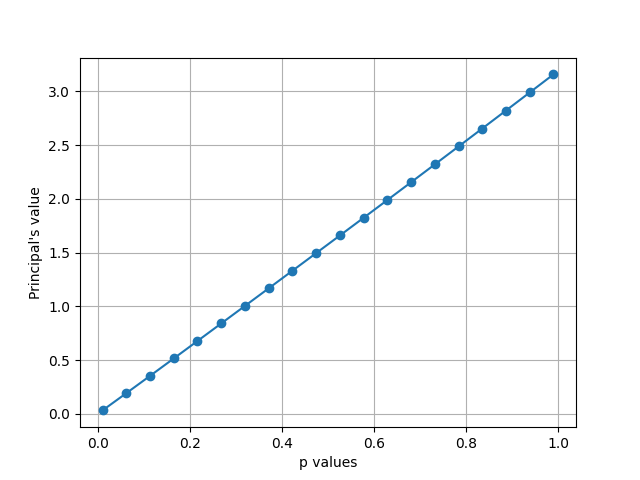}
         \caption{Principal's value}
         \label{fig:bd_uc_nd_principalvalue}
     \end{subfigure}
     \quad
     \begin{subfigure}[b]{0.47\textwidth}
         \centering
         \includegraphics[width=\textwidth]{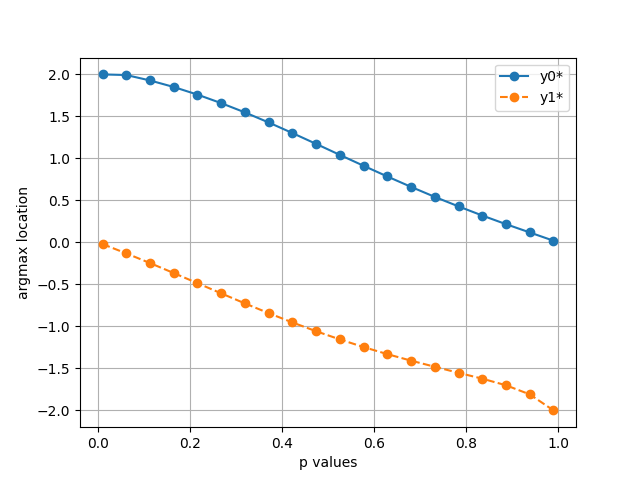}
         \caption{Argmax locations }
         \label{fig:bd_uc_nd_armax}
     \end{subfigure}
        \caption{Unconditional rationality}
    \label{fig:bd_uc_nd}
\end{figure}

As in Figure~\ref{fig:pv_bd_uc_d}, Figure~\ref{fig:bd_uc_nd_principalvalue} illustrates how the principal's optimal value varies with her initial belief that she is facing agent type $0$. Figure~\ref{fig:bd_uc_nd_armax} reports the corresponding optimal promised utilities assigned to type $0$ and type $1$.

We observe that the principal's value increases with the initial belief. Although the agents' cost functions are not ordered in the non-dominated case, the numerical results suggest that type $0$ still induces more favorable effort incentives for the principal in this parameter regime. Consequently, the optimal promised utilities may display a dominance pattern similar to that observed in Figure~\ref{fig:argmax_bd_uc_d}. However, this ordering is not robust in the non-dominated case: the relationship between the initial optimal promised utilities for the two types depends sensitively on the parameter choices, especially on $\overline{a}$ and $\underline{a}$. Therefore, unlike in the dominated case, no general monotonic ordering of the optimal promised utilities should be inferred from this figure.

\begin{figure}[H]
     \centering
     \begin{subfigure}[b]{0.47\textwidth}
         \centering
         \includegraphics[width=\textwidth]{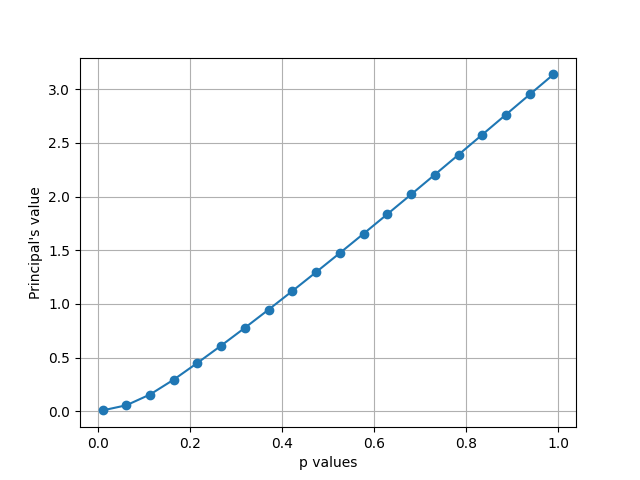}
         \caption{Principal's value}
        \label{fig:bd_c_nd_principalvalue}
     \end{subfigure}
     \quad
     \begin{subfigure}[b]{0.47\textwidth}
         \centering
         \includegraphics[width=\textwidth]{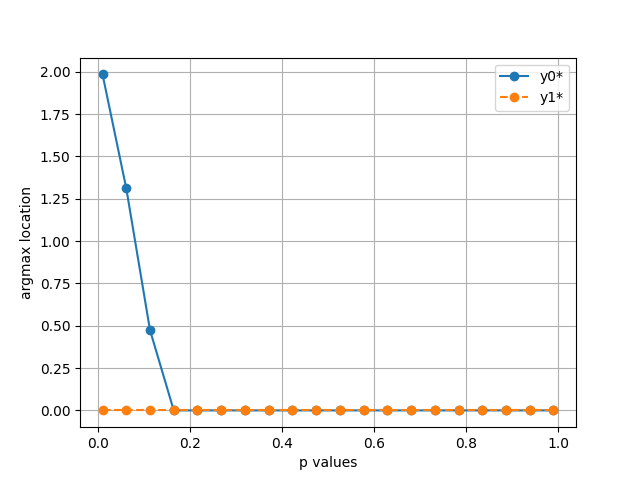}
         \caption{Argmax locations }
         \label{fig:bd_c_nd_armax}
     \end{subfigure}
        \caption{Conditional rationality}
    \label{fig:bd_c_nd}
\end{figure}

Figure~\ref{fig:bd_c_nd_principalvalue} illustrates the principal's optimal value as a function of the initial belief that the agent is of type $0$. In this setting, the optimal value is computed under separate participation constraints, namely $y_0\geq R$ and $y_1\geq R$. Figure~\ref{fig:bd_c_nd_armax} reports the corresponding optimal promised utilities offered to agent type $0$ and agent type $1$. In the non-dominated case, the principal cannot uniformly rank the two types as good or bad, since there is no global dominance relation between their cost functions.

As shown in Figure~\ref{fig:bd_c_nd_principalvalue}, the qualitative behavior of the principal's value is similar to that in Figure~\ref{fig:bd_uc_nd_principalvalue}. In particular, the principal's value is increasing in the initial belief and appears to be convex with respect to the prior. The structure of the optimal promised utilities in Figure~\ref{fig:bd_c_nd_armax} is also broadly similar to that of the unconditional case reported in Figure~\ref{fig:bd_uc_nd_armax}.

However, in the conditional case, the optimal promised utility assigned to the agent type $1$ decreases rapidly as the initial belief increases. For sufficiently large initial beliefs, the participation constraints of both types become binding over part of the belief range. Consequently, on this range, as in the moral-hazard benchmark, both types receive exactly their reservation utilities, and the principal extracts all surplus above the participation levels.



\begin{figure}[H]
    \centering
\includegraphics[width=0.45\linewidth]{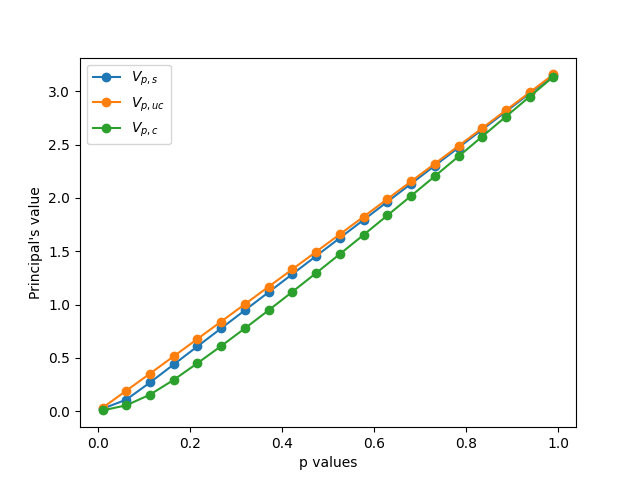}
    \caption{Comparison of the principal's values}
\label{fig:comp_principalvals_nd_bd_Hn}
\end{figure}

As in Figure~\ref{fig:comp_principalvals_newconditiond}, 
Figure~\ref{fig:comp_principalvals_nd_bd_Hn} compares the principal's value as a function of the initial belief $p_0$. The three curves represent $V_{p,uc}$, the value under the unconditional reservation-utility constraint with a single contract; $V_{p,c}$, the value under the conditional reservation-utility constraint with a single contract; and $V_{p,s}$, the value under the screening formulation with a menu of contracts.

The numerical ordering satisfies
\[
    V_{p,uc} \geq V_{p,s} \geq V_{p,c},
\]
which is consistent with Theorem~\ref{thm:HJB_principal_screens}. In contrast to Figure~\ref{fig:comp_principalvals_newconditiond}, the non-dominated case shows that the gap between the single-contract value and the screening value becomes small when the initial belief is close to the extremes. Economically, this means that when the principal is nearly certain about the agent's type, the value of screening is limited. Thus, in these regions, offering a single contract may approximate the screening value well, provided that the associated loss is acceptable.



\appendix

\section{Appendix}

\subsection{Proof of preliminary results}


\begin{proof}[Proof of Lemma \ref{lem:bsde_wp}]
First we verify the well-posedness of the BSDE. Since $\sigma^{-1}$ is bounded, $\FF^X$ coincides with the augmented filtration of $B^{\bar P}$ up to $\bar P$-null sets. For brevity, write $(\cY^\theta,\cZ^\theta):=(\cY^{\theta,0,X_0,\xi},\cZ^{\theta,0,X_0,\xi})$. Setting $\widetilde\cZ^\theta_r:=\sigma(r,X_r)^\top\cZ^\theta_r$, and using $dX_r=\sigma(r,X_r)\,dB^{\bar P}_r$, the BSDE \eqref{eq:bsdek} is equivalent to the Brownian BSDE 
\begin{equation}\label{eq:bsde_brownian_proof}
\cY^\theta_s=\beta(\theta)\xi+\int_s^T \widetilde f^\theta\bigl(r,X_r,\cY^\theta_r,\widetilde\cZ^\theta_r\bigr)\,dr-\int_s^T \bigl(\widetilde\cZ^\theta_r\bigr)^\top dB^{\bar P}_r,
\qquad \bar P\text{-a.s.,}
\end{equation}
with driver
\[
\widetilde f^\theta(r,x,y,\widetilde z):=H^\theta\!\bigl(r,x,(\sigma(r,x)^\top)^{-1}\widetilde z\bigr)-\kappa(r,x)\,y.
\]
The Carath\'eodory structure\footnote{Let $X,A$ be metric spaces. We say that a measurable function $f: X\times A\mapsto \RR$ satisfies the Carath\'eodory structure if $f(x,\cdot)$ is continuous for all $x\in X$.} of $\lambda$ and $c$ together with the compactness of $A$ ensures by the measurable maximum theorem \cite[Lemma~18.3]{guide2006infinite} that $H^\theta$ is jointly Borel measurable in $(r,x,z)$; hence $\widetilde f^\theta$ is progressively measurable in $(r,\omega)$. Since $H^\theta(r,x,\cdot)$ is the supremum of affine functions with slopes bounded by $\|\sigma\|_\infty\|\lambda\|_\infty$, $\widetilde f^\theta$ is Lipschitz in $(y,\widetilde z)$ with constant depending only on $\|\kappa\|_\infty,\|\sigma\|_\infty,$ $\|\sigma^{-1}\|_\infty,\|\lambda\|_\infty$. The free term satisfies $|\widetilde f^\theta(r,X_r,0,0)|=|H^\theta(r,X_r,0)|\le\|c\|_\infty$ by boundedness of $c$, and $\xi\in L^2(\bar P)$ by definition of $\cC_a$. The classical $L^2$-theory for Lipschitz BSDEs \cite[Theorems~2.1--2.2]{el1997backward} yields existence, uniqueness, and Lipschitz stability of \eqref{eq:bsde_brownian_proof} in $\bS_2(\bar P)\times\HH_2(\bar P)$ — with stability constant depending only on $T$ and on the Lipschitz constant of the driver, independently of the free term. Transferring back via $\cZ^\theta=(\sigma^\top)^{-1}\widetilde\cZ^\theta$, and using the boundedness of $\sigma,\sigma^{-1}$ gives existence, uniqueness, and the stated Lipschitz stability for $(\cY^\theta,\cZ^\theta)$ in $\bS_2(\bar P)\times\HH_2(\bar P)$.

\smallskip
Next, we verify the equality \eqref{eq:cE_rep} by showing both inclusions. We start by showing the 
{inclusion $\supseteq$}. Fix $\xi\in\cC_a$ and set
\[
y_\theta:=\cY^{\theta,0,X_0,\xi}_0,\qquad Z^\theta:=\cZ^{\theta,0,X_0,\xi},\qquad \theta\in\{0,1\}.
\]
By the well-posedness above, $Z^\theta\in\HH_2(\bar P)$, so $Z^\theta\in\cV(X_0)$ in the sense of Definition~\ref{def:response}. Rewriting \eqref{eq:bsdek} forward from time $0$, the pair $(\cY^{\theta,0,X_0,\xi},\cZ^{\theta,0,X_0,\xi})$ satisfies the forward equation \eqref{fsde} with initial value $y_\theta$ and control $Z^\theta$. The forward equation \eqref{fsde} is linear in $Y^\theta$ with bounded coefficients, hence pathwise unique; therefore
\[
Y^{\theta,0,X_0,y_\theta,Z^\theta}=\cY^{\theta,0,X_0,\xi}\quad\text{on }[0,T],
\]
and, in particular, $Y^{\theta,0,X_0,y_\theta,Z^\theta}_T=\beta_\theta \xi$ for both $\theta\in\{0,1\}$. Hence, $\beta_0^{-1} Y^{0,0,X_0,y_0,Z^0}_T=\beta_1^{-1} Y^{1,0,X_0,y_1,Z^1}_T$, so $(y_0,y_1)\in\cE(X_0)$. 
\smallskip
Next, we show the inclusion {$\subseteq$ in \eqref{eq:cE_rep}.} Fix $(y_0,y_1)\in\cE(X_0)$, and $(Z^0,Z^1)\in(\cV(X_0))^2$, and set
\[
\xi:=\beta_0^{-1} Y^{0,0,X_0,y_0,Z^0}_T=\beta_1^{-1} Y^{1,0,X_0,y_1,Z^1}_T.
\]
By Remark~\ref{rem:selection}, $Y^{\theta,0,X_0,y_\theta,Z^\theta}\in\bS_2(\bar P)$ for $\theta\in\{0,1\}$, so $\EE^{\bar P}[|\xi|^2]<\infty$ and $\xi\in\cC_a$. Rearranging the forward dynamics \eqref{fsde} backward from terminal value $\beta_\theta \xi$ shows that $(Y^{\theta,0,X_0,y_\theta,Z^\theta},Z^\theta)$ solves the BSDE \eqref{eq:bsdek}. By the uniqueness statement derived previously. Hence,
\[
\bigl(Y^{\theta,0,X_0,y_\theta,Z^\theta},Z^\theta\bigr)=\bigl(\cY^{\theta,0,X_0,\xi},\cZ^{\theta,0,X_0,\xi}\bigr)\quad\text{on }[0,T].
\]
Evaluating at $s=0$ yields $y_\theta=\cY^{\theta,0,X_0,\xi}_0$ for $\theta\in\{0,1\}$, which proves the inclusion.

\smallskip
Finally, we show that $\mathcal{E}(X_0)$ is connected. We note that $\cC_a=L^2(\Omega_X,\cF_T^X,\bar P)$ is convex, hence connected. By the Lipschitz stability shown above, the map
\[
\xi\;\longmapsto\;\bigl(\cY^{0,0,X_0,\xi}_0,\,\cY^{1,0,X_0,\xi}_0\bigr)
\]
is continuous from $L^2(\bar P)$ into $\RR^2$. By \eqref{eq:cE_rep}, $\cE(X_0)$ is the image of this map, and the continuous image of a connected set is connected. \qed 
\end{proof}
\begin{proof}[Proof of Theorem \ref{thm:domain}]
We prove the theorem in multiple steps. First, we show Part $(a)$. 
\smallskip

\emph{Step 1}: \textit{$(y_0, y_1)\in \cE(X_0)$ implies $\underline W(0,X_0)\leq \beta_1 y_0-\beta_0 y_1\leq \overline W(0,X_0)$.}

Let $(y_0,y_1)\in\cE(X_0)$ and $(Z^0,Z^1)\in\cV (y_0,y_1)$, so that $\beta_0^{-1}Y^{0,0,X_0,y_0,Z^0}_T=\beta_1^{-1}Y^{1,0,X_0,y_1,Z^1}_T$, $\bar P$-a.s. Set $Z_s:=\beta_1 Z^0_s- \beta_0 Z^1_s$ and $\Delta_s:= \beta_1 Y^{0,0,X_0,y_0,Z^0}_s-\beta_0 Y^{1,0,X_0,y_1,Z^1}_s$. Subtracting the dynamics \eqref{fsde} for $\theta=0,1$, we obtain
\begin{align}\label{eq:deltap}
    \Delta_s
    =&\ \beta_1 y_0-\beta_0 y_1\\
    &-\int_0^s
    \bigl[
        \beta_1 H^0(r,X_r,Z^0_r)
        -\beta_0 H^1(r,X_r,Z^1_r)
        -\kappa(r,X_r)\Delta_r
    \bigr]\,dr \notag\\
    &+\int_0^s
    Z_r^\top dX_r .
    \notag
\end{align}

Since $\Delta_T=0$, we may write \eqref{eq:deltap} backward:
\begin{align}\label{eq:bsdeproof}
    \Delta_s =&\int_s^T \bigl[\beta_1 H^0(r,X_r,\frac{\beta_0 Z^1_r+Z_r}{\beta_1})- \beta_0 H^1(r,X_r,Z_r^1)-\kappa(r,X_r)\Delta_r\bigr]dr-\int_s^T Z_r^\top dX_r.
\end{align}
By the definitions of $\underline H$ and $\overline H$ in \eqref{low_H}-\eqref{high_H}, we have
\[
\underline H\bigl(r,x,\delta,z\bigr)\leq \beta_1 H^0\bigl(r,x,\frac{\beta_0 z_1+z}{\beta_1}\bigr)-\beta_0 H^1\bigl(r,x,z_1\bigr)-\kappa(r,x)\delta\leq \overline H\bigl(r,x,\delta,z\bigr),
\]
for all $(r,x,\delta,z) \in [0,T]\times \RR^d\times \RR\times \RR^d$. Hence,  $(\Delta,Z)$ is a subsolution of \eqref{eq:bsdeup} and a supersolution of \eqref{eq:bsdedown}. By Assumption~\ref{ass:boundary}(iii), we obtain
\[
\underline W(0,X_0)=\underline Y_0\leq \beta_1 y_0- \beta_0 y_1\leq \overline Y_0=\overline W(0,X_0).
\]

\emph{Step 2:} \textit{$\underline W(0,X_0)\leq  \beta_1 y_0- \beta_0 y_1 \leq \overline W(0,X_0)$ implies $(y_0,y_1)\in\cE(X_0)$.}

By Step~1, $\cE(X_0)\subseteq\{(y_0,y_1)\in\RR^2:\underline W(0,X_0)\leq \beta_1 y_0- \beta_0 y_1\leq \overline W(0,X_0)\}$. By the connectivity of $\cE(X_0)$ (Lemma~\ref{lem:bsde_wp}), it suffices to show that $\mathcal{E}(X_0)$ reaches the two boundary values. Assume $\beta_1 y_0- \beta_0 y_1=\underline W(0,X_0)=\underline Y_0$ (the upper case is analogous). By Assumption~\ref{ass:boundary}(i), the Kuratowski--Ryll-Nardzewski measurable selection theorem provides an $\FF^X$-progressively measurable process $Z^1$ with
\[
Z^1_s\in\underline\cV(s,X_s,\underline Z_s),\qquad dt\times d\bar P\text{-a.s.}
\]
Set $Z^0_s:=\frac{\beta_0 Z^1_s+\underline Z_s}{\beta_1}$. By Assumption~\ref{ass:boundary}(ii), $(Z^0,Z^1)\in(\cV(X_0))^2$. 
With this choice of controls,
\[
\begin{aligned}
&\beta_1 H^0(s,X_s,Z^0_s)-\beta_0 H^1(s,X_s,Z^1_s) \\
&\quad =
\inf_{z_1\in\RR^d}
\left\{
\beta_1 H^0\!\left(s,X_s,\frac{\beta_0 z_1+\underline Z_s}{\beta_1}\right)
-\beta_0 H^1(s,X_s,z_1)
\right\}  \\
&\quad =
\underline H(s,X_s,\Delta_s,\underline Z_s)
+\kappa(s,X_s)\Delta_s .
\end{aligned}
\]
Then, the forward dynamics \eqref{eq:deltap} of $\Delta$ reduces to that of $\underline Y$ driven by $\underline Z$:
\[
d\Delta_s=-\underline H(s,X_s,\Delta_s,\underline Z_s)\,ds+\underline Z_s^\top dX_s,\; \Delta_0=\underline Y_0.
\]
By uniqueness of this linear-in-$y$ SDE, $\Delta_s=\underline Y_s$ for all $s\in[0,T]$, and in particular $\Delta_T=\underline Y_T=0$. Hence, $(Z^0,Z^1)\in\cV (y_0,y_1)$, and $(y_0,y_1)\in\cE(X_0)$.
\\

Steps~1 and 2 show $\cE(X_0) = \{(y_0,y_1)\in \RR^2 : \underline W(0,X_0)\leq  \beta_1 y_0- \beta_0 y_1 \leq \overline W(0,X_0)\} $. 
\\

\emph{Step 3:} \textit{Part~(b): equivalences (i)--(iv).}

\smallskip
\emph{(iii)$\Leftrightarrow$(iv).} By Lemma~\ref{lem:bsde_wp}, the BSDE \eqref{eq:bsdek} is well-posed for any $\xi\in\cC_a$, with $\cY^{\theta,0,X_0,\xi}_T=\beta_\theta \xi$. If (iii) holds with some $\xi\in\cC_a$, then $ \beta_\theta^{-1} Y^{\theta,0,X_0,y_\theta,Z^\theta}_T= \beta_\theta^{-1} \cY^{\theta,0,X_0,\xi}_T=\xi$ for $\theta\in\{0,1\}$, so $\beta_0^{-1} Y^{0,0,X_0,y_0,Z^0}_T=\beta_1^{-1} Y^{1,0,X_0,y_1,Z^1}_T$, which is (iv). Conversely, if (iv) holds, set $\xi:=\beta_0^{-1} Y^{0,0,X_0,y_0,Z^0}_T=\beta_1^{-1} Y^{1,0,X_0,y_1,Z^1}_T$. Since each $Y^{\theta,0,X_0,y_\theta,Z^\theta}\in\bS_2(\bar P)$ (Remark~\ref{rem:selection}), $\xi\in\cC_a$, and for each $\theta$ the pair $(Y^{\theta,0,X_0,y_\theta,Z^\theta},Z^\theta)$ solves the BSDE \eqref{eq:bsdek} with terminal condition $\beta_\theta \xi$. By the uniqueness result in Lemma~\ref{lem:bsde_wp}, this pair must coincide with $(\cY^{\theta,0,X_0,\xi},\cZ^{\theta,0,X_0,\xi})$, which proves (iii).

\smallskip
\emph{(iv)$\Leftrightarrow$(i).} If (iv) holds, applying Part~(a) to the system shifted to $[s,T]$ and started at $$(s,X_s,Y^{0,0,X_0,y_0,Z^0}_s,Y^{1,0,X_0,y_1,Z^1}_s)$$ gives
\[
\underline W(s,X_s)\leq \beta_1 Y^{0,0,X_0,y_0,Z^0}_s-\beta_0 Y^{1,0,X_0,y_1,Z^1}_s\leq \overline W(s,X_s),\qquad s\in[0,T],\ \bar P\text{-a.s.},
\]
which is (i). Conversely, evaluating (i) at $s=T$ and using $\underline W(T,\cdot)=\overline W(T,\cdot)=0$ gives $\Delta_T=0$, i.e.\ (iv).

\smallskip
\emph{(i)$\Rightarrow$(ii).} We treat the lower-boundary case (the upper case is analogous). Fix $s\in[0,T]$, and work on the event $A_s:=\{\Delta_s=\underline W(s,X_s)\}=\{\Delta_s=\underline Y_s\}$. By (i) and Part (a), $\Delta_T=\underline Y_T=0$; on $[s,T]$ the pair $(\Delta,Z)$ satisfies the BSDE \eqref{eq:bsdeproof} restricted to $[s,T]$ and, by the definition of $\underline H$,
\[
\beta_1 H^0\!\left(
r,X_r,\frac{\beta_0 Z^1_r+Z_r}{\beta_1}
\right)
-\beta_0 H^1(r,X_r,Z^1_r)
-\kappa(r,X_r)\Delta_r  \geq
\underline H(r,X_r,\Delta_r,Z_r).
\]
so $(\Delta,Z)|_{[s,T]}$ is a supersolution of \eqref{eq:bsdedown} on $[s,T]$ with $\Delta_s=\underline Y_s$ on $A_s$. The strict comparison principle of Assumption~\ref{ass:boundary}(iii) then forces, on $A_s$,
\[
\Delta_r=\underline Y_r\qquad\text{for all }r\in[s,T],
\]
together with the equality of drivers
\[
\underline H(r,X_r,\underline Y_r,\underline Z_r)
=
\beta_1 H^0\!\left(
r,X_r,\frac{\beta_0 Z^1_r+\underline Z_r}{\beta_1}
\right)
-\beta_0 H^1(r,X_r,Z^1_r)
-\kappa(r,X_r)\underline Y_r,\quad dr\times d\bar P\text{-a.e.\ on }A_s.
\]
The first identity gives $\underline W(r,X_r)=\Delta_r$ for $r\in[s,T]$. Moreover, since $\Delta=\underline Y$ on $[s,T]$ and $\sigma$ is non-degenerate, identifying the diffusion coefficients of $\Delta$ and $\underline Y$ pathwise yields $Z_r=\underline Z_r$ for $dr\times d\bar P$-a.e.\ $r\in[s,T]$, i.e.\ $Z^0_r= \frac{\beta_0 Z^1_r+Z_r}{\beta_1}$. Substituting the definition of $\underline H$ and using $\underline Y_r=\Delta_r$, the equality of drivers rewrites as
\[
Z^1_r\in\argmin_{z_1\in\RR^d}\bigl( \beta_1 H^0\!\left(r,X_s,\frac{\beta_0 z_1+\underline Z_s}{\beta_1}\right)
-\beta_0 H^1(r,X_s,z_1)\bigr)=\underline\cV(r,X_r,\underline Z_r),\qquad r\in[s,T]\text{ on }A_s.
\]
This is precisely $(ii)$ on the lower boundary.

\smallskip
\emph{(ii)$\Rightarrow$(i).} Consider the stopping times 
\begin{align*}
    \tau^{-} &:= \inf \{ t\geq 0 : \underline{W}(t,X_t) = \beta_1 Y^{0,0,x_0,y_0,Z^0}_t-\beta_0 Y^{1,0,x_0,y_1,Z^1}_t\}, \\
    \tau^{+} &:= \inf \{ t\geq 0 : \overline{W}(t,X_t) = \beta_1 Y^{0,0,x_0,y_0,Z^0}_t-\beta_0 Y^{1,0,x_0,y_1,Z^1}_t\}.
\end{align*}
By part a), $\underline{W}(0,x_0)\leq \beta_1 y_0-\beta_0 y_1 \leq \overline{W}(0,x_0)$. Then, by the $\bar{P}-a.s.$ continuity in time of the processes $(X,Y^0,Y^1)$, together with the continuity of the function $\underline{W}$, we see 
that the state constraint is satisfied in the stochastic interval $[0,\tau^{-}\wedge \tau^{+} ]\subset [0,T]$. Next, using $ii)$, $\Delta$ satisfies the following dynamics restricted on the stochastic interval $s\in [\tau^{-}\wedge \tau^{+} ,T]$,
\begin{align*}
    &\Delta_s = \underline{W}(\tau^{-},X_{\tau^{-}})- \int_{\tau^{-}}^s \underline H(r,X_r, \Delta_r,\underline Z_r)dr+\int_{\tau^{-}}^s \underline{Z}_r^\top dX_r , \qquad \text{if } \tau^- \leq \tau^+, \\
    &\Delta_s = \overline{W}(\tau^{+},X_{\tau^{+}})- \int_{\tau^{+}}^s \overline H(r,X_r, \Delta_r,\overline Z_r)dr+\int_{\tau^{+}}^s \overline{Z}_r^\top dX_r ,\qquad  \text{ if } \tau^+ < \tau^-.
\end{align*}
Then, using that $\overline{H}$, and $\underline{H}$ are globally Lischitz in $\Delta$, we obtain using standard existence and uniqueness results of forward SDEs
\begin{align*}
    \Delta_t = \underline{W}(t,X_t) \mathbbm{1}_{\{\tau^{+}>\tau^{-}\} } +  \overline{W}(t,X_t) \mathbbm{1}_{\{\tau^{-}\geq \tau^{+}\} }, \quad \tau^{+}\wedge\tau^{-} \leq  t \leq T.
\end{align*}
Therefore, $\Delta$ satisfies the desired state constraint, which proves $(i)$.
\qed 
\end{proof}

\begin{proof}[Proof of Lemma~\ref{lem:verify_cd}]
Recall $N_0$ defined in \eqref{eq:n0}, and let $C\geq 0$ denote a generic constant whose value may change from line to line, depending only on $N_0$, $\rho$, $\|\partial_\alpha c\|_{L^\infty(\{0,1\}\times A)}$, $\|c\|_{L^\infty(\{0,1\}\times A)}$, and $\kappa$. To avoid clash with the cost function $c$, the level-set constant in item~3 of the lemma will be denoted by $\mathfrak c\in\{\underline a,\overline a\}$.

\medskip
\noindent\emph{Step 1: Lipschitz continuity of $A^\theta$ and $H^\theta$, and \eqref{eq:infw}.}

Fix $\theta\in\{0,1\}$. By Assumption~\ref{assumption:cd}, $\alpha\mapsto z\alpha-c(\theta,\alpha)$ is strictly concave on the compact interval $A$ for every $z\in\RR$, so the supremum defining $H^\theta(z)$ is attained at a unique maximizer $A^\theta(z)\in A$, and $A^\theta(z)\in A$ gives $\|A^0\|_\infty+\|A^1\|_\infty\leq 2N_0<\infty$.

The first-order optimality condition on the convex set $A$ reads
\[
\bigl(z-\partial_\alpha c(\theta,A^\theta(z))\bigr)\bigl(\alpha-A^\theta(z)\bigr)\leq 0,\qquad \alpha\in A.
\]
Setting $\alpha=A^\theta(z')$ in the inequality at $z$, then setting $\alpha=A^\theta(z)$ in the inequality at $z'$, and adding,
\[
(z-z')\bigl(A^\theta(z)-A^\theta(z')\bigr)
\geq
\bigl(\partial_\alpha c(\theta,A^\theta(z))-\partial_\alpha c(\theta,A^\theta(z'))\bigr)\bigl(A^\theta(z)-A^\theta(z')\bigr).
\]
The uniform strong convexity $\partial_{\alpha\alpha}c(\theta,\alpha)\geq\rho$ implies the right-hand side is at least $\rho|A^\theta(z)-A^\theta(z')|^2$, and hence
\[
|A^\theta(z)-A^\theta(z')|\leq \rho^{-1}|z-z'|,\qquad z,z'\in\RR.
\]
Thus $A^\theta$ is Lipschitz. For $H^\theta$, the supremum definition gives
\[
H^\theta(z)-H^\theta(z') = \bigl(zA^\theta(z)-c(\theta,A^\theta(z))\bigr)-H^\theta(z')\leq (z-z')A^\theta(z),
\]
and the symmetric bound shows $|H^\theta(z)-H^\theta(z')|\leq N_0\,|z-z'|$. Hence $H^\theta$ is Lipschitz with constant $N_0$.

It remains to show $-\infty<\underline a<\overline a<+\infty$. Comparing the two suprema using a common test point $\alpha\in A$,
\[
\inf_{\alpha\in A}\bigl[c(1,\alpha)-c(0,\alpha)\bigr]\leq H^0(z)-H^1(z)\leq \sup_{\alpha\in A}\bigl[c(1,\alpha)-c(0,\alpha)\bigr],\qquad z\in\RR,
\]
which by compactness of $A$ and continuity of $c$ gives $-\infty<\underline a$ and $\overline a<+\infty$. For the strict inequality $\underline a<\overline a$, suppose for contradiction that $H^0-H^1$ is constant on $\RR$. The maps $H^\theta$ are convex (as suprema of affine functions), hence differentiable almost everywhere, and the envelope theorem yields $(H^\theta)'(z)=A^\theta(z)$ at points of differentiability. Constancy of $H^0-H^1$ forces $A^0(z)=A^1(z)$ for a.e.\ $z\in\RR$, and by the continuity of $A^0,A^1$ proved above, $A^0\equiv A^1$ on $\RR$. Now both $A^\theta$ are continuous with $A^\theta(z)=a_{\max}$ for all $z$ sufficiently large positive and $A^\theta(z)=a_{\min}$ for all $z$ sufficiently large negative, so by the intermediate value theorem $A^\theta(\RR)=A$. For interior points $\alpha\in(a_{\min},a_{\max})$, choose $z\in\RR$ with $A^0(z)=A^1(z)=\alpha$; the interior first-order condition gives
\[
\partial_\alpha c(0,\alpha)=z=\partial_\alpha c(1,\alpha).
\]
By continuity of $\partial_\alpha c(\theta,\cdot)$, the equality extends to all of $A$, so $c(0,\cdot)-c(1,\cdot)$ is constant on $A$, contradicting Assumption~\ref{assumption:cd}. Hence $\underline a<\overline a$ and \eqref{eq:infw} hold.

\medskip
\noindent\emph{Step 2: First growth estimate.}

For $(Z^0,Z^1)\in\RR^2$,
\[
H^1(Z^1)-H^0(Z^0)=\bigl(H^1(Z^1)-H^1(Z^0)\bigr)+\bigl(H^1(Z^0)-H^0(Z^0)\bigr).
\]
The first term is bounded by $N_0|Z^0-Z^1|$ by Step~1, and the second by $\overline a-\underline a$ (or any uniform bound on $H^0-H^1$ from Step~1). Hence,
\[
\bigl|H^1(Z^1)-H^0(Z^0)\bigr|\leq C\bigl(1+|Z^0-Z^1|\bigr).
\]

\medskip
\noindent\emph{Step 3: The estimate involving $\bar\lambda$.}

Substituting $H^\theta(Z^\theta)=Z^\theta A^\theta(Z^\theta)-c(\theta,A^\theta(Z^\theta))$ and $\bar\lambda(p,Z^0,Z^1)=pA^0(Z^0)+(1-p)A^1(Z^1)$, a direct calculation gives
\begin{align}\label{eq:identity_step3}
H^0(Z^0)+H^1(Z^1)-\bar\lambda(p,Z^0,Z^1)(Z^0+Z^1)
&= \bigl((1-p)Z^0-pZ^1\bigr)\bigl(A^0(Z^0)-A^1(Z^1)\bigr) \notag\\
&\quad - c(0,A^0(Z^0))-c(1,A^1(Z^1)).
\end{align}
The cost terms are uniformly bounded by $2\|c\|_\infty$. For the first term, fix
\[
C_0\;:=\;\max_{\theta\in\{0,1\}}\max\bigl(|\partial_\alpha c(\theta,a_{\min})|,\;|\partial_\alpha c(\theta,a_{\max})|\bigr),
\]
so that, by the first-order condition and strong convexity, $A^\theta(z)=a_{\max}$ for all $z\geq C_0$ and $A^\theta(z)=a_{\min}$ for all $z\leq -C_0$, $\theta\in\{0,1\}$. If $Z^0,Z^1\geq C_0$ or $Z^0,Z^1\leq -C_0$, then $A^0(Z^0)=A^1(Z^1)$ and the first term vanishes. Otherwise, at least one of the following holds:
\begin{itemize}
\item $|Z^0|\leq C_0$, so $|Z^0|+|Z^1|\leq C_0+(|Z^0|+|Z^0-Z^1|)\leq 2C_0+|Z^0-Z^1|$;
\item $|Z^1|\leq C_0$, similarly $|Z^0|+|Z^1|\leq 2C_0+|Z^0-Z^1|$;
\item $Z^0$ and $Z^1$ have opposite signs, so $|Z^0|+|Z^1|=|Z^0-Z^1|$.
\end{itemize}
In all cases, $|Z^0|+|Z^1|\leq 2C_0+|Z^0-Z^1|\leq C(1+|Z^0-Z^1|)$. Combined with $|(1-p)Z^0-pZ^1|\leq |Z^0|+|Z^1|$ for $p\in[0,1]$ and $|A^0-A^1|\leq 2N_0$,
\[
\bigl|((1-p)Z^0-pZ^1)(A^0(Z^0)-A^1(Z^1))\bigr|\leq C\bigl(1+|Z^0-Z^1|\bigr).
\]
Substituting back into \eqref{eq:identity_step3} and adding the bound from Step~2 gives
\[
\bigl|H^1(Z^1)-H^0(Z^0)\bigr|+\bigl|H^0(Z^0)+H^1(Z^1)-\bar\lambda(p,Z^0,Z^1)(Z^0+Z^1)\bigr|\leq C\bigl(1+|Z^0-Z^1|\bigr).
\]

\medskip
\noindent\emph{Step 4: Independence of $p$ in the large-sum region.}

Take $C\geq \max(1,2C_0)$, where $C_0$ is the saturation threshold from Step~3. We show that the condition
\begin{equation}\label{eq:large_sum_cond}
|Z^0+Z^1|\geq C\bigl(1+|Z^0-Z^1|\bigr)
\end{equation}
forces $A^0(Z^0)=A^1(Z^1)$.

If $Z^0,Z^1$ have opposite signs, then $|Z^0+Z^1|\leq\max(|Z^0|,|Z^1|)\leq |Z^0-Z^1|$, contradicting \eqref{eq:large_sum_cond} for $C\geq 1$. So $Z^0,Z^1$ have the same sign, in which case $|Z^0+Z^1|=|Z^0|+|Z^1|$ and $|Z^0-Z^1|=\bigl||Z^0|-|Z^1|\bigr|$, giving
\[
\min(|Z^0|,|Z^1|)=\tfrac12\bigl(|Z^0+Z^1|-|Z^0-Z^1|\bigr)\geq \tfrac12\bigl[C+(C-1)|Z^0-Z^1|\bigr]\geq \tfrac{C}{2}\geq C_0.
\]
Combined with same-sign, both $Z^0,Z^1\geq C_0$ or both $\leq -C_0$, so $A^0(Z^0)=A^1(Z^1)\in\{a_{\min},a_{\max}\}$.

In this region, three simplifications occur:
\begin{itemize}
\item $\bar\lambda(p,Z^0,Z^1)=pA^0(Z^0)+(1-p)A^1(Z^1)=A^0(Z^0)$ — independent of $p$;
\item $\Delta\lambda(Z^0,Z^1)=A^0(Z^0)-A^1(Z^1)=0$, so the $p$-diffusion $p(1-p)\Delta\lambda$ vanishes, and consequently both the $(p,p)$ and $(y,p)$ entries of $\Sigma\Sigma^\top$ are zero;
\item $\ell(t,Z^0,Z^1,p)=\bar\lambda+\frac{e^{\kappa(T-t)}}{2}\bigl[H^0(Z^0)+H^1(Z^1)-\bar\lambda(Z^0+Z^1)\bigr]$ inherits independence of $p$ from $\bar\lambda$.
\end{itemize}
The $p$-dependence of $\mathcal L^{Z^0,Z^1}(t,0,p;q,N)$ enters \emph{only} through these three quantities (recall \eqref{eq:defLinside} with $y=0$, which removes the $\kappa y$ term). Hence $\mathcal L^{Z^0,Z^1}(t,0,p;q,N)$ is independent of $p$ whenever \eqref{eq:large_sum_cond} holds.

\medskip
\noindent\emph{Step 5: Approximation by exact level points.}

Let $\mathfrak c\in\{\underline a,\overline a\}$ and $(z^n)_{n\geq 1}\subset\RR$ with $H^0(z^n)-H^1(z^n)\to\mathfrak c$.

\emph{Case 1: $(z^n)$ has a bounded subsequence $(z^{n_k})$.} Passing to a further subsequence, $z^{n_k}\to z^\ast\in\RR$. By continuity of $H^0-H^1$, $H^0(z^\ast)-H^1(z^\ast)=\mathfrak c$. Set $\tilde z^k:=z^\ast$; then $H^0(\tilde z^k)-H^1(\tilde z^k)=\mathfrak c$ and $z^{n_k}-\tilde z^k\to 0$.

\emph{Case 2: $|z^n|\to\infty$.} Passing to a subsequence, either $z^{n_k}\to+\infty$ or $z^{n_k}\to-\infty$. Suppose $z^{n_k}\to+\infty$ (the other case is symmetric). With $C_0$ from Step~3, for all $k$ sufficiently large, $z^{n_k}\geq C_0$, so $A^0(z^{n_k})=A^1(z^{n_k})=a_{\max}$ and
\[
H^0(z^{n_k})-H^1(z^{n_k})=\bigl[z^{n_k}a_{\max}-c(0,a_{\max})\bigr]-\bigl[z^{n_k}a_{\max}-c(1,a_{\max})\bigr]=c(1,a_{\max})-c(0,a_{\max}),
\]
a constant independent of $k$. By the convergence assumption, this constant equals $\mathfrak c$. Set $\tilde z^k:=z^{n_k}$ for $k$ large; then $H^0(\tilde z^k)-H^1(\tilde z^k)=\mathfrak c$ and $z^{n_k}-\tilde z^k=0$.

\end{proof}

\begin{proof}[Proof of Proposition \ref{prop:supersol}]
We prove both inequalities on $|y|\le(\overline a+|\underline a|)(T-t)$, which contains $\mathrm{cl}_y(\cD)$.

\smallskip
\noindent\emph{Supersolution inequality.} Fix $(t,y,p)$ in this region. Since $\partial_t\phi=-\overline C$, $\partial_y\phi=-2y$, $\partial_{yy}\phi=-2$, expanding $\mathcal L^{z^0,z^1}$ via \eqref{eq:effective_running_reward} and grouping by $1$ and $\tfrac{e^{\kappa(T-t)}}{2}$ yields, for every $(z^0,z^1)\in\RR^2$,
\begin{align}\label{eq:proof-regroup}
-\mathcal L^{z^0,z^1}\bigl(t,y,p;\partial_y\phi,\partial_{yy}\phi\bigr)
&=2\kappa y^2+(z^0-z^1)^2-2y\bigl[H^0(z^0)-H^1(z^1)-\bar\lambda(p,z^0,z^1)(z^0-z^1)\bigr]\notag\\
&\quad-\tfrac{e^{\kappa(T-t)}}{2}\bigl[H^0(z^0)+H^1(z^1)-\bar\lambda(p,z^0,z^1)(z^0+z^1)\bigr]\notag\\
&\quad-\bar\lambda(p,z^0,z^1).
\end{align}
By Lemma \ref{lem:verify_cd}.2, both $|H^0(z^0)-H^1(z^1)|$ and $|H^0(z^0)+H^1(z^1)-\bar\lambda(p,z^0,z^1)(z^0+z^1)|$ are bounded by $C(1+|z^0-z^1|)$, while Lemma \ref{lem:verify_cd}.1 gives $|\bar\lambda(p,z^0,z^1)(z^0-z^1)|\le N_0|z^0-z^1|$ and $|\bar\lambda(p,z^0,z^1)|\le N_0$. Inserting these into \eqref{eq:proof-regroup},
\[
-\mathcal L^{z^0,z^1}\;\ge\;|z^0-z^1|^2-\Bigl[2|y|(C+N_0)+\tfrac{e^{\kappa(T-t)}C}{2}\Bigr]|z^0-z^1|+2\kappa y^2-2|y|C-\tfrac{e^{\kappa(T-t)}C}{2}-N_0.
\]
The right-hand side is a quadratic in $|z^0-z^1|\ge 0$ whose minimum over $\RR_+$ equals $-\tfrac14\bigl[2|y|(C+ N_0)+\tfrac{e^{\kappa(T-t)}C}{2}\bigr]^2$. Expanding the square and infimizing over $(z^0,z^1)\in\RR^2$,
\begin{align*}
\inf_{(z^0,z^1)}\bigl\{-\mathcal L^{z^0,z^1}\bigr\}
&\ge\bigl[2\kappa-(C+ N_0)^2\bigr]y^2-\Bigl[2C+\tfrac{(C+ N_0)e^{\kappa(T-t)}C}{2}\Bigr]|y|\\
&\quad-\tfrac{e^{\kappa(T-t)}C}{2}-\tfrac{e^{2\kappa(T-t)}C^2}{16}- N_0.
\end{align*}
On $|y|\le(\overline a+|\underline a|)(T-t)\le(\overline a+|\underline a|)T$, the $y^2$-term is bounded below by $-\bigl((C+ N_0)^2-2\kappa\bigr)_+(\overline a+|\underline a|)^2T^2$, and using $e^{\kappa(T-t)}\le e^{\kappa T}$ in the remaining terms gives
\begin{align*}
\inf_{(z^0,z^1)}\bigl\{-\mathcal L^{z^0,z^1}\bigr\}
&\ge-\bigl((C+ N_0)^2-2\kappa\bigr)_+(\overline a+|\underline a|)^2T^2-\Bigl[2C+\tfrac{(C+ N_0)C}{2}e^{\kappa T}\Bigr](\overline a+|\underline a|)T\notag\\
&\quad-\tfrac{C}{2}e^{\kappa T}-\tfrac{C^2}{16}e^{2\kappa T}- N_0.
\end{align*}
By \eqref{eq:defHinside} the left-hand side equals $H(t,y,p,\partial_y\phi,\partial_{yy}\phi)$, while by definition of $\overline C$ the right-hand side is $-(\overline C-1)$. Since $-\partial_t\phi=\overline C$,
\[
-\partial_t\phi+H(t,y,p,\partial_y\phi,\partial_{yy}\phi)\;\ge\;\overline C-(\overline C-1)\;=\;1,
\]
which is the first inequality of \eqref{eq:supersol-ineq}.

\smallskip
\noindent\emph{Subsolution inequality.} For $\psi$, $\partial_t\psi=-\underline C$, $\partial_y\psi=-2y$, $\partial_{yy}\psi=-2$. Specializing \eqref{eq:proof-regroup} to $z^0=z^1=0$ and applying Lemma \ref{lem:verify_cd}.2 (which yields $|H^0(0)-H^1(0)|\le C$ and $|H^0(0)+H^1(0)|\le C$, the $\bar\lambda$-term vanishing because $z^0+z^1=0$) together with $|\bar\lambda(p,0,0)|\le N_0$,
\[
-\mathcal L^{0,0}\bigl(t,y,p;\partial_y\psi,\partial_{yy}\psi\bigr)\;\le\;2\kappa y^2+2|y|C+\tfrac{e^{\kappa(T-t)}C}{2}+ N_0.
\]
Hence $H(t,y,p,\partial_y\psi,\partial_{yy}\psi)=\inf_{(z^0,z^1)}\{-\mathcal L^{z^0,z^1}\}\le-\mathcal L^{0,0}$, and on $|y|\le(\overline a+|\underline a|)(T-t)\le(\overline a+|\underline a|)T$ we obtain
\[
H(t,y,p,\partial_y\psi,\partial_{yy}\psi)\;\le\;2\kappa(\overline a+|\underline a|)^2T^2+2C(\overline a+|\underline a|)T+\tfrac{e^{\kappa T}C}{2}+ N_0\;=\;-\underline C-1.
\]
Therefore $-\partial_t\psi+H=\underline C+H\le-1\le 0$, which is the second inequality of \eqref{eq:supersol-ineq}.

\smallskip
\noindent\emph{A priori bound \eqref{eq:prioribound}.} At $t=T$, $\phi(T,y)=\psi(T,y)=-y^2=w(T,y,p)$, so the terminal conditions agree. The right-hand inequality of \eqref{eq:prioribound} then follows from the supersolution property combined with the verification theorem applied to any admissible control satisfying the state constraint. The left-hand inequality follows by inserting the constant control $z^0=z^1=0$ into the verification theorem, which by the bound on $-\mathcal L^{0,0}$ above yields $w(t,y,p)\ge\underline C(T-t)-y^2$ on $\mathrm{cl}_y(\cD)$.
\end{proof}




\medskip
\subsection{Proof of Proposition \ref{prop:boundary}}

We treat the supersolution and subsolution properties separately, as they
require different penalisation strategies.

\subsubsection{Supersolution properties }

We prove that $\underline w_*$ is a viscosity supersolution; the
proof for $\overline w_*$ is symmetric and given at the end of this part.

\medskip
\noindent\textbf{\ Supersolution property of $\underline w_*$.} Let $\phi\in C^{1,2}([0,T)\times(0,1))$ and let $(t_0,p_0)$ be a strict
minimum of $\underline w_*-\phi$ with $(\underline w_*-\phi)(t_0,p_0)=0$.

\smallskip
\noindent\emph{Step 1 (Penalised test function and convergence).}
For $n\ge1$ define on $\mathrm{cl}_y(\cD)$
\[
  \varphi_n(t,y,p):=\phi(t,p)-n\bigl(y-\underline W(t)\bigr).
\]
Since $y\ge\underline W(t)$ on $\mathrm{cl}_y(\cD)$, we have $\varphi_n\le\phi$,
with equality on~$\cD_d$.
Let $(t_n,y_n,p_n)$ be a minimiser of $w_*-\varphi_n$ on $\mathrm{cl}_y(\cD)$
and set $\delta_n:=y_n-\underline W(t_n)\ge0$.
The minimality at $(t_n,y_n,p_n)$ gives
\begin{align*}
    0=&\underline w_*(t_0,p_0)-\phi(t_0,p_0)=( w_*-\varphi_n)(t_0,\underline W(t_0),p_0)\\
    &\geq ( w_*-\varphi_n)(t_n,y_n,p_n)=w_*(t_n,y_n,p_n)-\phi(t_n,p_n)+n(y_n-\underline W(t_n)).
\end{align*}
By compactness, up an extraction, $(t_n,y_n,p_n)\to (t_\infty,\underline W(t_\infty),p_\infty)$ and  we have the chain of inequalities 
\begin{align*}
    0=&\underline w_*(t_0,p_0)-\phi(t_0,p_0)\geq \limsup_n w_*(t_n,y_n,p_n)-\phi(t_n,p_n)+n(y_n-\underline W(t_n))\\
    &\geq \liminf_n w_*(t_n,y_n,p_n)-\phi(t_n,p_n)+\liminf_n n(y_n-\underline W(t_n))\\
    &\geq  w_*(t_\infty,\underline W(t_\infty),p_\infty)-\phi(t_\infty,p_\infty)+\liminf_n n(y_n-\underline W(t_n))\\
    &\geq \underline w_*(t_\infty,p_\infty)-\phi(t_\infty,p_\infty)
\end{align*}
which is larger than $\underline w_*(t_0,p_0)-\phi(t_0,p_0)$ by the strict minimality of $(t_0,p_0)$ for this function. Thus, all these values are equal, and after extraction, 
$$(t_n,p_n)\to(t_0,p_0),\,\delta_n\to0,\,n\delta_n\to0.$$

We now treat the two cases $\delta_n>0$ and $\delta_n=0$ separately;
both lead to the same conclusion.

\smallskip
\noindent\emph{Step 2 (Case $\delta_n>0$: interior point).}
If $y_n>\underline W(t_n)$ then $(t_n,y_n,p_n)\in\cD$, and $w_*$ satisfies
the viscosity supersolution inequality at this interior minimum point:
\[
  -\partial_t\varphi_n(t_n,y_n,p_n)+H^*\!\bigl(t_n,y_n,p_n,\,\partial_y\varphi_n,\,\nabla^2\varphi_n\bigr)\;\ge\;0.
\]
In particular, for any $z^n\in\underline\cV$, since
$\underline\cV\subset\{(z^n,z^n)\}$:
\[
  -\partial_t\varphi_n(t_n,y_n,p_n)-
L^{z^n,z^n}\!\bigl(t_n,y_n,p_n,\,\partial_y\varphi_n,\,\nabla^2\varphi_n\bigr)\;\ge\;0.
\]
Now $\partial_y\varphi_n=-n$, $\partial_{yy}^2\varphi_n=0$,
$\partial_t\varphi_n=\partial_t\phi+n\,\underline W'$,
and evaluating at $Z^0=Z^1=z^n$ cancels the terms involving $Z^0-Z^1$.
The inequality reduces to
\[
  -\partial_t\phi(t_n,p_n)
  \;\ge\;
  -n\bigl[{-H^0(z^n)+H^1(z^n)+\underline a+\kappa\delta_n
   }\bigr]
  +L^{z^n,z^n}(t_n,p_n,0,\partial_{pp}^2\phi(t_n,p_n)).
\]
Thanks $z^n\in\underline\cV$, we obtain ${-H^0(z^n)+H^1(z^n)+\underline a=0}$ and 
\begin{align*}
     &-\partial_t\phi(t_n,p_n)-L^{z^n,z^n}(t_n,p_n,0,\partial_{pp}^2\phi(t_n,p_n))\\
  & \geq-\partial_t\phi(t_n,p_n)+\underline H(t_n,p_n,\partial_{pp}^2\phi(t_n,p_n))
\;\ge\;-n\bigl[\kappa\delta_n\bigr].
\end{align*}
Sending $n\to\infty$ yields
\begin{equation}\label{eq:super-interior}
  -\partial_t\phi(t_0,p_0)+\underline H^*(t_0,p_0,\partial_{pp}^2\phi(t_0,p_0))\;\ge\;0.
\end{equation}

\smallskip
\noindent\emph{Step 3 (Case $\delta_n=0$: boundary point).}
If $y_n=\underline W(t_n)$ then $(t_n,y_n,p_n)\in\cD_d$ and the interior
supersolution property is not available. We appeal instead to the DPP.

Fix $z\in\underline\cV$. Thanks to \eqref{eq:defv}, 
$-H^0(z)+H^1(z)+\kappa\underline W(t_0)
    -\underline W'(t_0)=0$, the constant control
$Z^0_s=Z^1_s=z$ has zero volatility for~$Y$ and satisfies
$d(Y_s-\underline W(s))=\kappa(Y_s-\underline W(s))\,ds$
with $Y_{t_n}=\underline W(t_n)$, giving $Y_s=\underline W(s)$ for
all $s\ge t_n$.

{Take $(t_n^k,y_n^k,p_n^k)\in\cD$ with
$y_n^k\downarrow\underline W(t_n)$ and
$w(t_n^k,y_n^k,p_n^k)\to w_*(t_n,\underline W(t_n),p_n)$.
Applying the DPP lower bound~\eqref{eq:DPP_w_lower} at
$(t_n^k,y_n^k,p_n^k)$ with control $z\in\underline\cV$ and time increment
$\theta=t_n^k+h$ and using the global minimum property
$w_*\ge\varphi_n+(\underline w_*-\phi)(t_n,p_n)$ on $\mathrm{cl}_y(\cD)$,
we have 
\begin{align*}
     &w(t_n^k,y_n^k,p_n^k)\;\ge\;
  \EE\!\left[\int_{t^k_n}^{t^k_n+h}\!\ell(s,z,z,p_s)\,ds
  +w_*(t^k_n+h,Y_{t^k_n+h},p_{t^k_n+h})\right]\\
  \;&\ge\;\EE\!\left[\int_{t^k_n}^{t^k_n+h}\!\ell(s,z,z,p_s)\,ds
  +\varphi_n(t^k_n+h,Y_{t^k_n+h},p_{t^k_n+h})\right]+(\underline w_*-\phi)(t_n,p_n).
\end{align*}
We take the limit $k\to\infty$ to obtain
\[
  \phi(t_n,p_n)\;\ge\;
  \EE\!\left[\int_{t_n}^{t_n+h}\!\ell(s,z,z,p_s)\,ds
  +\varphi_n(t_n+h,Y_{t_n+h},p_{t_n+h})\right],
\]
where recalling \eqref{eq:defSigma} $dp_s=\Sigma_2(p_s,z,z)\,dB_s$ with $p_{t_n}=p_n$.
Due to the choice of control $Y_{t_n+h}=\underline W_{t_n+h}$ and therefore 
$\varphi_n(t_n+h,Y_{t_n+h},p_{t_n+h})=\phi(t_n+h,p_{t_n+h})$ which leads to 
\[
  \phi(t_n,p_n)\;\ge\;
  \EE\!\left[\int_{t_n}^{t_n+h}\!\ell(s,z,z,p_s)\,ds
  +\phi(t_n+h,p_{t_n+h})\right],
\]
}Applying It\^o's formula to $\phi(s,p_s)$ and recalling that
$L^{z,z}(t,p,0,\partial_{pp}^2\phi)
=\tfrac12\Sigma_2(p_s,z,z)^2\partial_{pp}^2\phi+\ell(t,z,z,p)$:
\[
  0\;\ge\;\EE\int_{t_n}^{t_n+h}\!
  \bigl[\partial_t\phi(s,p_s)
  +L^{z,z}(s,p_s,0,\partial_{pp}^2\phi(s,p_s))\bigr]\,ds.
\]
Dividing by~$h$, sending $h\downarrow0$, then $n\to\infty$
gives~\eqref{eq:super-interior} again.
\medskip

\noindent\textbf{ Supersolution property of $\overline w_*$.}
The proof is symmetric. Define
$\varphi_n(t,y,p):=\phi(t,p)-n(\overline W(t)-y)$.
The minimiser of $w_*-\varphi_n$ on $\mathrm{cl}_y(\cD)$ satisfies
$y_n=\overline W(t_n)$.

Either the viscosity supersolution property at $\cD$ or the DPP lower bound at the upper boundary with control $z\in\overline\cV$
then yields, by the same It\^o argument to 
\[
  -\partial_t\phi(t_0,p_0)+\overline H^*(t_0,p_0,\partial_{pp}^2\phi(t_0,p_0))\geq0.
\]

\subsubsection{Subsolution properties}

We prove that $\overline w^*$ is a viscosity subsolution; the proof for
$\underline w^*$ is symmetric and stated at the end.

\medskip
\noindent\textbf{ Subsolution property of $\overline w^*$.} Let $\phi\in C^{1,2}([0,T]\times[0,1])$ and let $(t_0,p_0)$ be a strict
maximum of $\overline w^*-\phi$ with $(\overline w^*-\phi)(t_0,p_0)=0$ with $t_0<T$. 

\smallskip
\noindent\emph{Step 1 (Penalised test function).}
For $n\ge1$ and $\varepsilon>0$, define on
$\mathrm{cl}_y(\cD)\cap\{0\le\overline W(t)-y\le n^{-1}\}$
\[
  \varphi_n(t,y,p):=\phi(t,p)
  +\varepsilon n\bigl(\overline W(t)-y\bigr)
  +\varepsilon n^2\bigl(\overline W(t)-y\bigr)
    \bigl(y-\overline W(t)+n^{-1}\bigr).
\]
Writing $\delta:=\overline W(t)-y\in[0,n^{-1}]$, this simplifies to
$\varphi_n=\phi+2\varepsilon n\delta-\varepsilon n^2\delta^2$.
Its relevant derivatives are
\begin{alignat*}{2}
  \partial_y\varphi_n &= -2\varepsilon n+2\varepsilon n^2\delta
    =:-2\varepsilon n(1-n\delta), &\qquad
  \partial_{yy}^2\varphi_n &= -2\varepsilon n^2,\\
  \partial_p\varphi_n &= \partial_p\phi, &\qquad
  \partial_{pp}^2\varphi_n &= \partial_{pp}^2\phi,\\
  \partial_t\varphi_n &= \partial_t\phi+2\varepsilon n(1-n\delta)\,\overline W'(t), &\qquad
  \partial_{yp}^2\varphi_n &= 0.
\end{alignat*}

\smallskip
\noindent\emph{Step 2 (Maximiser and convergence).}
Let $(t_n,y_n,p_n)$ be a maximiser of $w^*-\varphi_n$ on
$\mathrm{cl}_y(\cD)\cap\{0\le\delta\le n^{-1}\}$.
Set $\delta_n:=\overline W(t_n)-y_n$.
A standard argument (using $w^*-\varphi_n\le w^*-\phi\le0$ near
$(t_0,\overline W(t_0),p_0)$ together with the fact that $\varphi_n=\phi$
at $\delta=0$) yields, after extracting a subsequence,
\begin{equation}\label{eq:sub-conv}
  (t_n,p_n)\to(t_0,p_0),\quad n\delta_n\to0,\quad
  w^*(t_n,y_n,p_n)\to\phi(t_0,p_0).
\end{equation}

\smallskip
\noindent\emph{Step 3 (Subsolution inequality).}
Since $w^*$ is a viscosity subsolution of~\eqref{eq:hjb} on
$\mathrm{cl}_y(\cD)$, at the maximum point of $w^*-\varphi_n$
there exist near-optimisers $(Z^0_n,Z^1_n)\in\RR^2$ such that
\begin{equation}\label{eq:sub-ineq}
  -\partial_t\varphi_n(t_n,y_n,p_n)
  \;\le\;
  L^{Z^0_n,Z^1_n}\!\bigl(t_n,y_n,p_n,\,\partial_y\varphi_n,\,\nabla^2\varphi_n\bigr)
  +n^{-1}.
\end{equation}
Write $d_n:=Z^0_n-Z^1_n$, $\hat a_n:=\bar \lambda(p_n,Z^0_n,Z^1_n)$.
Substituting the derivatives from Step~1 and rearranging, the dominant
contributions to~\eqref{eq:sub-ineq} are
\begin{align}
  \notag-\partial_t\phi(t_n,p_n)
  &\;\le\;  -\;\varepsilon n^2 d_n^2
  \;+\;R_n\;+{\;2n\varepsilon(1-n\delta_n)\kappa \delta_n}+n^{-1}\\
  &
  -2\varepsilon(1-n\delta_n)\,
  \bigl[
    \bigl(-nH^0(Z^0_n)+nH^1(Z^1_n)+n\overline{a}+n\hat a_n\,d_n\bigr)
  \bigr]
  \label{eq:sub-expanded}
\end{align}
where $R_n$ collects the terms
\[
  R_n:=\tfrac12\Sigma_2(p_n,Z^0_n,Z^1_n)^2\partial_{pp}^2\phi(t_n,p_n)
  +\hat a_n
  +\tfrac{e^{\kappa(T-t_n)}}{2}
    \bigl[H^0(Z^0_n)+H^1(Z^1_n)-\hat a_n(Z^0_n+Z^1_n)\bigr]
\]
and the term involving $\kappa\delta_n$
vanishes by~\eqref{eq:sub-conv}.

\smallskip
\noindent\emph{Step 4 ($d_n=Z^0_n-Z^1_n\to0$ via Assumption~\ref{assumption:cd}(ii)).}
Given the Assumption \ref{assumption:cd} (2)-(3), we have that 
\begin{align*}
    \bigl[{
    \bigl(-nH^0(Z^0_n)+nH^1(Z^1_n)+n\overline{a}+n\hat a_n\,d_n\bigr)
  }\bigr]&\geq -C n(1+|d_n|) \\
  |R_n|&\leq C (1+|d_n|)
\end{align*}

Thus, for $n$ large enough 
\eqref{eq:sub-expanded}
gives
\begin{align*}
  -\partial_t\phi(t_n,p_n)
  &\;\le\;
  n\left(2C\varepsilon(1-n\delta_n)(1+|d_n|)-\varepsilon n d_n^2\right)+C,
\end{align*}
and we obtain that 
\begin{align}\label{eq:bz}
    |Z^0_n-Z^1_n|= |d_n|\leq \frac{C}{\sqrt{n}},
\end{align} 
otherwise the right hand side explodes to $-\infty$. 

\noindent\emph{Step 5: Perturbing $Z^0_n$ }
We denote 
{$$\nu_z:=-H^0(z)+H^1(z)+\overline{a}=-H^0(z)+H^1(z)+\sup_{z}H^0(z)-H^1(z)\geq 0$$}
where we used \eqref{def:w} to obtain the sign.
We use Lipschitz continuity (to pass from $H^1(Z^1_n)$ to $H^1(Z^0_n)$), \eqref{eq:sub-expanded}, and \eqref{eq:bz} to obtain that 
\begin{align*}
  -\partial_t\phi(t_n,p_n)
  &\;\le\; C(1+n|d_n|)
  -2\varepsilon(1-n\delta_n)\,
  n\nu_{Z^n_0}
 \\
  &\;\le\; C(1+\sqrt{n})-2\varepsilon(1-n\delta_n)\,
  n\nu_{Z^n_0},
\end{align*}
Thus, $\limsup \nu_{Z^n_0}\leq 0$ otherwise the right hand side explodes to $-\infty$ as $-cn +\sqrt{n}$.
which also implies that 
$$\lim_n -H^0(Z_n^0)+H^1(Z_n^0)+\overline a=0.$$

Using Assumption \ref{assumption:cd}, and taking a subsequence we can find $\tilde Z_n^0\in \cV$ so that 
$|\tilde Z_n^0- Z_n^0|\to 0$. 
Given the definition of $R_n$, since we have 
\begin{align*}
  &L^{\tilde Z^0_n,\tilde Z^0_n}(t_n,p_n,0,\partial_{pp}^2\phi(t_n,p_n))\\
  &=\tfrac12\Sigma_2(p_n,\tilde Z^0_n,\tilde Z^0_n)^2\partial_{pp}^2\phi(t_n,p_n)
  +\bar \lambda(p_n, \tilde Z^0_n,\tilde Z^0_n)
  \\
  &+\tfrac{e^{\kappa(T-t_n)}}{2}
    \bigl[H^0(\tilde Z^0_n)+H^1(\tilde Z^0_n)-\bar \lambda(p_n, \tilde Z^0_n,\tilde Z^0_n)(\tilde Z^0_n+\tilde Z^0_n)\bigr],    
\end{align*}
the convergence $|\tilde Z_n^0- Z_n^0|\to 0$ gives
$R_n-L^{\tilde Z^0_n,\tilde Z^0_n}(t_n,p_n,0,\partial_{pp}^2\phi(t_n,p_n))\to 0$. 

\smallskip
\noindent\emph{Step 6 (Passing to the limit).}
We inject these quantities in \eqref{eq:sub-expanded} to obtain that 
\begin{align*}
  -\partial_t\phi(t_n,p_n)
  &\;\le\;L^{\tilde Z^0_n,\tilde Z^0_n}(t_n,p_n,0,\partial_{pp}^2\phi(t_n,p_n))
  \\
  &-2n\varepsilon(1-n\delta_n)\,
  \bigl[
    \bigl(-H^0(Z^0_n)+H^1(Z^0_n)+\kappa\overline W(t_n)
    -\overline W'(t_n)\bigr)
  \bigr]
  \notag\\
  &
  -\;\varepsilon n^2 d_n^2
  \;+\;R_n-L^{\tilde Z^0_n,\tilde Z^0_n}(t_n,p_n,0,\partial_{pp}^2\phi(t_n,p_n))\;\\
  &+\;2\varepsilon(1-n\delta_n)\kappa \delta_n+n^{-1}\\
  &\;\le\;-\overline H(t_n,p_n,\partial_{pp}^2\phi(t_n,p_n))
  \\
  &-2n\varepsilon(1-n\delta_n)\,
  \bigl[
    \bigl(-H^0(Z^0_n)+H^1(Z^0_n)+\kappa\overline W(t_n)
    -\overline W'(t_n)\bigr)
  \bigr]
  \notag\\
  &
  -\;\varepsilon n^2 d_n^2
  \;+\;R_n-L^{\tilde Z^0_n,\tilde Z^0_n}(t_n,p_n,0,\partial_{pp}^2\phi(t_n,p_n))\;\\
  &+\;2\varepsilon(1-n\delta_n)\kappa \delta_n+n^{-1}.
\end{align*}
We first send $n\to \infty$ and send $\epsilon \to 0$ to obtain
\[
  {-\partial_t\phi(t_0,p_0)+\overline H_* (t_0,p_0,\partial_{pp}^2\phi(t_0,p_0))
  \;\le\;0}.
\]

\bigskip
\noindent\textbf{ Subsolution property of $\underline w^*$.}
The proof is symmetric, with the roles of the boundaries reversed.
Define
$\varphi_n(t,y,p):=\phi(t,p)
+\varepsilon n(y-\underline W(t))
+\varepsilon n^2(y-\underline W(t))(\underline W(t)-y+n^{-1})$
and work on $\{0\le y-\underline W(t)\le n^{-1}\}$.
The derivatives are $\partial_y\varphi_n=2\varepsilon n(1-n\delta)$
(with $\delta=y-\underline W$) and $\partial_{yy}^2\varphi_n=2\varepsilon n^2$.
The near-optimisers satisfy $Z^0_n-Z^1_n\to0$ 
together with the definition of lower bound and
the same argument as in Step~5 above.
Passing to the limit yields
\[
  -\partial_t\phi(t_0,p_0)
  +\underline H_*(t_0,p_0,\partial_{pp}^2\phi(t_0,p_0))\;\le\;0.
\]

\subsection{Proof of the Comparison Theorem}

Following the exponential change of variable we define rescaled generator is defined by 
\begin{align*}
L^{\lambda,z^0,z^1}(t,y,p;q,N)
&:=
\Bigl[
    -H^0(z^0)+H^1(z^1)+\kappa y
    +\bar\lambda(p,z^0,z^1)(z^0-z^1)
\Bigr]q \notag \\
&\quad
+\frac12\,\mathrm{tr}\bigl(\Sigma\Sigma^\top N\bigr)
+e^{\lambda t}\ell(t,z^0,z^1,p).
\end{align*}

Correspondingly, we have 
\begin{equation}\label{eq:defHauxinside}\notag
H^\lambda(t,y,p, q ,A) 
:= \inf_{(z^0,z^1)\in\RR^2}\bigl\{-L^{\lambda,z^0,z^1}(t,y,p, q ,A)\bigr\},
\end{equation}
where for $(t,y,p, q ,A)\in \mathrm{cl}(\mathcal{D})\times \RR\times \cS_2$.

\subsubsection{Preliminary results}

We need the following lemma, where we denote $\cS_2^-$ the set of
symmetric negative semi-definite matrices of dimension 2.

\begin{lemma}[Hamiltonian comparison]\label{lem:Hcomp}
Let Assumption~\ref{assumption:cd} hold.
There exists a constant $C>0$
(depending only on $\kappa$, $T$, $\lambda$, and the
constants in Assumption~\ref{assumption:cd}) such that the following holds.

For any $\varepsilon>0$, $t\in[0,T]$, $y,\bar y\in[\underline W(t),\overline W(t)]$,
$p,\bar p\in[0,1]$, $q,\bar q\in\RR$, $X,Y\in\cS_2^-$ with
$X_{yy}\le -\varepsilon$ satisfying
\begin{equation}\label{eq:trace_hyp}
r:=\sup_{(Z^0,Z^1)\in\RR^2}
\tfrac12\,\Sigma(\bar p,Z^0,Z^1)^\top X\,\Sigma(\bar p,Z^0,Z^1)
-\tfrac12\,\Sigma(p,Z^0,Z^1)^\top Y\,\Sigma(p,Z^0,Z^1)
+\varepsilon(Z^0-Z^1)^2<\infty,
\end{equation}
with $\Sigma$ as in \eqref{eq:defSigma}, we have
\begin{equation}\label{eq:H_comp}
H^\lambda(t,y,p,q,Y)-H^\lambda(t,\bar y,\bar p,\bar q,X)
\;\le\;r+C\,\frac{\delta^2}{\varepsilon}+C\,\delta,
\end{equation}
where $\delta:=|q-\bar q|+(|y-\bar y|+|p-\bar p|)(1+|\bar q|)$.
\end{lemma}

\begin{proof}
We first note that $H^\lambda(t,\bar y,\bar p,\bar q,X)>-\infty$. Indeed,
since $X\in\cS_2^-$ with $X_{yy}\le -\varepsilon$, the trace contribution
$-\tfrac12\,\Sigma(\bar p,z^0,z^1)^\top X\,\Sigma(\bar p,z^0,z^1)$ to
$-L^{\lambda,z^0,z^1}(t,\bar y,\bar p,\bar q,X)$ is non-negative and
quadratically coercive in $z^0-z^1$ with coefficient $\ge\varepsilon/2$,
while the drift and running-cost components grow at most linearly in
$z^0-z^1$, uniformly in $z^0+z^1$, by
Lemma~\ref{lem:verify_cd}(1)--(2). The infimum defining
$H^\lambda$ is therefore attained on a compact set, in particular finite.

\medskip\noindent\textit{Step 1 (Near-optimiser).}
Let $(\bar Z^0,\bar Z^1)$ be an $\eta$-optimiser for
$H^\lambda(t,\bar y,\bar p,\bar q,X)$ with $\eta>0$ arbitrary:
\[
-L^{\lambda,\bar Z^0,\bar Z^1}(t,\bar y,\bar p,\bar q,X)
\;\le\;H^\lambda(t,\bar y,\bar p,\bar q,X)+\eta.
\]
Set
\[
\bar\Delta:=\bar Z^0-\bar Z^1,\qquad
\bar\alpha:=A^0(\bar Z^0)-A^1(\bar Z^1),
\]
\[
\mu_0(p):=-H^0(\bar Z^0)+H^1(\bar Z^1)
+\bar\lambda(p,\bar Z^0,\bar Z^1)\,\bar\Delta.
\]
Since $H^\lambda\le -L^{\lambda,\bar Z^0,\bar Z^1}$ at every point,
\begin{equation}\label{eq:lem_diff}
H^\lambda(t,y,p,q,Y)-H^\lambda(t,\bar y,\bar p,\bar q,X)
\;\le\;\mathcal{E}+\eta,
\end{equation}
where
$\mathcal{E}:=-L^{\lambda,\bar Z^0,\bar Z^1}(t,y,p,q,Y)
+L^{\lambda,\bar Z^0,\bar Z^1}(t,\bar y,\bar p,\bar q,X)$.
The $y$-dependence of $L^{\lambda,\bar Z^0,\bar Z^1}$ is the additive
contribution $\kappa y\,q$, so
\[
\mathcal{E}=\underbrace{
-L^{\lambda,\bar Z^0,\bar Z^1}(t,0,p,q,Y)
+L^{\lambda,\bar Z^0,\bar Z^1}(t,0,\bar p,\bar q,X)}_{=:\,\mathcal{E}^0}
+\underbrace{\kappa(\bar y\bar q-yq)}_{=:\,\mathcal{R}}.
\]

\medskip\noindent\textit{Step 2 (The $y$-remainder $\mathcal{R}$).}
Splitting $\bar y\bar q-yq=y(\bar q-q)+(\bar y-y)\bar q$,
\[
|\mathcal{R}|\le\kappa\bigl(|y|\,|q-\bar q|+|\bar y-y|\,|\bar q|\bigr)
\le C\delta,
\]
since $|y|\le C$ on the bounded $y$-domain and
$|q-\bar q|+|\bar y-y|\,|\bar q|\le\delta$.

\medskip\noindent\textit{Step 3 (Decomposition of $\mathcal{E}^0$).}
We decompose $\mathcal{E}^0$ into trace, gradient mismatch at fixed~$p$,
and $p$-dependent contributions:
\begin{align*}
\mathcal{E}^0
&=\underbrace{\tfrac12\,\Sigma(\bar p,\bar Z^0,\bar Z^1)^\top X\,\Sigma(\bar p,\bar Z^0,\bar Z^1)
-\tfrac12\,\Sigma(p,\bar Z^0,\bar Z^1)^\top Y\,\Sigma(p,\bar Z^0,\bar Z^1)}_{\text{(i) trace}}\\
&\quad+\underbrace{\mu_0(p)(\bar q-q)}_{\text{(ii) gradient mismatch}}
+\underbrace{[\mu_0(\bar p)-\mu_0(p)]\bar q
+e^{\lambda t}[\ell(\bar p)-\ell(p)]}_{\text{(iii) $p$-dependent}},
\end{align*}
where in (iii) the running cost $\ell$ is evaluated at
$(t,\bar Z^0,\bar Z^1,\cdot)$.

\smallskip\noindent\textit{(i) Trace.}
By hypothesis \eqref{eq:trace_hyp} evaluated at $(\bar Z^0,\bar Z^1)$,
\[
\tfrac12\,\Sigma(\bar p,\bar Z^0,\bar Z^1)^\top X\,\Sigma(\bar p,\bar Z^0,\bar Z^1)
-\tfrac12\,\Sigma(p,\bar Z^0,\bar Z^1)^\top Y\,\Sigma(p,\bar Z^0,\bar Z^1)
\le -\varepsilon\bar\Delta^2+r.
\]

\smallskip\noindent\textit{(ii) Gradient mismatch.}
By Assumption~\ref{assumption:cd}(2),
$|H^0(\bar Z^0)-H^1(\bar Z^1)|\le C(1+|\bar\Delta|)$,
and $|\bar\lambda(p,\bar Z^0,\bar Z^1)\,\bar\Delta|\le \bar a\,|\bar\Delta|$,
so $|\mu_0(p)|\le C(1+|\bar\Delta|)$ and
$|\mu_0(p)(\bar q-q)|\le C(1+|\bar\Delta|)\,|q-\bar q|$.

\smallskip\noindent\textit{(iii) $p$-dependent terms: case distinction.}

\noindent\emph{Case 1: $|\bar Z^0+\bar Z^1|<C(1+|\bar\Delta|)$.}
By Assumption~\ref{assumption:cd}(1), $|\bar\alpha|\le \bar a$.
The $p$-variation of the drift,
\[
|\mu_0(\bar p)-\mu_0(p)|\leq |\bar p-p|\,|\bar\alpha|\,|\bar\Delta|,
\]
contributes $|\mu_0(\bar p)-\mu_0(p)|\,|\bar q|
\le C|\bar\Delta|\,|p-\bar p|\,|\bar q|$.
The running-cost difference,
\[
\ell(\bar p)-\ell(p)
=(\bar p-p)\,\bar\alpha\,
\bigl[1-\tfrac{e^{\kappa(T-t)}}{2}(\bar Z^0+\bar Z^1)\bigr],
\]
satisfies $|\bar\alpha|(1+|\bar Z^0+\bar Z^1|)\le C(1+|\bar\Delta|)$
by the case condition, hence
$e^{\lambda t}|\ell(\bar p)-\ell(p)|\le C(1+|\bar\Delta|)\,|p-\bar p|$.
Together,
\[
\text{(iii)}\le C(1+|\bar\Delta|)\,|p-\bar p|\,(1+|\bar q|).
\]

\noindent\emph{Case 2: $|\bar Z^0+\bar Z^1|\ge C(1+|\bar\Delta|)$.}
By Assumption~\ref{assumption:cd}(2),
$L^{\lambda,\bar Z^0,\bar Z^1}(t,0,p,q,X)$ is independent of~$p$, so the
$p$-dependent contribution in (iii) vanishes.

In both cases,
$\text{(iii)}\le C(1+|\bar\Delta|)\,|p-\bar p|\,(1+|\bar q|)$.

\medskip\noindent\textit{Step 4 (Absorbing $|\bar\Delta|$).}
Collecting Steps~2--3,
\[
\mathcal{E}\;\le\;-\varepsilon\bar\Delta^2+r
+C(1+|\bar\Delta|)\underbrace{\bigl[|q-\bar q|
+|p-\bar p|(1+|\bar q|)\bigr]}_{=:\,\delta_0}+C\delta.
\]
Note $\delta_0\le\delta$. By Young's inequality,
\[
(1+|\bar\Delta|)\delta_0
=\delta_0+|\bar\Delta|\,\delta_0
\le\delta_0+\tfrac{\varepsilon}{2}\bar\Delta^2
+\tfrac{1}{2\varepsilon}\delta_0^2.
\]
Substituting and using $\delta_0\le\delta$,
\[
\mathcal{E}\;\le\;-\tfrac{\varepsilon}{2}\bar\Delta^2
+r+\tfrac{C}{2\varepsilon}\delta_0^2+C\delta_0+C\delta
\;\le\;r+C\,\delta^2/\varepsilon+C\,\delta.
\]
Sending $\eta\to 0$ in \eqref{eq:lem_diff} yields \eqref{eq:H_comp}.
\end{proof}

\subsubsection{Proof of Theorem \ref{thm:compinside}}

\begin{proof}
We prove that \(u\le v\) on \(\mathrm{cl}_y(\cD)\). Suppose, to the contrary, that
\[
\sup_{\mathrm{cl}_y(\cD)}(u-v)>0.
\]

\medskip\noindent
\textbf{Step 1.}
Fix \(\lambda>0\) and define
\[
f(p):=\sqrt{-\ln p}+\sqrt{-\ln(1-p)},\qquad p\in(0,1).
\]
Then \(f(p)\to+\infty\) as \(p\to0^+\) and as \(p\to1^-\). Moreover,
\[
\sup_{p\in(0,1)}p^2(1-p)^2|f''(p)|<\infty.
\]
Set
\begin{equation}\label{eq:KH_def}
K_H
:=
\frac12\,\bar a^{\,2}e^{\lambda T}
\sup_{p\in(0,1)}p^2(1-p)^2|f''(p)|<\infty,
\end{equation}
where
\[
\bar a:=\|A^0\|_\infty+\|A^1\|_\infty .
\]
For \(\varepsilon\in(0,1]\), define
\[
\tilde u_\varepsilon(t,y,p)
:=
e^{\lambda t}
\left(
u(t,y,p)
-\frac{\varepsilon}{T-t}
-\varepsilon^2 f(p)
\right),
\qquad
\tilde v(t,y,p):=e^{\lambda t}v(t,y,p).
\]
Since \(\sup_{\mathrm{cl}_y(\cD)}(u-v)>0\), for all \(\varepsilon>0\) sufficiently small,
\[
\sup_{\mathrm{cl}_y(\cD)}(\tilde u_\varepsilon-\tilde v)>0.
\]
The perturbation \(-\varepsilon/(T-t)\) forces
\[
\tilde u_\varepsilon(t,y,p)-\tilde v(t,y,p)\to-\infty
\qquad\text{as }t\uparrow T,
\]
and the perturbation \(-\varepsilon^2 f(p)\) forces
\[
\tilde u_\varepsilon(t,y,p)-\tilde v(t,y,p)\to-\infty
\qquad\text{as }p\to0^+\text{ or }p\to1^-.
\]
Together with the boundary condition on \(\cD_d\cup\cD_u\) and at $t=T$, this implies that the positive maximum of
\(\tilde u_\varepsilon-\tilde v\) is attained at some point
\[
(t_0,y_0,p_0)\in\cD .
\]

We fix $\lambda>0$. The function \(\tilde u_\varepsilon\) is a subsolution of the rescaled equation with Hamiltonian \(H^\lambda\):
\begin{equation}\label{eq:strict_sub}
-\partial_t\tilde u_\varepsilon+\lambda\tilde u_\varepsilon
+H^\lambda_*
\bigl(t,y,p,\nabla\tilde u_\varepsilon,\nabla^2\tilde u_\varepsilon\bigr)
\le
-\frac{\varepsilon e^{\lambda t}}{(T-t)^2}
+
K_H\varepsilon^2
\end{equation}
on \(\mathrm{cl}_y(\cD)\cap\{t<T\}\). The term \(K_H\varepsilon^2\) comes from the \(p\)-diffusion contribution of the barrier \(-\varepsilon^2 f(p)\), using \eqref{eq:KH_def} and
\[
|A^0(Z^0)-A^1(Z^1)|\le \bar a .
\]
Since
\[
\frac{e^{\lambda t}}{(T-t)^2}\ge \frac1{T^2},
\]
the right-hand side of \eqref{eq:strict_sub} is strictly negative for \(\varepsilon>0\) small enough. Moreover, \(\tilde v\) is a supersolution:
\begin{equation}\label{eq:v_super}\notag
-\partial_t\tilde v+\lambda\tilde v
+
(H^\lambda)^*
\bigl(t,y,p,\nabla\tilde v,\nabla^2\tilde v\bigr)
\ge0
\qquad\text{on }\cD.
\end{equation}

\medskip\noindent
\textbf{Step 2 (Doubling of variables).}
For \(n\ge1\), define
\begin{equation}\label{eq:phi_n}\notag
\Phi_n(t,\bar y,\bar p,y,p)
:=
\tilde u_\varepsilon(t,\bar y,\bar p)
-
\tilde v(t,y,p)
-
\frac n2(\bar y-y)^2
-
\frac{n^2}{2}(\bar p-p)^2
+
\varepsilon(\bar y-y_0)^2
\end{equation}
and
\[
M_n:=\sup \Phi_n.
\]
Let
\[
(t_n,\bar y_n,\bar p_n,y_n,p_n)
\]
be a maximizer of \(\Phi_n\) which exists by compactness of the domain.

\medskip\noindent
\textbf{Step 3 (Convergence and interiority).}
The standard doubling-variable argument yields
\begin{equation}\label{eq:pen_bound}
n(\bar y_n-y_n)^2\to0,
\qquad
n^2(\bar p_n-p_n)^2\to0.
\end{equation}
Moreover,
\[
M_n\to M_0,
\]
where
\[
M_0
:=
\sup_{(t,y,p)\in\mathrm{cl}_y(\cD)}
\left[
\tilde u_\varepsilon(t,y,p)-\tilde v(t,y,p)
+\varepsilon(y-y_0)^2
\right].
\]
Every limit point \((t_*,y_*,p_*)\) of both
\[
(t_n,\bar y_n,\bar p_n)
\qquad\text{and}\qquad
(t_n,y_n,p_n)
\]
satisfies
\[
\tilde u_\varepsilon(t_*,y_*,p_*)
-
\tilde v(t_*,y_*,p_*)
+
\varepsilon(y_*-y_0)^2
=
M_0>0.
\]
Since the \(y\)-domain is bounded, choosing \(\varepsilon\) small enough gives
\[
\tilde u_\varepsilon(t_*,y_*,p_*)
-
\tilde v(t_*,y_*,p_*)>0.
\]
The strict boundary behavior from Step~1 then implies
\[
(t_*,y_*,p_*)\in\cD.
\]
Hence, for all \(n\) sufficiently large,
\begin{equation}\label{eq:interior}\notag
(t_n,\bar y_n,\bar p_n)\in\cD,
\qquad
(t_n,y_n,p_n)\in\cD.
\end{equation}

\medskip\noindent
\textbf{Step 4 (Ishii's lemma).}
Define
\[
\psi(t,\bar y,\bar p,y,p)
:=
\frac n2(\bar y-y)^2
+
\frac{n^2}{2}(\bar p-p)^2
-
\varepsilon(\bar y-y_0)^2.
\]
We write the spatial Hessian in the grouped ordering
$(\bar y,\bar p,y,p).
$
In this ordering,
\[
M:=D^2\psi
=
\begin{pmatrix}
n-2\varepsilon & 0 & -n & 0\\
0 & n^2 & 0 & -n^2\\
-n & 0 & n & 0\\
0 & -n^2 & 0 & n^2
\end{pmatrix}.
\]
Equivalently, writing \(M\) in block form according to
$
(\bar y,\bar p)\quad\text{and}\quad(y,p),
$
we have
\[
M=
\begin{pmatrix}
M_{11}&M_{12}\\
M_{21}&M_{22}
\end{pmatrix},
\]
where
\[
M_{11}
=
\begin{pmatrix}
n-2\varepsilon&0\\
0&n^2
\end{pmatrix},
\qquad
M_{22}
=
\begin{pmatrix}
n&0\\
0&n^2
\end{pmatrix},
\]
and
\[
M_{12}=M_{21}
=
\begin{pmatrix}
-n&0\\
0&-n^2
\end{pmatrix}.
\]

For each \(n\), choose \(\eta_n:=n^{-3}\). Since both points are interior, parabolic Ishii lemma
\cite[Theorem~8.3]{CIL92} gives the existence of 
\(a_n=b_n\in\RR\) and \(X_n,Y_n\in\cS_2\) such that
\begin{align}
\bigl(a_n,(q_n^u,n^2(\bar p_n-p_n)),X_n\bigr)
&\in
\overline{\mathcal P}^{2,+}
\tilde u_\varepsilon(t_n,\bar y_n,\bar p_n),
\label{eq:jet_u}\\
\bigl(b_n,(q_n^v,n^2(\bar p_n-p_n)),Y_n\bigr)
&\in
\overline{\mathcal P}^{2,-}
\tilde v(t_n,y_n,p_n),
\label{eq:jet_v}
\end{align}
with
\begin{equation}\label{eq:ishii_matrix}
\begin{pmatrix}
X_n&0\\
0&-Y_n
\end{pmatrix}
\le
M+\eta_n M^2.
\end{equation}
Here \(a_n=b_n\) because the penalization is independent of time. The \(y\)-gradients are
\[
q_n^u
=
n(\bar y_n-y_n)-2\varepsilon(\bar y_n-y_0),
\qquad
q_n^v
=
n(\bar y_n-y_n).
\]
Therefore,
\begin{equation}\label{eq:grad_mismatch}
|q_n^u-q_n^v|
=
2\varepsilon|\bar y_n-y_0|
\le
C_d\varepsilon,
\end{equation}
where \(C_d\) is independent of \(n\) and \(\varepsilon\).

\medskip\noindent
\medskip\noindent
\textbf{Step 5.}
Fix $(Z^0,Z^1)\in\RR^2$. Recall from \eqref{eq:defSigma} that
\[
\Sigma(p,Z^0,Z^1)
=\bigl(Z^0-Z^1,\;p(1-p)(A^0(Z^0)-A^1(Z^1))\bigr)^\top,
\]
and consider the vector in the grouped ordering $(\bar y,\bar p,y,p)$,
\[
\xi:=\bigl(\Sigma(\bar p_n,Z^0,Z^1),\,\Sigma(p_n,Z^0,Z^1)\bigr)^\top\in\RR^4.
\]
Since $\xi\xi^\top\ge 0$, contracting \eqref{eq:ishii_matrix} with $\xi\xi^\top$ gives
\[
\mathrm{tr}\!\left(\xi\xi^\top\begin{pmatrix}X_n&0\\0&-Y_n\end{pmatrix}\right)
\le
\mathrm{tr}\bigl(\xi\xi^\top(M+\eta_n M^2)\bigr),
\]
meaning 
\[
\Sigma(\bar p_n,Z^0,Z^1)^\top X_n\,\Sigma(\bar p_n,Z^0,Z^1)
-\Sigma(p_n,Z^0,Z^1)^\top Y_n\,\Sigma(p_n,Z^0,Z^1)
\le
\xi^\top(M+\eta_n M^2)\xi.
\]

From the explicit expression of $M$,
\[
\xi^\top M\xi
=
-2\varepsilon(Z^0-Z^1)^2
+n^2\bigl(\bar p_n(1-\bar p_n)-p_n(1-p_n)\bigr)^2\bigl(A^0(Z^0)-A^1(Z^1)\bigr)^2,
\]
and
\[
\xi^\top M^2\xi
=
4\varepsilon^2(Z^0-Z^1)^2
+2n^4\bigl(\bar p_n(1-\bar p_n)-p_n(1-p_n)\bigr)^2\bigl(A^0(Z^0)-A^1(Z^1)\bigr)^2.
\]
Hence
\begin{align*}
\xi^\top(M+\eta_n M^2)\xi
&=(-2\varepsilon+4\eta_n\varepsilon^2)(Z^0-Z^1)^2\\
&\quad+(n^2+2\eta_n n^4)\bigl(\bar p_n(1-\bar p_n)-p_n(1-p_n)\bigr)^2\bigl(A^0(Z^0)-A^1(Z^1)\bigr)^2.
\end{align*}
Dividing by $2$,
\begin{align}\label{eq:trace_final}
&\tfrac12\,\Sigma(\bar p_n,Z^0,Z^1)^\top X_n\,\Sigma(\bar p_n,Z^0,Z^1)
-\tfrac12\,\Sigma(p_n,Z^0,Z^1)^\top Y_n\,\Sigma(p_n,Z^0,Z^1)\notag\\
&\quad\le
\varepsilon(Z^0-Z^1)^2(-1+2\eta_n\varepsilon)\notag\\
&\qquad+\tfrac12(n^2+2\eta_n n^4)\bigl(\bar p_n(1-\bar p_n)-p_n(1-p_n)\bigr)^2\bigl(A^0(Z^0)-A^1(Z^1)\bigr)^2.
\end{align}

The map $p\mapsto p(1-p)$ is $1$-Lipschitz on $[0,1]$, and
$|A^0(Z^0)-A^1(Z^1)|\le\bar a$ by Lemma~\ref{lem:verify_cd}(1), so
\[
\bigl(\bar p_n(1-\bar p_n)-p_n(1-p_n)\bigr)^2\bigl(A^0(Z^0)-A^1(Z^1)\bigr)^2
\le\bar a^{\,2}|\bar p_n-p_n|^2.
\]
With $\eta_n=n^{-3}$,
\[
\tfrac12(n^2+2\eta_n n^4)\bar a^{\,2}|\bar p_n-p_n|^2
=\tfrac12(n^2+2n)\bar a^{\,2}|\bar p_n-p_n|^2\longrightarrow 0
\]
by \eqref{eq:pen_bound}.

\medskip\noindent
\textbf{Step 6 (Application of the Hamiltonian comparison lemma).}
By \eqref{eq:trace_final}, for every $(Z^0,Z^1)\in\RR^2$,
\begin{align*}
&\tfrac12\,\Sigma(\bar p_n,Z^0,Z^1)^\top X_n\,\Sigma(\bar p_n,Z^0,Z^1)
-\tfrac12\,\Sigma(p_n,Z^0,Z^1)^\top Y_n\,\Sigma(p_n,Z^0,Z^1)\\
&\qquad\le
\varepsilon(Z^0-Z^1)^2(-1+2\eta_n\varepsilon)
+\tfrac12(n^2+2\eta_n n^4)\bar a^{\,2}|\bar p_n-p_n|^2.
\end{align*}
Since $\eta_n=n^{-3}$, for all $n$ sufficiently large,
\[
\varepsilon(-1+2\eta_n\varepsilon)\le -\frac{\varepsilon}{2}.
\]
Hence
\begin{align*}
&\tfrac12\,\Sigma(\bar p_n,Z^0,Z^1)^\top X_n\,\Sigma(\bar p_n,Z^0,Z^1)
-\tfrac12\,\Sigma(p_n,Z^0,Z^1)^\top Y_n\,\Sigma(p_n,Z^0,Z^1)
+\tfrac{\varepsilon}{2}(Z^0-Z^1)^2\\
&\qquad\le
\tfrac12(n^2+2\eta_n n^4)\bar a^{\,2}|\bar p_n-p_n|^2.
\end{align*}
Taking the supremum over $(Z^0,Z^1)\in\RR^2$, the trace hypothesis in Lemma~\ref{lem:Hcomp} is satisfied with coercivity parameter $\varepsilon/2$ and remainder
\[
r_n:=\tfrac12(n^2+2\eta_n n^4)\bar a^{\,2}|\bar p_n-p_n|^2.
\]
Moreover,
\[
r_n=\tfrac12(n^2+2n)\bar a^{\,2}|\bar p_n-p_n|^2\longrightarrow 0
\]
by \eqref{eq:pen_bound}.

Applying Lemma~\ref{lem:Hcomp}, after changing the constant $C$, gives
\begin{equation}\label{eq:H_diff_lemma}
H^\lambda(t_n,y_n,p_n,q_n^v,Y_n)
-H^\lambda(t_n,\bar y_n,\bar p_n,q_n^u,X_n)
\le r_n+\frac{C}{\varepsilon}\delta_n^2+C\delta_n,
\end{equation}
where
\[
\delta_n:=|q_n^v-q_n^u|+\bigl(|\bar y_n-y_n|+|\bar p_n-p_n|\bigr)(1+|q_n^u|).
\]
By \eqref{eq:grad_mismatch},
\[
|q_n^v-q_n^u|\le C_d\varepsilon,
\qquad
|q_n^u|\le n|\bar y_n-y_n|+C_d\varepsilon.
\]
Therefore,
\[
|\bar y_n-y_n|(1+|q_n^u|)
\le|\bar y_n-y_n|+n|\bar y_n-y_n|^2+C_d\varepsilon|\bar y_n-y_n|\longrightarrow 0,
\]
and
\begin{align*}
|\bar p_n-p_n|(1+|q_n^u|)
&\le|\bar p_n-p_n|+n|\bar p_n-p_n|\,|\bar y_n-y_n|+C_d\varepsilon|\bar p_n-p_n|\\
&\le|\bar p_n-p_n|+\tfrac12 n|\bar y_n-y_n|^2+\tfrac12 n|\bar p_n-p_n|^2+C_d\varepsilon|\bar p_n-p_n|\longrightarrow 0.
\end{align*}
Indeed,
\[
n|\bar y_n-y_n|^2\to 0,\qquad
n|\bar p_n-p_n|^2=\tfrac{1}{n}\,n^2|\bar p_n-p_n|^2\to 0
\]
by \eqref{eq:pen_bound}. Hence
\[
\limsup_{n\to\infty}\delta_n\le C_d\varepsilon,
\]
and consequently
\[
\limsup_{n\to\infty}\left(\frac{C}{\varepsilon}\delta_n^2+C\delta_n\right)\le C\varepsilon,
\]
with $C$ independent of $\varepsilon$. Combining this estimate with $r_n\to 0$ in \eqref{eq:H_diff_lemma},
\begin{equation}\label{eq:H_limsup}\notag
\limsup_{n\to\infty}\Big[
H^\lambda(t_n,y_n,p_n,q_n^v,Y_n)
-H^\lambda(t_n,\bar y_n,\bar p_n,q_n^u,X_n)
\Big]\le C\varepsilon.
\end{equation}
\medskip\noindent
\textbf{Step 7 .}
The jets \eqref{eq:jet_u}--\eqref{eq:jet_v} only give the viscosity
inequalities at the semicontinuous envelopes:
\begin{align}
-a_n+\lambda\tilde u_\varepsilon(t_n,\bar y_n,\bar p_n)
+(H^\lambda)_*\!\bigl(t_n,\bar y_n,\bar p_n,q_n^u,X_n\bigr)
&\le -\frac{\varepsilon\,e^{\lambda t_n}}{(T-t_n)^2}+K_H\varepsilon^2,
\notag\\[2pt]
-a_n+\lambda\tilde v(t_n,y_n,p_n)
+(H^\lambda)^{*}\!\bigl(t_n,y_n,p_n,q_n^v,Y_n\bigr)
&\ge 0.
\label{eq:sup_env}
\end{align}
We show that, at the maximizer, the envelopes coincide with $H^\lambda$.

\smallskip
\emph{(i) Upper semicontinuity of $H^\lambda$.}
The second-order argument of the Hamiltonian is a $2\times 2$ symmetric
matrix
\[
X=\begin{pmatrix} X_{yy} & X_{yp}\\[2pt] X_{yp} & X_{pp}\end{pmatrix}\in\cS_2,
\]
acting on the $(y,p)$-Hessian. With $\Sigma(p,z^0,z^1)$ as in
\eqref{eq:defSigma},
\begin{align*}
\tfrac12\,\mathrm{tr}(\Sigma\Sigma^\top X)
&=\tfrac12\,X_{yy}\,(z^0-z^1)^2
+X_{yp}\,(z^0-z^1)\,p(1-p)\bigl(A^0(z^0)-A^1(z^1)\bigr)\\[2pt]
&\quad+\tfrac12\,X_{pp}\,p^2(1-p)^2\bigl(A^0(z^0)-A^1(z^1)\bigr)^2,
\end{align*}
so the coefficient of $(z^0-z^1)^2$ in $-L^{z^0,z^1}(t,y,p,q,X)$ equals
$-\tfrac12 X_{yy}$, while the remaining terms grow at most linearly in
$(z^0,z^1)$ (using boundedness of $A^0,A^1$ and the Lipschitz property
of $H^0,H^1$ from Lemma~\ref{lem:verify_cd}(1)). In particular,
\[
X_{yy}>0\ \Longrightarrow\ -L^{z^0,z^1}(t,y,p,q,X)
\xrightarrow[|z^0-z^1|\to\infty]{}-\infty
\ \Longrightarrow\ H^\lambda(t,y,p,q,X)=-\infty.
\]

For every fixed $(z^0,z^1)\in\RR^2$, the map
$(t,y,p,q,X)\mapsto -L^{z^0,z^1}(t,y,p,q,X)$ is jointly continuous, so
$H^\lambda$, as the infimum of a family of continuous (hence upper
semicontinuous) functions, is upper semicontinuous on its domain. Hence
\begin{equation}\label{eq:Hstar_eq}
(H^\lambda)^{*}=H^\lambda \quad\text{everywhere.}
\end{equation}

\smallskip
\emph{(ii) Sign of $(Y_n)_{yy}$.}
By \eqref{eq:Hstar_eq} and \eqref{eq:sup_env},
\[
H^\lambda(t_n,y_n,p_n,q_n^v,Y_n)
=(H^\lambda)^{*}(t_n,y_n,p_n,q_n^v,Y_n)
\ge a_n-\lambda\tilde v(t_n,y_n,p_n)>-\infty.
\]
By (i), this rules out $(Y_n)_{yy}>0$, so
\begin{equation}\label{eq:Yyy}
(Y_n)_{yy}\le 0.
\end{equation}

\smallskip
\emph{(iii) Ishii's matrix inequality forces $(X_n)_{yy}\le -\varepsilon$.}
Apply \eqref{eq:ishii_matrix} to the rank--one test vector
$e:=(1,0,1,0)^\top\in\RR^4$. Direct computation in the grouped ordering
$(\bar y,\bar p,y,p)$ gives
\[
Me=(-2\varepsilon,0,0,0)^\top,\qquad
e^\top Me=-2\varepsilon,\qquad e^\top M^2 e=4\varepsilon^{\,2}.
\]
The left-hand side of \eqref{eq:ishii_matrix} contracted with $e$ equals
$(X_n)_{yy}-(Y_n)_{yy}$, hence
\[
(X_n)_{yy}-(Y_n)_{yy}\;\le\;-2\varepsilon+4\eta_n\varepsilon^{\,2}.
\]
Since $\eta_n=n^{-3}\to 0$, for all $n$ large enough
$4\eta_n\varepsilon^2\le \varepsilon$, and combining with \eqref{eq:Yyy},
\begin{equation}\label{eq:Xyy}
(X_n)_{yy}\;\le\;(Y_n)_{yy}-\varepsilon\;\le\;-\varepsilon\;<\;0.
\end{equation}

\smallskip
\emph{(iv) $(H^\lambda)_*=H^\lambda$ at $X_n$.}
On the open half-space $\{X_{yy}<0\}$, the integrand $-L^{z^0,z^1}$ is
quadratically coercive in $(z^0-z^1)$ with coefficient
$-X_{yy}/2\ge\varepsilon/2$ on a neighbourhood of any base point, while
the remaining terms are linear in $(z^0,z^1)$. A standard
Berge-type maximum theorem argument then yields that the infimum
defining $H^\lambda$ is attained on a compact set varying continuously
with the parameters, and that $H^\lambda$ is continuous on
$\{X_{yy}<0\}$. By \eqref{eq:Xyy} and non positivity of $(Y_n)_{yy}$,
\begin{equation}\label{eq:Hstarstar_eq}
(H^\lambda)_*\bigl(t_n,\bar y_n,\bar p_n,q_n^u,X_n\bigr)
\;=\;H^\lambda\bigl(t_n,\bar y_n,\bar p_n,q_n^u,X_n\bigr).
\end{equation}

\medskip\noindent
\textbf{Step 8.}
The subsolution and supersolution inequalities and \eqref{eq:Hstar_eq}, \eqref{eq:Hstarstar_eq} give
\[
\lambda\bigl(\tilde u_\varepsilon(t_n,\bar y_n,\bar p_n)
  -\tilde v(t_n,y_n,p_n)\bigr)
\;\le\;H^\lambda(t_n,y_n,p_n,q_n^v,Y_n)
  -H^\lambda(t_n,\bar y_n,\bar p_n,q_n^u,X_n)+K_H\,\varepsilon^2.
\]
Since $M_n\le\tilde u_\varepsilon(t_n,\bar y_n,\bar p_n)
-\tilde v(t_n,y_n,p_n)$ and $M_n\to M_0\ge\sup(\tilde u_\varepsilon
-\tilde v)$, taking $\limsup$ using Step~6:
\[
\lambda\sup_{\mathrm{cl}_y(\cD)}
(\tilde u_\varepsilon-\tilde v)
\;\le\;\lambda\,M_0\;\le\;C'\varepsilon+K_H\,\varepsilon^2.
\]
Since $\sup(\tilde u_\varepsilon-\tilde v)
=\sup_{\mathrm{cl}_y(\cD)}\bigl[e^{\lambda t}(u-v)
-e^{\lambda t}\varepsilon/(T-t)-e^{\lambda t}\varepsilon^2 f(p)\bigr]
\;\to\;\sup_{\mathrm{cl}_y(\cD)}e^{\lambda t}(u-v)
\;\ge\;\sup(u-v)>0$
as $\varepsilon\to0$, we obtain
\[
\lambda\sup(u-v)\;\le\;\lim_{\varepsilon\to0}\bigl(C'\varepsilon+K_H\,\varepsilon^2\bigr)=0,
\]
contradicting $\sup(u-v)>0$.
We conclude $u\le v$ on $\mathrm{cl}_y(\cD)$.
\qed
\end{proof}

\bibliographystyle{apalike}
\bibliography{main}
\end{document}